\pdfoutput=1
\PassOptionsToPackage{letterpaper}{geometry}
\documentclass{article}
\usepackage[T1]{fontenc}
\usepackage[utf8x]{inputenc}
\usepackage[text={6in,8in},centering]{geometry}
\usepackage{amsmath,amssymb,amsfonts,amsthm,mathtools,array,booktabs,caption,microtype,graphicx,mathrsfs}
\usepackage[vcentermath]{youngtab}
\usepackage{tikz}
\usetikzlibrary{calc,trees,arrows}
\tikzset{level 1/.style={level distance=1.5cm, sibling distance=3.5cm}}
\tikzset{level 2/.style={level distance=1.5cm, sibling distance=2cm}}
\usetikzlibrary{decorations.markings}
\usetikzlibrary{decorations.pathreplacing}
\usepackage{hyperref,bookmark}
\hypersetup{linktoc = all}
\hypersetup{pdfborderstyle={/S/U/W 0.5}}
\hypersetup{linkbordercolor = gray}
\hypersetup{bookmarksnumbered}

\newtheorem{lemma}{Lemma}[section]
\newcommand\bee{\begin{equation}}
\newcommand\eeq{\end{equation}}
\renewcommand\appendixautorefname{appendix}

\newcommand{\Sectionref}[1]{\mbox{\def\sectionautorefname{Section}\vphantom{y}\autoref{#1}}}
\newcommand{\Appendixref}[1]{\mbox{\def\appendixautorefname{Appendix}\vphantom{y}\autoref{#1}}}
\newcommand{\ZZ}{{\mathbb{Z}}}

\DeclarePairedDelimiter{\abs}{\lvert}{\rvert}
\DeclareMathOperator{\Tr}{Tr}
\newcommand{\sym}{\mathrm{S}}
\newcommand{\AdjG}{\chi_{\raisebox{-1.5pt}{\scriptsize{\rm Adj}\,{\it G}}}}
\makeatletter
\newcommand{\yngantisym}{{\let\@nomath\@gobble\let\typeout\@gobble\tiny\yng(1,1)}}
\makeatother

\title{Marginal Deformations of  3d  ${\cal N}=4$ Linear Quiver Theories} 
\author{}
\date{}
\hypersetup
  {
    pdftitle = {Marginal Deformations of 3d N=4 Linear Quiver Theories},
    pdfauthor = {Constantin Bachas, Ioannis Lavdas and Bruno Le Floch}
  }

\begin{document}

\maketitle
\begin{center}
Constantin Bachas$^a$, Ioannis Lavdas$^a$ and  Bruno Le Floch$^b$\\
\medskip\medskip\medskip

$^a$%
\textit{Laboratoire de Physique de l'\'Ecole Normale Sup\'erieure,
  \'ENS, Universit\'e PSL, CNRS,\\
  Sorbonne Universit\'e, Universit\'e Paris-Diderot, Sorbonne Paris Cit\'e, Paris, France}

\medskip

$^b$%
\textit{Institut Philippe Meyer, Département de Physique,\\ \'Ecole Normale Sup\'erieure, Université PSL, Paris, France}

\vskip 2cm

\begin{abstract}

We study  superconformal 
deformations of the $T_\rho^{\hat\rho}[SU(N)]$  theories of Gaiotto-Hanany-Witten,
paying special attention to  mixed-branch operators  with  both electrically- and magnetically-charged fields. 
We explain why all  marginal  ${\cal N}=2$ operators of an ${\cal N}=4$ CFT$_3$ 
can be extracted unambiguously  from 
  the  superconformal index.  Computing the index at the appropriate order we show 
   that  the mixed moduli in  $T_\rho^{\hat\rho}[SU(N)]$  theories
  are double-string operators   transforming   in
the  (Adjoint, Adjoint)  representation of the electric and magnetic flavour groups,  up to  some 
overcounting  
for   quivers with abelian gauge nodes.  
 We  comment  on  the   holographic interpretation of the  results,  arguing in particular that
  gauged supergravities can   capture the entire moduli space if,  in addition to the 
  (classical) parameters of the background solution,  one takes also into account  the  (quantization) moduli   of 
  boundary conditions.

\end{abstract}

\end{center}

\numberwithin{equation}{section}

\newpage
{\let\oldpdfendlink\pdfendlink
\renewcommand\pdfendlink{\oldpdfendlink\vphantom{Iy}} % make depth of links more uniform in table of contents
\tableofcontents
}
\newpage

\section{Introduction}

   Superconformal field theories  (SCFT)  often  have continuous deformations  preserving
   some  superconformal
   symmetry. The space of such deformations is  a Riemannian  manifold (the `superconformal manifold')
   which
   coincides with   the moduli space of  supersymmetric 
   Anti-de Sitter (AdS) vacua when the SCFT has  a holographic dual.  
   Mapping out such   moduli spaces is of    interest both for field theory and  for  the study
    of the string-theory landscape.

       In this paper  we will be interested in  superconformal manifolds  in the vicinity  of the    `good'    
 theories    $T_\rho^{\hat\rho}\left[SU(N)\right]$  whose existence was   conjectured  by  Gaiotto and Witten \cite{Gaiotto:2008ak}.  These are three-dimensional ${\cal N}=4$ SCFTs   arising  as
  infrared fixed points of   a certain class of  quiver gauge theories introduced by Hanany and Witten \cite{Hanany:1996ie}.
    Their   holographic  duals  are   four-dimensional Anti-de Sitter 
   (AdS$_4$) solutions   of  type-IIB string theory \cite{Assel:2011xz,Assel:2012cj,Cottrell:2016nsu,Lozano:2016wrs}. 
                     Our main motivation in  this  work  was  to  extract   features  of  these   moduli   spaces
         not readily  accessible  from  the  gravity side.   
          We build on  the analysis  of   ref.\,\cite{Bachas:2017wva} which 
          we   complete and amend  in significant  ways. 
\smallskip
 
Superconformal deformations  of a  $d$-dimensional theory $T_\star$  are  generated  by the  set of marginal
operators $\{\mathcal{O}_{i}\}$ that preserve some or  all  of its 
supersymmetries.\footnote{One  exception to this general rule 
  is  the  gauging of  a global symmetry with vanishing $\beta$ function  in  four dimensions.} 
 The existence of such operators is  constrained by the analysis  of  
  representations of the superconformal  algebra \cite{Cordova:2016xhm}.  
  In particular,   unitary SCFTs  have no 
   moduli  in $d=5$ or $6$  dimensions,  whereas  in  the case $d=3$ of  interest  here 
   moduli   preserve at most  ${\cal N}=  2$ supersymmetries. 
  Those  preserving only ${\cal N}=1$ belong to long (`D-term')  
  multiplets whose  dimension is not protected against quantum corrections. 
  The existence of such  ${\cal N}=1$  moduli (and of  non-supersymmetric ones)  is  fine-tuned
  and thus accidental.  For this reason we focus here on  the ${\cal N}=  2$ moduli. 
  
    \smallskip
  
   The general  local structure of ${\cal N}=  2$ superconformal manifolds in three dimensions 
(and  of   the  closely-related  case ${\cal N}= 1$ in  $d=4$)
has been described in   \cite{Leigh:1995ep,Kol:2002zt,Aharony:2002hx,Benvenuti:2005wi,Green:2010da}. 
These  manifolds are  K\"ahler quotients 
of the
space $\{\lambda^i\}$
of  marginal supersymmetry-preserving couplings modded out by the complexified global (flavor)
symmetry group $G_{\rm global}$,  
\bee\label{quotient}
{\cal M}_{\rm SC} \,\simeq\,  \{\lambda^i\}/G_{\rm global}^{\mathbb{C}} \,\simeq\,
 \{\lambda^i \vert D^a=0\}/G_{\rm global}  \ . 
\eeq
The meaning of  this   is as follows:  marginal scalar operators $\mathcal{O}_{i}$
 fail to be {\it exactly}  marginal if and only if they combine 
with  conserved-current multiplets of $G_{\rm global}$ to form  long (unprotected)  current multiplets. 
 Requesting    this   not to happen 
     imposes the  moment-map  conditions 
   \bee\label{quotient1}
D^a \,=\,  \lambda^i T^a_{i\bar j} \bar\lambda^{\bar j} +  O(\lambda^3) \,=\, 0\ , 
   \eeq 
where    $T^a$  are the  generators of $G_{\rm global}$ 
in the representation of the  couplings.  The second quotient by   $G_{\rm global}$  in \eqref{quotient} 
identifies   deformations that belong to  the same   orbit. 
The complex dimension of the moduli space    is therefore equal to the difference
\bee\label{dimMSCgeneral}
\dim_{\mathbb{C}} {\cal M}_{\rm SC} = \# \{   \mathcal{O}_{i}\}  - \dim G_{\rm global}\ . 
\eeq
In the dual  
 gauged supergravity  (when  one exists)  the fields dual to single-trace operators $\mathcal{O}_{i}$ are    
 ${\cal N}= 2$ hypermultiplets,  and
         $D^a=0$ are    D-term conditions \cite{deAlwis:2013jaa}.
\smallskip

  The global flavour symmetry of the $T_\rho^{\hat\rho}\left[SU(N)\right]$ theories, viewed 
   as  ${\cal N}=2$ SCFTs,   is a product
      \bee\label{Ggl}
      G_{\rm global} = G \times \hat G \times U(1)  \ , 
       \eeq 
      where  $G$ and $\hat G$ are the 
     flavour groups  of the electric and  magnetic   theories that are   related by mirror symmetry, 
      and $U(1) $ is the
       subgroup of the  $SO(4)_R$   symmetry  
      which commutes with the unbroken ${\cal N}=2$. 
      As for any 3d ${\cal N}=2$ theory, the local moduli space is the K\"ahler quotient~\eqref{quotient}.
      To determine this moduli space we must  thus   list   all  marginal   
      supersymmetric operators   
      and   the $G_{\rm global}$-representation(s)  in which they  transform.
      The ${\cal N}=4$ supersymmetry helps to identify these unambiguously.
    Many of these marginal deformations are standard superpotential deformations 
involving hypermultiplets  of  either the electric theory or its magnetic mirror.
Some marginal operators  involve,  however,   both kinds of hypermultiplets and do not
admit a local Lagrangian description. We refer to  such  deformations as  `mixed'. 
They are specific to three dimensions, and will be the focus of our paper. 

       Marginal  ${\cal N}=2$ 
       deformations  of  ${\cal N}=4$ theories belong to three kinds of superconformal
       multiplets~\cite{Bachas:2017wva}. The  Higgs-  and  Coulomb-branch  superpotentials  belong,   respectively,  
        to     $(2,0)$ and $(0,2)$ 
         representations of  
 $SO(3)_{H}\times SO(3)_{C}$,
where $SO(3)_{H}\times SO(3)_{C} \simeq  SO(4)_R$ is the   ${\cal N}=4$  
$R$-symmetry.\footnote{$SO(3)_{H}$ and $SO(3)_{C}$ act on the chiral rings of the pure Higgs 
and pure Coulomb branches of the   theory, whence their names. They are exchanged by mirror symmetry.}
The mixed marginal operators on the other hand transform in the $(J^H, J^C)=(1,1)$ representation. 
  In the holographic dual supergravity  the  $(2,0)$ and $(0,2)$
multiplets describe massive ${\cal N}=4$ vector bosons, while the $(1,1)$ multiplets  contain
also spin-${3\over 2}$ fields.   These latter are also special for another reason: they are Stueckelberg
fields  capable of  rendering   the ${\cal N}=4$ graviton multiplet  massive \cite{Bachas:2017rch, Bachas:2019rfq}. 
In representation theory  they are the unique short multiplets that can combine with the
conserved energy-momentum tensor into  a   
   long multiplet.   This monogamous relation will allow us to identify them unambiguously in  
   the superconformal index. 

          More generally, one cannot   distinguish  in  the superconformal index the contribution of the 
    ${\cal N}=2$ chiral ring,  which contains scalar operators with arbitrary $(J^H, J^C)$,  
   from  contributions of other   short  multiplets.   Two exceptions to this   rule
    are the pure Higgs- and 
   pure Coulomb-branch chiral rings whose  $R$-symmetry quantum numbers 
  are   $(J^H, 0)$ and $(0, J^C)$. The corresponding multiplets
   are absolutely protected, i.e.\ they
    can never   recombine to form  long representations of  the
    ${\cal N}=4$ superconformal algebra    \cite{Cordova:2016emh}. 
    These  two  subrings of the chiral ring can thus be unambiguously identified. Their generating functions
    (known  as  the Higgs-branch  and Coulomb-branch Hilbert series
    \cite{Hanany:2011db,Cremonesi:2013lqa,Cremonesi:2014uva,Cremonesi:2014kwa,Cremonesi:2017jrk})  are indeed  simple limits of the
    superconformal index \cite{Razamat:2014pta}. 
   Arbitrary  elements of the chiral ring, on the other hand,  are out of reach of  presently-available 
   techniques.\footnote{Though there do exist some  interesting suggestions   \cite{Cremonesi:2016nbo,Carta:2016fjb} on which we will 
    comment at the end of this paper.}
    Fortunately this will not be an obstacle for the marginal $(1,1)$ operators  of interest here.  
\smallskip

       The  result of our calculation has no  surprises. As we will show, the  mixed marginal  operators
       transform in the  $({\rm Adj}, {\rm Adj}, 0)$  representation of the global symmetry \eqref{Ggl}, 
       up to some overcounting  when (and only when)   the quivers of 
       $T_\rho^{\hat\rho}\left[SU(N)\right]$ have  abelian gauge  nodes.\footnote{However, mixed marginal operators of more general 3d ${\cal N}=4$ theories may transform in a representation larger than $({\rm Adj}, {\rm Adj}, 0)$.  We give an example in \autoref{exs}.}
       More generally, the set
       of   all marginal ${\cal N}=2$ operators is of the form
       \bee
            \sym^2 ( {\rm Adj}\, G  +  {\rm Adj}\, \hat G)  \, + \, [\text{length-4 strings}]  \ - \ {\rm redundant} \ ,  
       \eeq
       where $\sym^2$ is the symmetrized square  of representations, and  the  `length-4 string'  operators  are   
       quartic superpotentials made out of the  hypermultiplets of the electric or the magnetic theory  only.
       All redundancies  arise  due to symmetrization and electric or magnetic 
       $F$-term conditions. Calculating them
        is the main  technical result of our  paper. On the way  we will  find  also some  
       new checks of  3$d$  mirror symmetry. 
       
         %\smallskip
         
           Our  calculation settles one  issue  about the dual AdS moduli that was
            left open in ref.\,\cite{Bachas:2017wva}. As is standard in holography, the  global 
            symmetries $G$ and $\hat G$ of the  CFT are realized as   gauge  symmetries on the gravity side. 
            The corresponding ${\cal N}=4$ vector bosons live on 
             stacks of magnetized D5-branes and NS5-branes which wrap two different 2-spheres
             (S$^2_H$ and S$^2_C$)  in the ten-dimensional
             spacetime~\cite{Assel:2011xz}. The $R$-symmetry spins $J^H$ and $J^C$
             are the   angular momenta    on these  spheres.   
              As was explained in~\cite{Bachas:2017wva},  the   Higgs-branch superconformal moduli correspond to
              open-string states on the D5-branes: 
               either    non-excited  single  strings with $J^H=2$,   or bound states of two  $J^H=1$ strings.  
              The Coulomb branch   superconformal moduli correspond likewise to open D-string states on   NS5-branes.
              For  mixed moduli ref.\,\cite{Bachas:2017wva} suggested two possibilities:  either 
              bound states of a $J^H=1$
              open string on the   D5-branes with a $J^C=1$   D-string  from  the  NS5 branes,  or single closed-string
              states that are scalar partners of massive gravitini. Our results  
                       rule out the  second possibility 
                      for  the backgrounds that are  dual to   linear 
                      quivers.\footnote{%
                      In the interacting theory   single- and multi-string states  with the same charges  
                       mix and cannot be distinguished. The above
                      statement  should be understood in the sense of cohomology: 
                        in  linear-quiver theories  all  $(1,1)$ elements 
                       of the  $\Delta=2$ chiral ring are accounted for by 2-string states.}
                         
                              It was also noted in ref.\,\cite{Bachas:2017wva} that although gauged ${\cal N}=4$  
                               supergravity can in principle
                              account for the $(2,0)$ and $(0,2)$ moduli that are scalar partners of spontaneously-broken
                              gauge bosons,  it has no massive spin-${3\over 2}$ multiplets to account for single-particle
                              $(1,1)$ moduli. But if  {all}  (1,1)  moduli are  2-particle states, they can  be
                              in principle   accounted for by modifying 
                               the AdS$_4$  boundary conditions 
                              \cite{Witten:2001ua,Berkooz:2002ug}.  The dismissal  in ref.\,\cite{Bachas:2017wva}
                              of gauged supergravity,   as
                              not capturing   the entire moduli space,   was thus   premature.   
                          We stress  however that changing the   boundary conditions does  not affect the classical 
                          AdS solution  but only  the fluctuations around it.  Put differently these   moduli 
                          show up only upon  quantization.   The 
                           analysis of   ${\cal N}=2$ AdS$_4$  moduli spaces in 
                          gauged supergravity~\cite{deAlwis:2013jaa} must be  revisited  in order 
                           to incorporate  such  `quantization moduli.'

                    \smallskip
                                                
                      This paper is organized as follows: \Sectionref{sec:2}   reviews some generalities
                      about   good $T_\rho^{\hat{\rho}}[{  SU(N)}]$ theories,   and exhibits  their  superconformal index
                      written as a multiple   integral and sum   over Coulomb-branch moduli and monopole fluxes. 
                      Our aim is  to recast  this expression into   a sum of superconformal
                      characters with   fugacities restricted as pertaining to the index. 
                                                                  These restricted  
                                                                 characters and the  linear relations that they obey are
                                                                 derived  in \autoref{sec3}.  We  
                                                                   also  explain in this section  how   the   
                                                                 ambiguities  inherent in the decomposition of the  index  as a sum
                                                                 over representations can be resolved  for the problem at hand. 
                        
                            \Sectionref{sec:4}  contains our  main calculation. We first expand the determinants  so as 
                           to only keep  contributions from operators with scaling dimension 
                          $\Delta \leq 2$, and then  perform  explicitly the
                              integrals  and  sums.  
                          The result is re-expressed as a sum of
                          characters of $OSp(4\vert 4)\times G\times \hat G$ in  \autoref{sec:5}. 
                           We  identify the superconformal moduli,   
                           comment on   their holographic interpretation (noting the role of a stringy exclusion
                           principle)
                           and conclude. 
                          Some   technical material is relegated to  appendices. \Appendixref{appB}
                          sketches
                           the derivation  of the superconformal index as a localized path  integral over the Coulomb branch. 
                           This is  standard material included   for the reader's convenience.  In 
                          \autoref{sec:A}  we  prove  a combinatorial lemma    needed    in 
                          the  main calculation.  Lastly a closed-form expression for
                          the index of  
                          $T[SU(2)]$,  which is sQED$_3$ with two `selectrons', is derived in  
                           \autoref{app:C}. This  
                          renders  manifest   a general property  (which we do not use in this  paper), namely
                            the  factorization of the index  in  holomorphic blocks  \cite{Pasquetti:2011fj,Dimofte:2011py,Beem:2012mb}.

\emph{Note added:} Shortly before ours, the paper~\cite{Okazaki:2019ony} was posted to the arXiv.  It checks mirror symmetry by comparing the index of mirror pairs, including many examples of coupled 4d-3d systems.  The papers only overlap marginally.

%%%%%%%%%%%%%%%

\section{Superconformal index of \texorpdfstring{$T_\rho^{\hat{\rho}}[{SU(N)}]$}{T[SU(N)]}}\label{sec:2}

\subsection{Generalities}
We   consider  the 3d  ${\cal N}=4$ gauge theories  \cite{Hanany:1996ie} based on the linear quivers 
  of \autoref{fig:linquiv}.  
 Circle nodes  in these quivers stand for  unitary gauge groups  $U(N_i)$, squares 
designate  fundamental hypermultiplets 
and  horizontal links stand for
bifundamental hypermultiplets.  One can   generalize  to theories with
orthogonal and symplectic gauge groups  and to quivers with non-trivial topology, but we
 will not consider  such  complications here. We are interested in the infrared limit  of `good
 theories'  \cite{Gaiotto:2008ak}  for which  $N_{j-1}+N_{j+1}+M_j\geq 2N_j\ \forall j$. 
These conditions ensure 
  that at a
 generic point of the Higgs branch
  the gauge symmetry is   completely broken. 
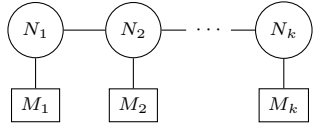
\begin{figure}\centering
  \begin{tikzpicture}
    \node(N1) [circle,draw] at (-.3,0) {\scriptsize $N_1$};
    \node(N2) [circle,draw] at (1,0) {\scriptsize $N_2$};
    \node(Ndots) at (2,0) {\scriptsize ${\cdots}$};
    \node(Nk) [circle,draw] at (3,0) {\scriptsize $N_k$};
    \node(M1) [rectangle,draw] at (-.3,-1) {\scriptsize $M_1$};
    \node(M2) [rectangle,draw] at (1,-1) {\scriptsize $M_2$};
    \node(Mk) [rectangle,draw] at (3,-1) {\scriptsize $M_k$};
    \draw (N1) -- (N2) -- (Ndots) -- (Nk);
    \draw (M1) -- (N1);
    \draw (M2) -- (N2);
    \draw (Mk) -- (Nk);
  \end{tikzpicture}
  \caption{\label{fig:linquiv}Linear quiver  with gauge group $U(N_1)\times\dots\times U(N_k)$.}
\end{figure}

   The   theories are defined in the ultraviolet (UV) by the standard
  ${\cal N}=4$ Yang-Mills plus matter 3d  action. All  masses and Fayet-Iliopoulos terms are set to zero
  and there are no Chern-Simons terms. We choose the vacuum  at the origin of both the Coulomb and Higgs branches,
  where all scalar expectation values vanish. 
Thus the only continuous parameters  are the dimensionful gauge couplings  $g_i$,
 which flow to infinity in the   infrared. 
 
   Every good  linear quiver has a mirror which is also a good linear quiver  and whose
   discrete data we denote  by  hats, 
    $\{\hat N_{\hat\jmath}, \hat M_{\hat\jmath}, \hat k \}$.
    A useful parametrization of both quivers  is in terms of an ordered
    pair of partitions, $(\rho, \hat\rho)$ with $\rho^T > \hat\rho$, see  \autoref{sec:A}.
               The SCFT has global  (electric and magnetic) flavour symmetries
    \bee
  G  \times \hat G  =   \Bigl(\prod_j U(M_j)\Bigr)/U(1)\, \times\, \Bigl(\prod_{\hat\jmath} 
   U(\hat M_{\hat\jmath})\Bigr)/U(1)\ ,
    \eeq
    with rank$\,G = \hat k$ and rank$\,\hat G =k$.   In  the string-theory embedding the 
    flavour symmetries are realized on  $(\hat k +1)$   D5-branes and $(k+1)$  NS5-branes
     \cite{Hanany:1996ie}.  
    The  symmetry  $G$
    is manifest in the microscopic  Lagragian of the electric theory,   
    as  is the Cartan subalgebra of $\hat G$ which is the
    topological symmetry whose  conserved currents are  the dual field strengths  $\Tr{\star F_{(j)}}$. 
    The non-abelian extension of $\hat G$ is realized in the infrared by monopole operators \cite{Borokhov:2002ib,Borokhov:2003yu}.

    \smallskip
    In addition to $G  \times \hat G$  the infrared fixed-point theory has global superconformal symmetry. 
         The  ${\cal N}=4$  superconformal group  in three dimensions is
 $OSp(4\vert 4)$. It  has
  eight real Poincar\'e 
  supercharges transforming in the  $({1\over 2}, {1\over 2}, {1\over 2})$ 
    representation of  $SO(1,2)\times SO(3)_{H}\times SO(3)_{C} $.
 The two-component 3d Lorentzian spinors  can be
  chosen real.  The marginal deformations studied in this paper  leave unbroken an
  ${\cal N}=2$ superconformal symmetry   $ {OSp}(2\vert 4) \subset 
   OSp(4\vert 4)$. 
      This is  generated by  two out of the four  real  $SO(1,2)$
   spinors, so modulo   $SO(4)_R$ rotations the embedding   is unique. 
   Let  $ Q^{(\pm\pm )} $
 be a complex basis for  the four  Poincar\'e supercharges, 
 where the   superscripts
are the eigenvalues of  the diagonal R-symmetry generators  $J_3^H$ and $ J_3^C$. 
  Without loss of generality we can choose the two unbroken  supercharges
  to be the complex pair $Q^{(++)}$ and $Q^{(--)}$, so that the  ${\cal N}=2$
$R$-symmetry is generated by $J_3^H+  J_3^C$  and 
the extra commuting $U(1)$   by $J_3^H - J_3^C$\,. We  use this same basis
in   the definition of the superconformal  index. 
 
         %%%%%        
  
  \subsection{Integral expression for the index}  
        
    There is  a large  literature on the ${\cal N}=2$ superconformal index in three dimensions,
    for a partial list of references see
     \cite{Bhattacharya:2008zy,Kim:2009wb,Imamura:2011su,Krattenthaler:2011da,Kapustin:2011jm,Benini:2013yva,Willett:2016adv}.   The index is defined in terms of the cohomology of the supercharge~$Q^{(++)}_-$.  It is a weighted sum over  
     local operators of the SCFT, or equivalently  over all quantum  states on the two-sphere, 
\begin{equation}\label{sci}
\mathcal{Z}_{S^{2}\times S^1} =\Tr_{\mathcal{H}_{S^{2}}}(-1)^{F}q^{\frac{1}{2}(\Delta + J_3)}t^{J_3^H-J_3^C}e^{-\beta(\Delta
-J_3 -J_3^H -J_3^C)} \ . 
\end{equation}
In this formula $F$ is the fermion number of the state, $J_3$ the third component of the spin, 
$\Delta$ the  energy, and $q, t, e^{-\beta}$ are fugacities. 
Only states for which $\Delta = 
 J_3 + J_3^H+ J_3^C$ contribute to the index, which is therefore independent of the fugacity $\beta$. 
 
 The non-abelian $R$-symmetry guarantees
(for good theories) 
that the $U(1)_R$ of the ${\cal N}=2$ subalgebra is the same in the ultraviolet and the infrared. 
We can therefore compute  $\mathcal{Z}_{S^{2}\times S^1}$
 in the UV where the 3d gauge theory is free.
 The   index can be further refined by turning on fugacities for the flavour symmetries, 
 and  background fluxes on $S^2$ for the flavour groups \cite{Kapustin:2011jm}. 
 In our calculation we will include flavour fugacities   but set the  flavour fluxes to zero. 
 
      The   superconformal index eq.\,\eqref{sci}   is the appropriately twisted partition function of the theory
 on $S^2\times S^1$. It can be computed  by supersymmetric 
 localization of   the functional   integral, 
for a review   see ref.\,\cite{Willett:2016adv}. 
   For each  gauge-group factor  $U(N_j)$
    there is a sum over monopole charges $\{ m_{j, \alpha}\}\in\ZZ^{N_j}$ and 
    an integral over gauge fugacities  (exponentials of  gauge holonomies)
    $\{ z_{j, \alpha}\}\in U(1)^{N_j}$. The calculation is standard   and
     is summarized  in \autoref{appB}.  The result is most conveniently expressed with the help  of the
     plethystic exponential (PE) symbol, 
    \begin{equation}\label{fullindexoflinearquiverQpQm} 
  \begin{aligned}
    \mathcal{Z}_{S^{2}\times S^1}
    & = \! \prod_{j=1}^k\Biggl[ \frac{1}{N_j!} \! \sum_{m_j\in\ZZ^{N_j}} \! \int \! \prod_{\alpha=1}^{N_{j}}\frac{dz_{j,\alpha}}{2\pi i z_{j,\alpha}} \Biggr] \Biggl\{  \! (q^{\frac{1}{2}}t^{-1})^{\Delta({\bf m})} \! 
     \prod_{j=1}^k      \biggl[   \prod_{\alpha=1}^{N_j} w_j^{ m_{j,\alpha}} \! \!
     \prod_{\alpha\neq\beta}^{N_j} (1-q^{{1\over 2} \abs{ m_{j,\alpha}-m_{j,\beta}}}z_{j,\beta}z_{j,\alpha}^{-1})
     \biggr]
    \\
    & \qquad \times {\rm PE} \Biggl( \sum_{j=1}^k \sum_{\alpha,\beta=1}^{N_j} \frac{q^{1\over 2} (t^{-1}-t)}{1-q}\,
    q^{ \abs{ m_{j,\alpha}-m_{j,\beta}}}z_{j,\beta}z_{j,\alpha}^{-1}
    \\
    & \qquad\qquad\qquad  + \frac{(q^{1\over 2} t)^{1\over 2} (1 - q^{1\over 2} t^{-1})}{1-q}
    \sum_{j=1}^k \sum_{p=1}^{M_j} \sum_{\alpha=1}^{N_j}
   q^{{1\over 2} \abs{m_{j,\alpha}}} \sum_{\pm}z_{j,\alpha}^{\mp 1}\mu_{j,p}^{\pm 1}
    \\
    & \qquad\qquad\qquad + \frac{(q^{1\over 2} t)^{1\over 2} (1 - q^{1\over 2} t^{-1})}{1-q}
    \sum_{j=1}^{k-1} \sum_{\alpha=1}^{N_j}\sum_{\beta=1}^{N_{j+1}}
    q^{{1\over 2} \abs{m_{j,\alpha}-m_{j+1,\beta}}} \sum_{\pm}z_{j,\alpha}^{\mp 1}z_{j+1,\beta}^{\pm 1}
    \Biggr)
    \Biggr\}\ .  
  \end{aligned}
\end{equation} 
  Here    $z_{j, \alpha}$ is 
  the   $S^1$ holonomy of the   $U(N_j)$ gauge field 
 and $m_{j, \alpha}$ its   2-sphere  fluxes ({\it viz}. the monopole charges of the corresponding local operator in $\mathbb{R}^3$) with  $\alpha$ labeling  the Cartan generators; 
  $\mu_{j,p}$ are  flavour fugacities with $p = 1, \ldots, M_j$, 
  and  $w_j$ is a fugacity for the topological $U(1)$  
whose conserved current is 
$\Tr{\star F_{(j)}}$\,. 
    The plethystic exponential of a  function  
     $f(v_1, v_2, \cdots )$  is given  by  
 \bee
 {\rm PE}(f) = \exp\left( \sum_{n=1}^\infty {1\over n} f(v_1^n, v_2^n, \cdots) \right)\ . 
 \eeq
  Finally ${\bf m}$ denotes collectively all magnetic charges, and the crucial exponent $\Delta({\bf m}) $ reads
  \begin{equation}\label{Delta}
  \Delta({\bf m})  =  -{1\over 2} \sum_{j=1}^k \sum_{\alpha,\beta=1}^{N_j}  \abs{m_{j,\alpha}-m_{j,\beta}}
  + {1\over 2} \sum_{j=1}^k M_j \sum_{\alpha=1}^{N_j} \abs{m_{j,\alpha}}
  + {1\over 2} \sum_{j=1}^{k-1} \sum_{\alpha=1}^{N_j} \sum_{\beta=1}^{N_{j+1}} \abs{m_{j,\alpha}-m_{j+1,\beta}} \ .
\end{equation}
Note that the smallest power of $q$ in any given monopole sector
is
${1\over 2} \Delta({\bf m})$. Since the contribution of any  state to the index 
is proportional to $q^{{1\over 2}(\Delta + J_3)}$, we see  that  $\Delta({\bf m})$
is the   Casimir energy of the ground state in the sector ${\bf m}$, or equivalently 
 the scaling dimension [and the $SO(3)_C$ spin]  of the 
corresponding monopole operator \cite{Borokhov:2002ib,Borokhov:2003yu}. 
As  shown by Gaiotto and Witten \cite{Gaiotto:2008ak}  this
dimension  is strictly positive (for ${\bf m}\neq 0$) for all the good theories
   that interest us  here. 
  \smallskip 
  
      We would now like  to extract from the index \eqref{fullindexoflinearquiverQpQm}  the number,     
      flavour 
      representations and $U(1)$  charges
      of all marginal ${\cal N}=2$  operators.  To this end we need to rewrite  the index 
      as a sum over characters of the global\  $OSp(4\vert 4)\times G\times \hat G$\  symmetry, 
       \bee\label{chdec}
      \mathcal{Z}_{S^{2}\times S^1} =  \sum_{(\mathfrak{R}, {\bf r}, {\bf \hat r})}  
      {\cal I}_\mathfrak{R}  ( q, t)\, \chi_{\bf r}(\mu) \chi_{\bf \hat r}(\hat \mu ) \ 
      \eeq
      where the sum runs over all triplets of representations $(\mathfrak{R}, {\bf r}, {\bf \hat r})$, 
      $\chi_{\bf r}$  and $\chi_{\bf \hat r}$ are characters of $G$ and $ \hat G$, 
      and ${\cal I}_\mathfrak{R}  $ are characters of $OSp(4\vert 4)$ with 
      fugacities restricted as pertaining for  the index.\footnote{The restriction on fugacities can also be understood as the fact that ${\cal I}_\mathfrak{R}$ are characters of the commutant of $Q^{(++)}_-$ inside $OSp(4|4)$.}
      To proceed we must now  make a detour to   review   the  unitary representations of the
       ${\cal N}=4$ superconformal algebra in three dimensions.

   %%%%%%%%%%%%%%%%

\section{Characters of \texorpdfstring{$OSp(4\vert 4)$}{Osp(4|4)} and Hilbert series}\label{sec3}
 
 \subsection{Representations and recombination rules}
 
   All  unitary highest-weight representations 
   of $OSp(4\vert 4)$ have been  classified in refs.\,\cite{Dolan:2008vc,Cordova:2016emh}. 
  As shown in these references, in addition to the generic long representations
    there exist three  series  of short or BPS representations: 
 \bee\label{listr}
  A_1[j]_{1+j+j^H+j^C}^{(j^H,\, j^C)}\quad (j >0) \ , \qquad
  A_2[0]_{1 +j^H+j^C}^{(j^H,\, j^C)}\ , \quad {\rm and}\quad
  B_1[0]_{j^H+j^C}^{(j^H,\, j^C)}\ . 
 \eeq
 We follow  the notation of \cite{Cordova:2016emh} where $[j]_\delta^{(j^H, \, j^C)}$ denotes a superconformal
 primary  with energy  
  $\delta$,  and   $SO(1,2)\times SO(3)_H\times SO(3)_C$ 
 spin quantum numbers $j, j^H, j^C$.\footnote{Factor of~$2$ differences from \cite{Cordova:2016emh}
 are because we use spins rather than Dynkin labels.}
     We   use  lower-case symbols for the quantum numbers
 of the superconformal primaries  in order to distinguish them from  those of arbitrary states in the representation.
 The subscripts labelling $A$ and  $B$  indicate  the level of the first null states in  the representation. 
 
     The $A$-type representations lie at the unitarity threshold  ($\delta_A = 1+j+j^H+j^C$)
 while those of  $B$-type are separated from this  threshold by a gap,  $\delta_B =  \delta_A  -1$. 
 Since  for  short representations
  the primary dimension $\delta$  is fixed 
  by the spins and the representation type, 
 we will from now on   drop  it  in order to make the notation lighter. 
 
 \smallskip
 
    The general character of  $OSp(4\vert 4)$ is a function of four fugacities, corresponding to  the
    eigenvalues of the four commuting bosonic generators $J_3, J_3^H, J_3^C$ and $\Delta$.  For  the index
    one fixes the fugacity of   $J_3$ and then  a second fugacity automatically  drops out. More explicitly
 \begin{equation}
   \begin{aligned}
     {\cal I}_\mathfrak{R}  ( q, t)  & = \chi_\mathfrak{R}(e^{i\pi}, q, t, e^{\beta})
     \\
     &   {\rm where} 
     \quad  \chi_\mathfrak{R}( w, q, t, e^{\beta})  = \Tr_\mathfrak{R}\,  w^{2J_3} q^{\frac{1}{2}(\Delta + J_3)}t^{J_3^H-J_3^C}e^{-\beta(\Delta
       -J_3 -J_3^H -J_3^C)}\ . 
   \end{aligned}
 \end{equation}
 Although  general  characters are linearly-independent functions,  this is not the case for  indices. 
 The index of  long representations  vanishes, 
 and  the indices of  short representations that can recombine into a long one  sum up
   to zero. This is why, as is well known, $\mathcal{Z}_{S^{2}\times S^1}$ does not  determine 
   (even) the BPS spectrum of the theory unambiguously. Fortunately, we can avoid this
   difficulty for our purposes here,  as we will now explain.

   In  any 3d ${\cal N} = 4$ SCFT 
      the    ambiguity  in extracting the BPS spectrum from the index 
     can be   summarized  by the following recombination rules   \cite{Cordova:2016emh}
   \begin{subequations} \label{recomb}
\begin{align}  
 L[0]^{(j^H  \hskip -0.8mm ,\,  j^C)} 
 \to   & \ \  A_2[0]^{(j^H  \hskip -0.8mm ,\,  j^C)}   \, \oplus \, B_1[0]^{(j^H+1,\, j^C+1)} \ , 
 \\
   L\bigl[\tfrac{1}{2}\bigr]^{(j^H  \hskip -0.8mm ,\,  j^C)}  \to
   &  \ \   A_1\bigl[\tfrac{1}{2}\bigr]^{(j^H  \hskip -0.8mm ,\,  j^C)} 
   \,  \oplus\,  A_2[0]^{(j^H+{1\over 2},\,  j^C+{1\over 2})} \ , 
\\
   L[{j}\geq 1]^{(j^H  \hskip -0.8mm ,\,  j^C)}  \to
   &  \ \   A_1[j]^{(j^H  \hskip -0.8mm ,\,  j^C)} 
   \, \oplus\,  A_1\bigl[j- \tfrac{1}{2}\bigr]^{(j^H+{1\over 2},\, j^C+{1\over 2})} \ . 
\end{align}    
     \end{subequations}       
 The long representations on the  left-hand side 
 are  taken at  the 
    unitarity threshold $\delta\to \delta_A$.  
             From  these recombination rules one sees that 
              the characters of  the $B$-type  multiplets  form a basis for  contributions
             to the index.  Simple induction indeed gives
             \bee\label{relations}
            (-)^{2j}\,   \mathcal{I}_{A_1[j]^{(j^H  \hskip -0.8mm ,\,  j^C)} }  \, =\,  
              \mathcal{I}_{A_2[0]^{(j^H+j  ,\,  j^C +j)}} \, =\, 
                - \mathcal{I}_{B_1[0]^{(j^H+j+1  ,\,  j^C +j+1)} } \ . 
             \eeq
 We need therefore to  compute the index only for   $B$-type multiplets. 
 The  decomposition  of these latter into highest-weight  representations  
        of  the bosonic subgroup  $SO(2,3)\times SO(4)$ can be found  
        in  ref.\,\cite{Cordova:2016emh}.  Using the known characters of $SO(2,3)$ and $SO(4)$ and taking 
     carefully    the 
         limit    $w\to e^{i\pi}$ leads to 
    the following   indices
             \begin{subequations} \label{indexlist}
\begin{align}   
&\mathcal{I}_{B_1[0]^{(0,0)}} = 1\ , 
\\  
        &\mathcal{I}_{B_1[0]^{(j^H >0, 0)} }    \, =\,  (q^{ {1\over 2}} t)^{j^H}\, {  (1  - 
        q^{1\over 2} t^{-{1}} ) \over (1-q)} \ , 
        \\
             &\mathcal{I}_{B_1[0]^{(0, j^C >0)} }    \, =\, (q^{ {1\over 2}} t^{-1})^{j^C}\,  {  (1  - 
        q^{1\over 2} t  ) \over (1-q)} \ , 
        \\
        & \mathcal{I}_{B_1[0]^{(j^H>0, j^C>0)} }  \, =\, q^{{1\over 2}(j^H+j^C)} t^{j^H-j^C}  {(1
        -q^{1\over 2} (t+  t^{-1})   + q)\over (1-q)} \ . 
 \end{align}    
     \end{subequations}    
  Note that all  superconformal primaries of type $B$  are scalar fields 
   with  $\delta =  j^H + j^C$, so one of them  saturates 
         the BPS bound $\delta = j_3+j_3^H+j_3^C$ and contributes 
         the leading power $q^{{1\over 2}( j^H + j^C)}$ to the index. Things work differently for type-$A$ multiplets whose primary states have $\delta = 1+j+j^H+j^C > j_3+j_3^H+j_3^C$, 
     so they cannot contribute to the index. Their descendants can however saturate the BPS bound
     and contribute, because even though a
     Poincar\'e  supercharge  raises the   dimension  by  ${1\over 2}$, it can at the same time
      increase  $J_3+J_3^H+J_3^C$ by 
     as much as ${3\over 2}$.  
 
    %%%%%%%%%%%%%%
          
 \subsection{Protected multiplets and Hilbert series}\label{sec:HS}    
        
            General contributions to the index can   be  attributed either to a $B$-type
             or  to an $A$-type  multiplet. 
       There exists,  however,  a special class  of  {\emph {absolutely protected}} $B$-type representations which  do not
   appear in the decomposition of any long multiplet. Their contribution to the index can therefore
   be extracted unambiguously.   Inspection of    \eqref{recomb} gives  the following list of multiplets that are
   \bee\label{absprot}
   \underline{\rm absolutely\  protected}: \qquad B_1[0]^{(j^H, j^C)}\qquad {\rm with}\quad 
   j^H \leq {1\over 2}  \quad {\rm or}\quad   j^C\leq {1\over 2}\ . 
   \eeq
   Consider in particular the $B_1[0]^{(j^H, 0)}$ series.\footnote{The representations 
   $B_1[0]^{(j^H, {1\over 2})}$ and  $B_1[0]^{( {1\over 2},\, j^C)}$  only appear in theories with free (magnetic or electric) hypermultiplets and  play no role in good theories.}
   The highest-weights of   these
   multiplets are  chiral ${\cal N}=2$ scalar fields  that do not transform under $SO(3)_C$
   rotations. This is  precisely the Higgs-branch chiral ring   consisting   of    operators 
   made out of ${\cal N}=4$ hypermultiplets of the electric quiver. 
    It is  defined entirely  by the classical $F$-term conditions. 
    Likewise the highest-weights of the $B_1[0]^{(0, j^C)}$ series, which are
   singlets of  $SO(3)_H$, 
   form the   chiral ring of the Coulomb branch whose building blocks are magnetic hypermultiplets. 
  Redefine the fugacities as follows
   \bee
  x_\pm =  q^{1\over 4}t^{\pm{1\over 2}}\ . 
  \eeq
   It follows then  immediately   from \eqref{indexlist} that in the limit  $x_-=0$  the index only receives
   contributions from the Higgs-branch chiral ring, while  in   the limit $x_+=0$
  it only receives contributions from   the chiral ring of the Coulomb branch. 
  
        The generating functions  of these chiral rings, graded according to their dimension and quantum numbers under global symmetries, are known as Hilbert series ($\operatorname {HS}$).  In the context of 3d ${\cal N}=4$ theories
  elegant general formulae for the Higgs-branch  and Coulomb-branch Hilbert series  were derived in      
  refs.\,\cite{Hanany:2011db,Cremonesi:2013lqa,Cremonesi:2014uva,Cremonesi:2014kwa}, see also   \cite{Cremonesi:2017jrk}
  for a review.  It follows  from our  discussion that 
  \bee
      \mathcal{Z}_{S^{2}\times S^{1}}\Bigl\vert_{x_-=0} \, =\,   \operatorname{ HS}^{\rm Higgs}(x_+) \  \qquad
      {\rm and}\qquad 
      \mathcal{Z}_{S^{2}\times S^{1}}\Bigl\vert_{x_+=0} \, =\,   \operatorname{ HS}^{\rm Coulomb}(x_-) \ . 
      \eeq 
       These relations between the superconformal index and the Hilbert series were established in   ref.\,\cite{Razamat:2014pta}
        by matching the corresponding  integral  expressions. 
        Here we   derive them  directly  from the  ${\cal N}=4$ superconformal characters.

                                                                                                                                                                                                                                                                                                            What  about  other operators  of the chiral ring?  The complete ${\cal N}=2$ chiral ring consists of the 
                                                                                                                                                                                                                                                                                                            highest weights in \underline{all}   $B_1[0]^{(j^H, j^C)}$ multiplets of the 
                                                                                                                                                                                                                                                                                                            theory.\footnote{The $A$-type multiplets do not contribute to
                                                                                                                                                                                                                                                                                                            the chiral ring, since none has scalar states that saturate  the BPS bound   
                                                                                                                                                                                                                                                                                                             (i.e.\ $\Delta =   J_3^H+J_3^C$ {\it and} $J=0$).}  As   seen however  from eq.\,\eqref{relations} 
                                                                                                                                                                                                                                                                                                             the mixed-branch
                                                                                                                                                                                                                                                                                                             operators (those with both $j^H$  and $j^C$  $ \geq 1$)  cannot be extracted unambiguously  from the
      index. This shows that there is no simple relation between the  Hilbert series of the full chiral ring
      and the superconformal index. The Hilbert series is better adapted for studying supersymmetric
      deformations of a SCFT, but we lack a
                                                                                                                                                                                                                                                                                                             general method to compute it (see however 
                                                                                                                                                                                                                                                                                                              \cite{Cremonesi:2016nbo,Carta:2016fjb} for interesting ideas in this direction). 
   Fortunately    these complications will  not be important  for the problem at hand.

                                                                                                                                                                                                                                                                                                            The reason  is that  marginal deformations exist only in the restricted set of 
                                                                                                                                                                                                                                                                                                           multiplets:   
    \bee\label{deformm}
     \underline{\rm marginal}: \qquad  B_1[0]^{(j^H, j^C)}\qquad {\rm with}\quad   j^H+j^C = 2\ .  
    \eeq
 These are   in the absolutely protected list \eqref{absprot} with the exception of  $B_1[0]^{(1,1)}$, 
 a very  interesting multiplet  that contains also four spin-3/2 fields in its spectrum.  
This multiplet  is not absolutely protected, but it  is part of a `monogamous relation': its 
  unique  recombination partner is   $A_2[0]^{(0,0)}$ and vice versa. 
  Furthermore    $A_2[0]^{(0,0)}$  is the 
   ${\cal N} = 4$
  multiplet 
  of the conserved energy-momentum tensor  \cite{Cordova:2016emh},\footnote{In the
  dual gravity theory, this recombination   makes the ${\cal N} =4$ supergraviton massive.
  Thus  $B_1[0]^{(1,1)}$  is  a Stueckelberg
   multiplet for the `Higgsing' of  ${\cal N} =4$  AdS supergravity  \cite{Bachas:2017rch, Bachas:2019rfq}. 
  We also  note in passing ref.\,\cite{Distler:2017xba} where  the monogamous relation
   is used in order to extract the
   number of conserved energy-momentum tensors from the superconformal index of $d=4$ class-S theories.}
  which is unique  in any irreducible SCFT\@. 
  As a result the  contribution  of   $B_1[0]^{(1,1)}$ multiplets can  be  also  unambiguously extracted from  the
    index.
 \smallskip 
    
   A similar  though   weaker form of the   argument actually applies to  all\,  ${\cal N}=2$ SCFT\@.  
   Marginal chiral operators belong to short 
     $OSp(2\vert 4)$ multiplets whose only  recombination partners are the conserved 
      ${\cal N}=2$ vector-currents.    We already alluded to 
  this fact  when explaining why  the  3d ${\cal N}=2$ superconformal manifold has the  
         structure of  a moment-map quotient \cite{Green:2010da}. 
     If  the global symmetries of the SCFT are known (they are not  always manifest), one can  extract unambiguously
      its marginal
     deformations from the index (see e.g., \cite{Razamat:2016gzx,Fazzi:2018rkr} for applications).

   %%%%%%%%%%%%%
         
 \section{Calculation of the index}\label{sec:4}
   
    We turn now to the main calculation of this paper, namely the expansion of the 
    expression \eqref{sci} in terms of characters of the global symmetry  $OSp(4\vert 4)\times G\times \hat G$.
    Since we are only interested in the marginal  multiplets \eqref{deformm}  
    whose contribution  starts at order~$q$, it will be sufficient  to expand 
   the  index    to this order. In terms of the fugacities $x_\pm$ we must keep terms up to order~$x^4$. 
             As we have just seen,  each of the terms  in the expansion to this order can be 
             unambiguously  attributed to
  an ${ {OSp}}(4\vert 4)$ representation.

      We  will organize the calculation in terms of the magnetic Casimir energy  
        eq.\,\eqref{Delta}. 
  We start   with the zero-monopole sector, and then proceed to positive   values of   $\Delta(\bf m)$. 
  
  %%%%%%%%%%%%%
  
  \subsection{The zero-monopole sector}\label{sec:41}
  
    In the ${\bf m} =0$ sector  all  magnetic  fluxes  vanish and the gauge symmetry  is unbroken. 
  The expression in front of  the plethystic exponential in \eqref{sci}   reduces to
  \bee
   \prod_{j=1}^k\Biggl[ \frac{1}{N_j!}   \int \prod_{\alpha=1}^{N_{j}}\frac{dz_{j,\alpha}}{2\pi i z_{j,\alpha}}  
     \prod_{\alpha\neq\beta}^{N_j} (1- z_{j,\beta}z_{j,\alpha}^{-1})
     \Biggr]\ .  
  \eeq
 This  can be  recognized as the invariant Haar measure for  the gauge group $\prod_{j=1}^k U(N_j)$. 
 The  measure is   normalized so that for any irreducible representation $R$ of $U(N)$
 \bee
  \frac{1}{N!} \int \prod_{\alpha=1}^{N }\frac{dz_{ \alpha}}{2\pi i z_{ \alpha}}  
     \prod_{\alpha\neq\beta}^{N } (1- z_{ \beta}z_{ \alpha}^{-1})\,\chi_R
     (z) \ = \ \delta_{R, 0}\ . 
 \eeq
Thus the   integral   projects    to gauge-invariant  states,  
 as expected. 
 We  denote  this operation on any combination, $X$,  of 
  characters as      $X\bigl\vert_{\rm singlet}$\,. 
  
 \smallskip
  
  Since we work to order $O(q)$ we may drop the 
  denominators $ (1-q)$ in the plethystic exponential.  
    The contribution of the ${\bf m}=0$ sector to the index  can then be written as 
 \bee\label{43} 
   \mathcal{Z}_{S^{2}\times S^{1}}^{{\bf m}=0}   =  
    {\rm PE} \Bigl(   
     x_+ (1 - x_-^2) X   + (x_-^2-x_+^2)  Y 
    \Bigr)\Bigl\vert_{\rm singlet}\,+\, O(x^5)\  
 \eeq 
 with
  \bee\label{44} 
      X    = \sum_{j=1}^k \bigl(\square_j \overline{\square}^\mu_j + \overline{\square}_j \square^\mu_j\bigr)
    + \sum_{j=1}^{k-1} \bigl(\square_j \overline{\square}_{j+1} + \overline{\square}_j \square_{j+1}\bigr)  \qquad
  {\rm and}\quad 
    Y   = \sum_{j=1}^k \overline{\square}_j \square_j \ . 
\eeq
  The notation here is as follows:  $\square_j$ denotes the character of the fundamental representation
  of the $j$th unitary group, and  $\overline{\square}_j$ that of the anti-fundamental.  To distinguish
  gauge from global (electric) flavour groups  we specify the latter  with   the symbol  of the corresponding
  fugacities $\mu$, while 
  for the gauge group  the  dependence   on the fugacities  $z$ is implicit. The entire plethystic exponent 
  can be considered as a character of  ${\cal G}\times G \times U(1)\times \mathbb{R}^+$, 
  where ${\cal G}$ is the gauge group and  $U(1)\times \mathbb{R}^+\subset  {\rm OSp}(4\vert 4)$ 
  are the superconformal symmetries  generated by 
  $J_3^H-J_3^C$ and by $\Delta +J_3$. 
    The ``singlet'' 
  operation projects on singlets of the gauge group only.

  The plethystic exponential is  a sum 
  of   powers $\mathbb{S}^k\chi$ of characters, where $\mathbb{S}^k$ is a 
  multiparticle symmetrizer that 
  takes into account fermion statistics. 
  For instance 
  \bee
  \mathbb{S}^2(a+b-c-d)= \sym^2 a    +ab+\sym^2b  -(a+b)(c+d)+  {  \Lambda^2 c + cd  +\Lambda^2d }  \   
  \eeq
  where $\sym^k$ and $\Lambda^k$ denote standard  symmetrization or antisymmetrization. 
  Call $\Omega$  the exponent in eq.\,\eqref{43}\,. 
  To the quartic order     that we care about we compute  
\begin{equation}\label{47}
  \begin{aligned}
    \mathbb{S}^2  \Omega  & = x_+^2  \sym^2 X  
    + x_+(x_-^2 - x_+^2) X Y
    + x_-^4 \sym^2 Y   + x_+^4  \Lambda^2 Y 
    - x_+^2 x_-^2 \bigl( X^2 + Y^2 \bigr) + O(x^5) \ ,
    \\
    \mathbb{S}^3  \Omega 
    & = x_+^3 \sym^3X     + x_+^2 (x_-^2 - x_+^2)  Y  \sym^2 X + O(x^5)  \ , 
    \\
    \mathbb{S}^4  \Omega 
    & = x_+^4 \sym^4 X  + O(x^5)    \ . 
  \end{aligned}
\end{equation}
Upon projection on  the gauge-invariant sector one finds  
\bee\label{48}
 X  \bigl\vert_{\rm singlet} =  XY \bigl\vert_{\rm singlet} = 0\,  
 \qquad {\rm and}  \qquad Y \bigl\vert_{\rm singlet} = k\ . 
\eeq
Second  powers of $Y$  also give ($\mu$-independent)  pure numbers, 
\[
Y^2 \Bigl\vert_{\rm singlet} = 
      \sym^2  Y \Bigl\vert_{\rm singlet}  + \,  \Lambda^2 Y \Bigl\vert_{\rm singlet}
\]
with
\bee
\sym^2  Y \Bigl\vert_{\rm singlet}  =  \,  {1\over 2} k(k+1)  + \sum_{j=1}^k  \delta_{N_j\not= 1}  \ ,
\qquad
\Lambda^2 Y \Bigl\vert_{\rm singlet}  =  \,  {1\over 2} k(k-1) \ . 
\eeq
The remaining terms in the expansion require a little  more work with the result
\begin{equation}\label{410}
  \begin{aligned}
      &X^2 \Bigl\vert_{\rm singlet} = 
    2 \, \sym^2  X  \Bigl\vert_{\rm singlet}   = 2\,\Bigl( k - 1 + \sum_{j=1}^k \overline{\square}^\mu_j \square_j^\mu
    \Bigr) \, ,  \\
        & \sym^3 X  \Bigl\vert_{\rm singlet}  = \ \sum_{j=1}^{k-1} \bigl( \square^\mu_j \overline{\square}^\mu_{j+1} + \overline{\square}^\mu_j \square^\mu_{j+1} \bigr)\ ,  \\
        & Y \sym^2X  \Bigl\vert_{\rm singlet} = k^2 + k - 2 + \delta_{N_1= 1} + \delta_{N_k= 1}
    + \sum_{j=1}^k \Bigl[ 
     (k + \delta_{N_j\not=1}) \overline{\square}^\mu_j \square_j^\mu  - 2 \delta_{N_j=1} \Bigr]\ , 
   \end{aligned}
\end{equation}  
  and finally (and most tediously)
  \begin{equation}\label{411}
  \begin{aligned}
    \sym^4 X  \Bigl\vert_{\rm singlet}   & =
   \sum_{j=2}^{k-1} \delta_{N_j\not=1}  + \sum_{j=1}^{k-1} \delta_{N_j\not=1}\delta_{N_{j+1}\not=1}
    + \frac{(k-1)k}{2}
     + \sum_{j=1}^{k-2} \bigl( \square^\mu_j \overline{\square}^\mu_{j+2} + \overline{\square}^\mu_j \square^\mu_{j+2} \bigr) \\
    & \quad
    + \sum_{j=1}^k  { \delta_{N_j\not=1}} (2-\delta_{j=1}-\delta_{j=k}) \vert\square^\mu_j\vert^2 
    + (k-1) \sum_{j=1}^k \abs{\square^\mu_j}^2
    \\
    & \quad + \sum_{j<j'}^k \abs{\square^\mu_j}^2 \abs{\square^\mu_{j'}}^2
    + { 
    \sum_{j=1}^k \abs{ \square\hskip -0.46mm \square^\mu_j}^2
    + \sum_{j=1}^k \delta_{N_j\not=1} \left\vert\, 
    \yngantisym^{\,\mu}_{\,j} \right\vert^2
    }
  \end{aligned}
\end{equation}
where in the last equation we used the shorthand $\abs{R}^2$ for  the character of $R\otimes\overline{R}$,
and denoted 
 the (anti)symmetric  representations of $U(M_j)$ by  Young diagrams.
\smallskip

Let us explain how to  compute   the singlets in $Y  \sym^2 X $.
One obtains gauge-invariant contributions to that term in three different ways: the product of a gauge-invariant from $Y$ and one from $\sym^2 X  $, or the product of an $SU(N_j)$ adjoint in $Y$ with either a fundamental and an antifundamental, or a pair of bifundamentals, coming from $\sym^2  X $.
This gives three terms:
\begin{equation}
  Y \sym^2 X  \Bigl\vert_{\rm singlet} 
  = k\biggl( k - 1 + \sum_{j=1}^k \vert \square_j^\mu\vert^2 \biggr)
  + \biggl( \sum_{j=1}^k \delta_{N_j\not= 1} \vert \square_j^\mu \vert^2 \biggr)
  + \biggl( - \delta_{N_1\not= 1} - \delta_{N_k\not=1} + \sum_{j=1}^k 2 \delta_{N_j\not= 1} \biggr)
\end{equation}
where we used that the $SU(N_j)$ adjoint is absent when $N_j=1$, and that the outermost nodes have a single bifundamental hypermultiplet rather than two. After a small rearrangement, this is the same as the  last line
of \eqref{410}.

For $\sym^4 X  \vert_{\rm singlet}$ we organized  terms according to how many bifundamentals they involve.
First, four bifundamentals can be connected in self-explanatory notation
 as \tikz\draw(0,0)--(.3,.1)--(.6,0)--(.3,-.1)--cycle; or \tikz\draw(0,0)--(.5,.2)--(0,.2)--(.5,0)--cycle; or 
$\sym^2($\tikz\draw(0,0)to[bend left=20](.5,0)to[bend left=20]cycle;$)$.
Next, two bifundamentals and two fundamentals of different gauge groups can be  connected as \tikz\draw(0,0)--(0,.3)--(.25,.25)--(.5,.3)--(.5,0);, while for  the same group they can be either connected as \tikz\draw(0,0)--(0,.2)--(.5,.25)--(-.1,.3)--(-.1,0); or \tikz\draw(0,0)--(0,.2)--(-.5,.25)--(.1,.3)--(.1,0); , or disconnected  as a pair of bifundamentals \tikz\draw(0,0)to[bend left=20](.5,0)to[bend left=20]cycle; and a flavour current \tikz\draw(0,0)--(.05,.3)--(.1,0); (see below).  
When the node is abelian  the first two  terms are already included in the third and should not
be counted separately. 
  Finally, four fundamental hypermultiplets 
   can form two pairs at different nodes, or if they  come  from the same node
  they should be split in two conjugate pairs, $Q_{j, \alpha}^p Q_{j, \beta}^r$  and 
  $\tilde Q_{j, \alpha }^{\bar p} \tilde Q_{j, \beta }^{\bar r} $, with  each pair  separately symmetrized or antisymmetrized.  
  When the  gauge group is abelian  the antisymmetric piece is   absent.

%%%%%%%%%%

\subsection{Higgs-branch chiral ring}

As  a check,  let us use the above results to calculate  the  Hilbert series of the Higgs branch. We have 
 explained in \autoref{sec:HS} that  this is equal to  the index evaluated at 
 $x_-=0$. Non-trivial monopole sectors make a contribution  proportional to $x_-^{2\Delta({\bf m})}$ 
and since   $\Delta({\bf m})>0$ they can be neglected. 
 The Higgs-branch Hilbert series therefore reads
 \bee
   \operatorname{HS}^{\rm Higgs}(x_+) =  \mathcal{Z}_{S^{2}\times S^{1}}^{{\bf m}=0} \Bigl\vert_{x_-=0}\ .  
 \eeq
   Setting  $x_-=0$ in eqs.\,\eqref{43} and \eqref{47}  we find  \smallskip
       \begin{equation} 
  \begin{aligned}
  \operatorname{HS}^{\rm Higgs}(x_+)  \, = \,  1    + \, x_+^2  \bigl( \sym^2 X  -   Y \bigr)\Bigl\vert_{\rm singlet} &
+ x_+^3    \sym^3 X \Bigl\vert_{\rm singlet}  \\   +  x_+^4   \bigl( \sym^4 X    + & \Lambda^2 Y - Y  \sym^2 X  \bigr)\Bigl\vert_{\rm singlet}
  + \, \, O(x_+^5)  \, . 
 \end{aligned}
\end{equation}
Inserting   now  \eqref{48}-\eqref{411} gives,  after some straightforward algebra in which we distinguish $k=1$ from $k>1$ because simplifications are somewhat different,
\begin{align}
  \nonumber
  \operatorname{HS}^{\rm Higgs}(x_+)
  & \overset{k=1}{=} 1 + x_+^2 \biggl(  \underbrace{\textstyle\vert \square^\mu_1\vert^2  - 1 }_{ {\rm Adj}\,G  }\biggr)
    + x_+^4 \biggr( \underbrace{\textstyle
    \abs{ \square\hskip -0.46mm \square^\mu_1}^2
    + \delta_{N_1\neq 1}
    \left\vert\, \yngantisym^{\,\mu}_{\,1} \right\vert^2
    - (1+\delta_{N_1\neq 1}) \vert  \square^\mu_1 \vert^2
    }_{\text{double-string operators}}
    \biggr) + O(x_+^5)\ ,
  \\[.5\baselineskip]
  \nonumber
  \operatorname{HS}^{\rm Higgs}(x_+)
    & \overset{k>1}{=} 1\, +\,  
    x_+^2\, \biggl(  \underbrace{\textstyle\sum_{j=1}^k \vert \square_j^\mu\vert^2  - 1 }_{ {\rm Adj}\,G  }\biggr)
    \,+\,
    x_+^3 \,   \underbrace{\textstyle
    \sum_{j=1}^{k-1} \bigl( \square^\mu_j \overline{\square}^\mu_{j+1} + \overline{\square}^\mu_j \square^\mu_{j+1} \bigr)}_{\chi_{\ell=3}:\ {\rm length=3\  strings}}
  \\
  \label{414}
  & \qquad + \, x_+^4 \, \biggl(\, \underbrace{\textstyle  \sum_{j<j'}^k \abs{\square^\mu_j}^2  
  \abs{\square^\mu_{j'}}^2
  +  
  \sum_{j=1}^k    \Bigl(  \abs{ \square\hskip -0.46mm \square^\mu_j}^2
  +   \delta_{N_j\not=1}
  \left\vert\, 
  \yngantisym^{\,\mu}_{\,j} \right\vert^2   -        \vert  \square^\mu_j \vert^2  \Bigr) }_{\text{double-string operators}}
  \\[-.3\baselineskip]
  \nonumber
  & \qquad \qquad \qquad +\, 
    \underbrace{\textstyle  \sum_{j=2}^{k-1} \Bigl(\square^\mu_{j-1} \overline{\square}^\mu_{j+1} + \overline{\square}^\mu_{j-1} \square^\mu_{j+1}
    + \delta_{N_j\neq 1} \vert \square^\mu_j   \vert^2 \Bigr)  - \, \Delta n_{\rm H}}_{\chi_{\ell=4}:\ {\rm length=4\  strings}}  \, 
    \biggr)\ + \ O(x_+^5)\ . 
\end{align}
where $\Delta n_{\rm H}$ in the last line is a pure number given by
 \bee\label{414a}
\Delta n_{\rm H} = 1 + \sum_{j=2}^{k-1} \delta_{N_j =1} - \sum_{j=1}^{k-1} \delta_{N_j =1}\delta_{N_{j+1} =1} \ .
\eeq

\vskip 3mm
This result agrees with expectations. Recall that the Higgs branch is classical and its
 Hilbert series  counts chiral
operators made out of the scalar fields,  $Q_j^p$ and $\tilde Q_j^{\bar p}$, 
of the 
  (anti)fundamental 
hypermultiplets, and the scalars  of the bifundamental 
hypermultiplets  $Q_{j, j+1}$
and $\tilde Q_{j+1,j}$ (the 
  gauge indices are here suppressed).  Gauge-invariant products of these scalar fields can be drawn as strings on the quiver diagram \cite{Bachas:2017wva}, and
   they obey the following $F$-term matrix relations  derived from the ${\cal N}=4$ superpotential, 
  \bee\label{Fterm}
     Q_{j, j+1}\tilde Q_{j+1,j} +   \tilde Q_{j ,j-1} Q_{j-1, j } +  \sum_{p,\bar p=1}^{M_j}   Q_j^p \tilde Q_j^{\bar p}
     \delta_{p  \bar p}  = 0 \qquad \forall j=1, \cdots ,  k\ . 
  \eeq
 The length of each string gives the   $SO(3)_H$ spin and  scaling dimension of the operator, 
 and hence  the power of $x_+$ in the index.  Since good theories have no free hypermultiplets there are
 no contributions at order $x_+$. At order $x_+^2$ one finds the scalar partners of the conserved flavour
 currents  that transform in the adjoint representation of $G$.  Higher powers come  either from 
 single longer strings or, starting at order $x_+^4$, from
 multistring `bound states'.   One   indeed recognizes  the  second line in \eqref{414} as  the symmetrized
 product of strings of length two, 
  \bee
 \sym^2 \AdjG  \, = \,  \sym^2  \Bigl(  \sum_{j=1}^k \vert \square_j^\mu\vert^2  - 1  \Bigr)  \ , 
 \eeq  
       modulo the fact that for abelian gauge nodes some of the states are
        absent. These and the additional   single-string operators of length 3 and 4  
  can be enumerated  by  diagrammatic rules, we refer the reader
 to \cite{Bachas:2017wva}
 for   details.

Note that single- and double-string operators with the same flavour 
 quantum numbers may mix.  The  convention adopted  in eq.\,\eqref{414} is to count   
  such operators as double strings. 
 In particular,   length~$4$ single-string operators transforming in the adjoint of the flavour symmetry group (at a   non-abelian node)
 are related by the $F$-term constraint~\eqref{Fterm} to products of currents,
 which explains the coefficient~$1$ of $\abs{\square^\mu_j}^2$ in contrast with its coefficient~$2$ in $\sym^4X|_{\text{singlet}}$ eq.\,\eqref{411}. 
 In the special case $k=1$,    all length~$4$ strings are products of currents and  some   
 vanish by the $F$-term constraint~\eqref{Fterm}. 
 Note also  that the correction term $\Delta n_{\rm H}$ is the number of disjoint parts
 of the quiver when all  abelian nodes are deleted.
 For each such part (consecutive non-abelian nodes) one neutral  length-4 operator  turns 
 out to be redundant by
 the $F$-term conditions.\footnote{Let $j_1,\dots,j_2$ be the non-abelian nodes in such a part, and focus on the case where the nodes $j_1-1$ and $j_2+1$ are abelian (the discussion is essentially identical if instead we have the edge of the quiver).  Because of the abelian node, the closed length~$4$ string \tikz\draw(0,0)--(.5,.2)--(0,.2)--(.5,0)--cycle; passing through nodes $j_1-1$ and $j_1$ factorizes as a product of currents.  On the other hand the $F$-term constraint~\eqref{Fterm} at $j_1$ expresses $Q_{j_1,j_1-1}Q_{j_1-1,j_1}$ as a sum of two terms and squaring it relates the string under consideration to a sum of three terms: a string of the same shape \tikz\draw(0,0)--(.5,.2)--(0,.2)--(.5,0)--cycle; passing through $j_1$ and $j_1+1$, a string of shape
\tikz\draw(-.05,0)--(0,.2)--(.5,.25)--(-.1,.3)--cycle; passing through these two nodes and the flavour node $M_{j_1}$, and a string \tikz\draw(0,0)--(0,.3)--(.1,0)--(.1,.3)--cycle; visiting the gauge and flavour nodes~$j_1$.  The third is a product of currents.  The first can be rewritten using the $F$-term condition of node $j_1+1$.  Continuing likewise until reaching a string of the same shape \tikz\draw(0,0)--(.5,.2)--(0,.2)--(.5,0)--cycle; passing through $j_2$ and $j_2+1$, one finally obtains the sought-after relation between many neutral length~$4$ operators and products of conserved currents.}

    The quartic term of the Hilbert series   counts marginal  Higgs-branch operators. 
  When the  electric flavour-symmetry group $G$ is large,   the  vast majority of  
  these  are  double-string operators. 
   Their  number far exceeds   the number  (dim$\,G$) 
   of moment-map constraints,  eq.\,\eqref{quotient1}, so generic $T_\rho^{\hat \rho}$ theories have 
   a large number of double-string ${\cal N}=2$ moduli. 
  
  %%%%%%%%%%%%

 \subsection{Contribution of  monopoles}

 Going back   to the full  superconformal  index,  we  separate it   in three  parts as follows  
    \bee\label{breakup}
  \mathcal{Z}_{S^{2}\times S^{1}}  =  - 1 +   \operatorname{HS}^{{\rm Higgs}}(x_+, \mu)    \,
  +\,    \operatorname{HS}^{{\rm Coulomb}}(x_-, \hat\mu)   \,  + \  \mathcal{Z}^{{\rm mixed}}(x_+, x_-, \mu, \hat\mu)
 \eeq
 where  the remainder
 $\mathcal{Z}^{{\rm mixed}}$   vanishes if either $x_-=0$ or $x_+=0$. 
 The Higgs-branch Hilbert series  only
 depends   on  the electric-flavour fugacities $\mu_{j,p}$, and  the   Hilbert series of the Coulomb branch only
 depends on the magnetic-flavour  fugacities $w_j$. 
 To render the  notation  mirror-symmetric  these latter should be redefined
 as follows 
   \bee
 w_j = \hat\mu_j \hat\mu_{j+1}^{-1}\ .  
 \eeq
 Note that since  the index \eqref{fullindexoflinearquiverQpQm} 
 only depends on ratios  of the $\hat\mu_j$, 
 the last fugacity $\hat \mu_{k+1}$ is arbitrary and can be fixed at will. 
 This  reflects the fact  that a phase rotation of all fundamental magnetic quarks
is  a gauge  rather than  global symmetry.  
\smallskip

 Mirror symmetry predicts  that    $ \operatorname{HS}^{{\rm Coulomb}}$ 
  is   given by   the same expression 
 \eqref{414}    with $x_+$ replaced by $x_-$ and all other quantities replaced by their hatted mirrors.  
 We will  assume that this is indeed  the case\footnote{It is straightforward to verify the
 assertion at the quartic order computed here. Mirror symmetry of 
 the complete index  can be proved  by induction (I. Lavdas and B. Le Floch, work in progress).} 
 and  focus   on the mixed  piece $\mathcal{Z}^{{\rm mixed}}$. 
  
 \smallskip 
  As opposed  to   the two  Hilbert series,  which only receive contributions from $B$-type primaries, 
 $\mathcal{Z}^{{\rm mixed}}$    has contributions  from both $A$-type   and $B$-type
 multiplets,  and    from both superconformal primaries and descendants.  
 Let us first collect for later reference 
  the terms of the  ${\bf m}=0$ sector that were not included in the Higgs-branch
 Hilbert series. From the results in \autoref{sec:41} one finds
        \begin{equation}\label{420m} 
  \begin{aligned}
   \mathcal{Z}_{S^{2}\times S^{1}}^{{\bf m}=0 } -  HS^{{\rm Higgs}} \ & = \, 
   \Bigl[ x_-^2 Y  \,  +\,  x_-^4 \sym^2 Y    \, +  \  x_+^2 x_-^2 
   \bigl(  Y \sym^2 X  - X^2 - Y^2 \bigr)  
   \Bigr]_{\rm singlet}\\  
   &   = \  x_-^2  k \,  +\,  x_-^4 \Bigl(    {1\over 2} k(k+1)  + \sum_{j=1}^k  \delta_{N_j\not= 1}   \Bigr)   \\ 
   &\hskip -2cm   +  \  x_+^2 x_-^2 \Bigl(   \sum_{j=1}^k  (k-1 - \delta_{N_j=1}) \vert \square_j^\mu\vert^2 
   - 2k  -
     \sum_{j=1}^k  \delta_{N_j=1}  + \delta_{N_1=1}+ \delta_{N_k=1}
   \Bigr)\ + O(x^5)\,  .  
  \end{aligned}
\end{equation}   
     The  two terms in the second line  contribute  to the Coulomb-branch Hilbert series,   
       while the third line is a  contribution  to the mixed piece.
 
  We  turn next  to  non-trivial   monopole sectors whose contributions are proportional to  $x_-^{ 2\Delta({\bf m})}$.
  At  the order of interest we can restrict ourselves  to  sectors with $0 < \Delta({\bf m})\leq 2$\,. Finding
  which monopole charges contribute  to  a generic  value of  $\Delta({\bf m})$ is a hard 
      combinatorial problem. 
      For the   lowest values $\Delta({\bf m}) = {1\over 2}, 1$\, 
      and for  good theories it was solved in  ref.\,\cite{Gaiotto:2008ak}.
      
      Fortunately this will be sufficient  for our purposes here since,  to the order of interest,  
      the sectors  $\Delta({\bf m})= 2$ and $\Delta({\bf m})= {3\over 2}$  only contribute  
        to  the Coulomb-branch
        Hilbert series, not to the mixed piece. 
                This is obvious for $\Delta({\bf m})= 2$, while for
        $\Delta({\bf m})= {3\over 2}$   subleading terms in  \eqref{fullindexoflinearquiverQpQm}  with a single
          additional power of    $q^{1/4}$ have  unmatched   gauge fugacities $z_{j, \alpha}$,  
          and   vanish after projection to the 
         invariant  sector (see below). In addition,  good theories have  no
           monopole operators  with  $\Delta({\bf m})= {1\over 2}$.  Such  
            operators would have been   free twisted hypermultiplets, and there are none    in the  spectrum
           of good theories.  This leaves us with $\Delta({\bf m})= 1$.
 
  \smallskip 
   The key concept  for describing monopole charges   is that of {\it balanced}  quiver nodes,    
  defined as the  nodes that  saturate  the `good'  inequality $N_{j-1}+N_{j+1}+M_j \geq 2N_j$. Let 
  ${\cal B}_\xi$  denote  the   sets of consecutive balanced nodes, i.e.\  the disconnected parts
  of the quiver diagram after  non-balanced nodes have been deleted. 
  As shown in \cite{Gaiotto:2008ak} each such 
  set corresponds to a non-abelian flavor group  $SU(\vert {\cal B}_\xi\vert +1)$ in the mirror magnetic 
  quiver.\footnote{As a result  $\xi$ ranges over the different components of the magnetic flavour group,
  i.e.\  the subset of   gauge nodes ($\hat\jmath = 1, \cdots \hat k$)
   in the mirror quiver of the magnetic theory for which $\hat M_{\hat\jmath}>1$.} 
   Monopole charges in   the sector
  $\Delta({\bf m})=  1$ are necessarily  of the following  form: 
  all $m_{j,\alpha}$ vanish except
  \bee\label{421}
  m_{j_1 ,\alpha_1} = m_{j_1+1,\alpha_2}= \cdots = m_{j_1+\ell,\alpha_\ell} = \pm 1\   
   \qquad  {\rm with}\quad  \  [ j_1,   j_1+\ell ]   \subseteq \   {\cal B}_\xi 
   \eeq 
  for one   choice of color    indices   at   each gauge  node,  and for one  given  set of balanced nodes, ${\cal B}_\xi$. 
  Up to permutations
  of the color indices we  can    choose  $\alpha_1 =   \alpha_2 =   \cdots  =\alpha_\ell = 1$. 
 
   Define  $j_1+\ell \equiv  j_2$, and let
  $\Gamma$ be the sequence of gauge nodes   $\Gamma = \{j_1, j_1+1, \cdots j_1+\ell \equiv j_2\}$. 
  To determine the contribution of \eqref{421} to the index,  
     we  first note that the above  assignement of
     magnetic fluxes breaks  the gauge symmetry down 
    to 
    \bee
    {\cal G}_\Gamma = \prod_{j\notin \Gamma} U(N_j) \times \prod _{j\in \Gamma} 
    \left[ U(N_j-1)\times U(1)\right] \ . 
    \eeq
  Let us pull out of the integral expression  \eqref{fullindexoflinearquiverQpQm}  
   the   fugacities $\prod_{j\in \Gamma} w_j^\pm
   $ and  the overall factor  $x_-^2$. Setting  $q=0$  everywhere else  
   and summing over equivalent permutations of color indices gives precisely 
    the   invariant measure of  ${\cal G}_\Gamma$, normalized so that it  integrates  to $1$.
   To  calculate all  terms systematically we must therefore  expand  the integrand in powers of $q^{1/4}$,  
   and  then project  on  the 
     ${\cal G}_\Gamma$ invariant sector.
 To  the order of interest we find
   \begin{equation}
 \mathcal{Z}_{S^{2}\times S^{1}}^{\Delta({\bf m})=1} =
   x_-^2   \sum_{{\cal B}_\xi }    \sum_{\Gamma \subseteq  {\cal B}_\xi } 
   \biggl( \prod_{j\in \Gamma} w_j  + \prod_{j\in \Gamma} w_j^ {-1} \biggr) \,  {\rm PE} \Bigl(  x_+  X^\prime +  (x_-^2-x_+^2) Y^\prime   - x_+ x_- Z^\prime
   \Bigr)\Biggl\vert_{{\cal G}_\Gamma \, {\rm singlet}} \hskip -6mm  +\, O(x^5)
 \end{equation}
   where 
\begin{equation}
\begin{aligned}
 X^\prime & = \sum_{j=1}^k  ( \overline{\square}^\mu_j \square_j^\prime + {\square}^\mu_j 
   \overline\square_j^\prime) + \sum_{j=1}^{k-1} ( \overline{\square}^\prime_j \square_{j+1}^\prime + {\square}^\prime_j 
   \overline\square_{j+1}^\prime) + \sum_{j, j+1\in\Gamma}  (z_{j, 1}z_{j+1, 1}^{-1} + z_{j, 1}^{-1}z_{j+1, 1}) \ , \\
   Y^\prime & =  \sum_{j=1}^k  \overline{\square}_j \square_j  =  \sum_{j=1}^k  \overline{\square}^\prime_j \square_j^\prime  +  (j_2-j_1+1)\ , \qquad
   Z^\prime = \sum_{j=1}^k \bigl( \overline{\square}_j^\prime z_{j,1} + z_{j,1}^{-1} \square_j^\prime \bigr) \ ,
\end{aligned}
\end{equation}
and in these expressions $\square_j^\prime$ denotes  the fundamental of  $U(N_j')$ where $N_j'=N_j-1$ if $j\in\Gamma$, 
and $N_j'=N_j$ if  $j\notin\Gamma$.
By convention\,  $\square_j^\prime =0$  when  $N_j^\prime =0$.
\smallskip
        
            Performing the  projection  onto ${\cal G}_\Gamma$ singlets
  gives
  \begin{equation}
    \begin{gathered}
      X^\prime\vert_{{\cal G}_\Gamma\, {\rm singlet} } = 0\, , \qquad
      Y^\prime\vert_{{\cal G}_\Gamma\, {\rm singlet} } =    (j_2-j_1+1)  + \sum_{j=1}^k  \delta_{N_j^\prime \not=0} \,, \qquad
      Z^\prime\vert_{{\cal G}_\Gamma\, {\rm singlet} } = 0\ ,
      \\
      \text{and}\quad
      \sym^2 X^{\prime } \vert_{{\cal G}_\Gamma \, {\rm singlet}} \, =\, \,  \sum_{j=1}^k \delta_{N_j^\prime \not=0}
    \,  \overline{\square}^\mu_j   \square^\mu_j       
    +    \sum_{j=1}^{k-1}  \delta_{N_j^\prime\not=0}\delta_{N_{j+1}^\prime\not=0}  + (j_2-j_1) \ . \qquad
  \end{gathered}
\end{equation}
  Collecting  and rearranging  terms  gives  
     \begin{equation}\label{426m} 
  \begin{aligned} 
    &    \mathcal{Z}_{S^{2}\times S^{1}}^{\Delta({\bf m})=1} \     =\     \sum_{{\cal B}_\xi }     \sum_{[j_1, j_2]\subseteq
       {\cal B}_\xi } 
     (\hat\mu_{j_1} \hat\mu_{j_2+1}^{-1}   + \hat\mu_{j_1}^{-1} \hat\mu_{j_2+1}) \Biggl[
     x_-^2 \, +\,  x_-^4  \Bigl(k + \sum_{j\in \Gamma}  \delta_{N_j\not=1}\Bigr)
     \\
     & \qquad\quad  +    x_-^2 x_+^2\, \Bigl(\sum_{j=1}^k  \delta_{N_j^\prime \not=0}\,  \overline{\square}^\mu_j \square^\mu_j   +  \sum_{j=1}^{k-1}  \delta_{N_j^\prime\not=0}\delta_{N_{j+1}^\prime\not=0} -
     \sum_{j=1}^k  \delta_{N_j^\prime \not=0} -  1 
      \Bigr) \Biggr] \ + \ O(x ^5)\ . 
     \end{aligned}
\end{equation}         
  
     The terms that do not vanish for $x_+=0$ are contributions to the Hilbert series of the Coulomb branch.
   For a check let us consider the leading  term.   Combining it with the one from  
   eq.\,\eqref{420m} gives the adjoint representation of $\hat G$, as predicted by mirror symmetry
    \begin{equation} 
  \begin{aligned} 
             \operatorname{HS}^{{\rm Coulomb}} \,   =\,     1 + x_-^2\,  
            \Bigl[   \underbrace{   k +   \sum_{{\cal B}_\xi }    \sum_{j_1\leq j_2\in  {\cal B}_\xi } 
     (\hat\mu_{j_1} \hat\mu_{j_2+1}^{-1}   + \hat\mu_{j_1}^{-1} \hat\mu_{j_2+1}) 
      }_{{\rm Adj}\,\hat G} \Bigr] \,    +  \,   O(x_-^3)
    \  .      
               \end{aligned}
\end{equation}    
   Note that the  $k$ Cartan generators of $\hat G$  (those corresponding  to the 
   topological symmetry) contribute to the  index 
   in  the ${\bf m}=0$ sector. The monopole operators  that enhance  this symmetry  in the infrared
   to the full non-abelian magnetic  group    enter in the sector $\Delta({\bf m})=1$.

%%%%%%%%%%%%%%

  \subsection{The mixed term}\label{44a}
      
     Let us now put together  the  mixed terms  from eqs.\,\eqref{420m} and \eqref{426m}.
      If   the quiver has no abelian nodes  all $N_j>1$
     and all $N_j^\prime >0$,  and our    expressions  simplify enormously. The last line in  eq.\,\eqref{420m}
     collapses to $(k-1) \sum_j \vert \square^\mu_j\vert^2 - {2 k   }$,   
     and the last line of  \eqref{426m} collapses 
     to $\sum_j \vert \square^\mu_j\vert^2 - 2$. Combining  the two  gives 
     the following result  for quivers with

     \vspace{.6\baselineskip}

     \noindent\underline{\bf No abelian nodes}:\vspace{-.6\baselineskip}
        \begin{equation}\label{mixs} 
  \begin{aligned} 
    \mathcal{Z}^{{\rm mixed}}  \      =\ &   x_+^2 x_-^2
      \Bigl[   \Bigl(\sum_j \vert \square^\mu_j\vert^2 - 2\Bigr)   \Bigl( k - 1 
     +  \sum_{{\cal B}_\xi }  \sum_{ \epsilon =\pm}   \sum_{[j_1, j_2]\subseteq
       {\cal B}_\xi }  (\hat\mu_{j_1} \hat\mu_{j_2+1}^{-1})^{\epsilon}\Bigr) - 2 \Bigr] \,  + \, O(x^5)   
              \\       &
   =  x_+^2 x_-^2\,  \Bigl[   (\AdjG  - 1)( \chi_{{\rm Adj}\,\hat G} - 1) - 2  \Bigr] \,  + \, O(x^5) \  . 
       \end{aligned}
\end{equation}      
   We will interpret this result in the following section. But first let us consider the
   corrections coming from  abelian nodes.  
    
  \vskip 0.6mm

     The $\mu$-dependent correction in the ${\bf m}=0$ sector, 
     eq.\,\eqref{420m},  is  a  sum of  $\vert \square^\mu_j\vert^2$   over all abelian gauge nodes,
    which  should be subtracted from the above result. 
   We expect, by mirror symmetry,  a similar subtraction   for  abelian gauge nodes of the 
   magnetic quiver. 
    To see how this comes about  note first  that  $N_j'=0$ in \eqref{426m}
    implies that  $j$ is an abelian  balanced node in 
    $ \Gamma = [j_1,j_2]\subseteq {\cal B}_\xi$. Now an  abelian balanced node has exactly two fundamental
    hypermultiplets, so it is necessarily one of the following four types:
    \vskip 7mm

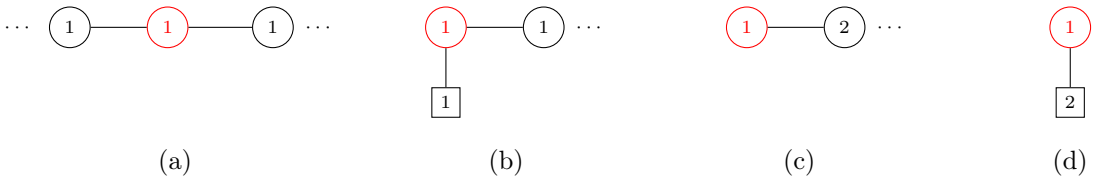
  \begin{center}
   \begin{tikzpicture}
   \begin{scope}
    \node(Ndots) at (-1,0) {\scriptsize ${\cdots}$};
     \node(N1) [circle,draw] at (-.3,0) {\scriptsize $1$};
    \node(N2) [circle,draw,red] at (1,0) {\color{red} {\scriptsize $1$}};
    \node(Nk) [circle,draw] at (2.4,0) {\scriptsize $1$};
       \node(Ndots) at (3.0,0) {\scriptsize ${\cdots}$};
    \draw (N1) -- (N2) --  (Nk)  ;
    \node at (1.1,-1.8){(a)};
  \end{scope}
    \begin{scope}[xshift={5cm}]
    \node(N1)   [circle,draw,red] at (-.3,0)  {\color{red}   {\scriptsize $1$}  }  ;
    \node(N2) [circle,draw] at (1,0) {\scriptsize $1$};
    \node(Ndots) at (1.6,0) {\scriptsize ${\cdots}$};
    \node(M1) [rectangle,draw] at (-.3,-1) {\scriptsize $1$};
    \draw (N1) -- (N2) ; 
    \draw (M1) -- (N1);
    \node at (0.5,-1.8){(b)};
  \end{scope}
    \begin{scope}[xshift={9cm}]
    \node(N1)   [circle,draw,red] at (-.3,0)  {\color{red}   {\scriptsize $1$}  }  ;
    \node(N2) [circle,draw] at (1,0) {\scriptsize $2$};
    \node(Ndots) at (1.6,0) {\scriptsize ${\cdots}$};
    \draw (N1) -- (N2) ; 
    \node at ( .4,-1.8){(c)};
  \end{scope}
    \begin{scope}[xshift={13.3cm}]
    \node(N1)   [circle,draw,red] at (-.3,0)  {\color{red}   {\scriptsize $1$}  }  ;
     \node(M1) [rectangle,draw] at (-.3,-1) {\scriptsize $2$};
     \draw (M1) -- (N1);
    \node at (-0.3,-1.8){(d)};
  \end{scope}
  \end{tikzpicture}
\end{center} 
The balanced node is drawn in red, and the dots indicate that the [good] quiver extends 
beyond the piece shown in the figure, with extra flavour and/or gauge nodes.  The   set ${\cal B}_\xi$ 
may contain several balanced nodes,
as many as the rank of the corresponding non-abelian factor of the magnetic-flavour symmetry. 
Notice however that abelian nodes of type (c) cannot coexist in the same ${\cal B}_\xi$ with abelian nodes
of the other types.  So we split the calculation of the $\Delta({\bf m})=1$ sector according to whether
${\cal B}_\xi$ contains abelian nodes of type  (a) and/or (b),  or  nodes of type  (c). The case (d) corresponds
to a single theory called   $T[SU(2)]$ and will be treated separately. 
\smallskip

   Replacing   
    $\delta_{N_j'\not=0} $ by $ 1 - \delta_{N_j' =0}$ in the last line of \eqref{426m}  and doing the straightforward
    algebra  leads to  the following result 
for the $x_+^2x_-^2$ piece: 
\begin{equation}\label{429c}
     \begin{aligned}       
      \sum_{j=1}^k \delta_{N_j'\not=0} \overline{\square}^\mu_j \square^\mu_j 
      +  \sum_{j=1}^{k-1}  \delta_{N_j^\prime\not=0}\delta_{N_{j+1}^\prime\not=0} -
     \sum_{j=1}^k  \delta_{N_j^\prime \not=0} - 1\   
       = \   \sum_{j=1}^k \overline{\square}^\mu_j \square^\mu_j  -  2 -
       \#\bigl\{ {\cal B}_\xi \bigm| \text{types (a)\&(b)} \bigr\}
       % \begin{cases}    1  \ \ &{\rm (a)+(b)}  \\
       %  0\ \ &({\rm c}) \end{cases}
              \end{aligned}
\end{equation}      
    The first two terms on the right-hand  side  were already accounted for  in
    \eqref{mixs}. 
   The extra subtraction vanishes for each ${\cal B}_\xi$ of type (c),
   and equals $-1$ for each ${\cal B}_\xi$ whose nodes are of type (a) and/or (b).
   This is  precisely what one   expects from mirror symmetry.
    Indeed, as shown in \autoref{sec:A}, the  two cases  in eq.\,\eqref{429c}
    correspond to the $\hat M_{\xi} = \vert {\cal B}_\xi\vert + 1$ 
   magnetic flavours  being charged under a  non-abelian, respectively abelian gauge group 
   in the magnetic quiver ($\hat N_{\xi} > 1$,  respectively  $\hat N_{\xi} =1$).  
   In the first case there is no correction to \eqref{mixs},  while in the second  summing over all 
   monopole-charge assignements  in   ${\cal B}_\xi$ reconstructs,  up to a fugacity-independent term
    equal to the rank, 
   the  adjoint character of the non-abelian magnetic-flavour symmetry.
   \smallskip
   
             Putting  everything together   we finally get the following for all linear quivers except $T[SU(2)]$.

             \vspace{.6\baselineskip}
             
             \noindent\underline{\bf Arbitrary quivers except \boldmath$T[SU(2)]$}:\vspace{-.6\baselineskip}
             \begin{equation} \label{430}  
  \begin{aligned}
    \mathcal{Z}^{{\rm mixed}} \    
   & =  x_+^2 x_-^2\,  \Bigl[   \bigl(\AdjG (\mu)  - 1\bigr)\bigl( 
   \chi_{{\rm Adj}\,\hat G}(\hat\mu) - 1\bigr) - 2 \\
    & 
    \qquad\qquad\quad
     - \sum_{j\mid N_j=1} \abs{\square_j^\mu}^2 - \sum_{\hat\jmath\mid \hat N_{\hat\jmath}=1} \abs{\square_{\hat\jmath}^{\hat\mu}}^2  +  \Delta n_{\rm mixed}  \Bigr] \,  + \, O(x^5) \  ,
    \end{aligned}
\end{equation}      
  where the  fugacity-independent correction  reads
\bee\label{430a} 
\Delta n_{\rm mixed} =    \sum_{\hat\jmath\mid \hat N_{\hat\jmath}=1}   \hskip -1mm  \hat M_{\hat\jmath} \ + 
     \delta_{N_1=1} + \delta_{N_k=1}
    - \sum_{j=1}^k \delta_{N_j=1}
   \ . 
 \eeq
 We show in \autoref{lem:Deltanmixed} that $\Delta n_{\rm mixed}$ is  
 (like the rest of the expression)  mirror symmetric, albeit
 not manifestly so. 
 
    For completeness we  give finally   the result  for  $T[SU(2)]$, the  theory  described by the quiver (d).   
    This is a self-dual  abelian theory   with  global  symmetry   $SU(2)\times {S
    {\widehat U(2})}$.   In  self-explanatory  notation  the result for  this case   reads

  \vspace{.6\baselineskip}
  
  \noindent\underline{\bf\boldmath$T[SU(2)]$}:\vspace{-1.1\baselineskip}
 \bee\label{432}
  \mathcal{Z}^{{\rm mixed}}  \    
   =  x_+^2 x_-^2\,  (-3 - \mu - \mu^{-1} - \hat\mu - \hat\mu^{-1})  \,  + \, O(x^5)  \ . 
 \eeq
It turns out that  for this theory the full superconformal index can be expressed  in closed form,  in terms of the
$q$-hypergeometric function. This  renders manifest a general property of the  index,  
its factorization in holomorphic blocks~\cite{Pasquetti:2011fj,Dimofte:2011py,Beem:2012mb}. Since we are
not using this feature in our  paper,   
 the calculation is relegated to  \autoref{app:C}. 
\smallskip

This completes our calculation of the mixed quartic terms of  the superconformal  index.   
We will next rewrite the index  as a sum of characters of  $OSp(4\vert 4)$  and interpret the result.

 %%%%%%%%%%%%%%%%%%% 
     
   \section{Counting   the \texorpdfstring{${\cal N}=2$}{N=2}   moduli}\label{sec:5}
 
 The  full superconformal index up to order $O(q) \sim O(x^4)$  is given by   \eqref{breakup}
 together with  expressions \,\eqref{414}-\eqref{414a} for the Higgs branch Hilbert series, their
 mirrors  for the Coulomb branch Hilbert series, and expressions \eqref{430}-\eqref{430a} for the mixed term.
 Collecting  everything   and using also  
 \eqref{indexlist} for the indices  of individual representations of the superconformal algebra
 $OSp(4\vert 4)$ leads  to the main result of this paper

\begin{equation}\label{5c}
     \begin{aligned}     
\mathcal{Z}_{S^{2}\times S^{1}}  =&  1 + \,    \underbrace{ x_+^2(1  - x_-^2)}_{{\cal I}_{B_1[0]^{(1,0)}}}\,
 \AdjG \,  +  \,    \underbrace{ x_-^2(1  - x_+^2)}_{{\cal I}_{B_1[0]^{(0,1)}}}\,
 \chi_{{\rm Adj}\,\hat G} 
 +\hskip -2mm \underbrace{ x_+^3}_{{\cal I}_{B_1[0]^{({3/ 2},0)}}}\hskip -3mm \chi_{\ell=3}^{} 
 \,+\hskip -2mm
 \underbrace{ x_-^3 }_{{\cal I}_{B_1[0]^{(0,{3/ 2} )}}} \hskip -3mm \hat\chi_{\ell=3}^{}
 \\[-1pt]
 &  
 \ +   \underbrace{ x_+^4 }_{{\cal I}_{B_1[0]^{(2,0)}}}\hskip -2mm
   \boxed{ \bigl(\sym^2  \AdjG  + \chi_{\ell=4}^{}
   -  \Delta\chi^{(2,0)}\bigr) }  \,+ 
  \underbrace{ x_-^4 }_{{\cal I}_{B_1[0]^{(0,2)}}} \hskip -2mm
   \boxed{ \bigl(\sym^2  \chi_{{\rm Adj}\,\hat G} + 
   \hat \chi_{\ell=4}^{} - \Delta\chi^{(0, 2)}\bigr)}
  \\[-1pt]
  &\qquad\qquad +\,  \underbrace{x_+^2x_-^2 }_{{\cal I}_{B_1[0]^{(1,1)}}} 
  \boxed{ \bigl( \AdjG \, 
  \chi_{{\rm Adj}\,\hat G}  -  \Delta\chi^{(1,1)} \bigr)  }
  \, + \,   \underbrace{(-x_+^2x_-^2)}_{{\cal I}_{A_2[0]^{(0,0)}}} \
  \, +\ O(x^5)\ . 
   \end{aligned}
\end{equation} 
where $\chi_{\ell=n}^{}$ 
counts independent single strings of length $n=3,4$ on the electric quiver, as in~\eqref{414},\footnote{Some quivers such as $T[SU(N)]$ have $\chi_{\ell=4}<0$: double-string operators then obey extra $F$-term relations.}
\begin{equation}
  \begin{aligned}
    \chi_{\ell=3} & = \sum_{j=1}^{k-1} \bigl( \square^\mu_j \overline{\square}^\mu_{j+1} + \overline{\square}^\mu_j \square^\mu_{j+1} \bigr) , \\
    \chi_{\ell=4} & = \begin{cases}
      0 \quad \text{for } k = 1 \text{, and otherwise} \\
      \displaystyle \sum_{j=2}^{k-1} \Bigl(\square^\mu_{j-1} \overline{\square}^\mu_{j+1} + \overline{\square}^\mu_{j-1} \square^\mu_{j+1}
      + \delta_{N_j\neq 1} \vert \square^\mu_j   \vert^2 \Bigr)  - 1 - \sum_{j=2}^{k-1} \delta_{N_j =1} + \sum_{j=1}^{k-1} \delta_{N_j =N_{j+1} =1} ,
    \end{cases}
  \end{aligned}
\end{equation}
and $\hat \chi_{\ell=n}^{}$  counts likewise single strings  on the magnetic  quiver, while the correction
terms coming from  abelian (electric and magnetic) gauge nodes are given by
 \begin{equation}\label{5cc}
     \begin{aligned}  
&\Delta\chi^{(2,0)} = \delta_{k=1}\delta_{N_1\neq 1}\abs{\square_1^\mu}^2 + \sum_{j=1}^k  \delta_{N_j =1}
   \bigl|\, {\yngantisym^{\,\mu}_{\,j}} \bigr|^2\ , \qquad
      \Delta\chi^{(0,2 )} = \delta_{\hat k=1}\delta_{\hat N_1\neq 1}\abs{\square_1^{\hat\mu}}^2 + \sum_{\hat\jmath=1}^{\hat k}  \delta_{\hat N_{\hat\jmath} =1}
   \bigl|\, {\yngantisym^{\,\hat \mu}_{\,\hat\jmath}} \bigr|^2\ , \\ 
    &\quad {\rm and}\quad 
    \Delta\chi^{(1,1)}  =
    \begin{cases}
      \AdjG \, \chi_{{\rm Adj}\,\hat G} \quad \text{for } T[SU(2)] \text{, and otherwise}
      \\[2pt]
      \displaystyle \sum_{j\mid N_j=1} \abs{\square_j^\mu}^2 +  \sum_{\hat\jmath\mid \hat N_{\hat\jmath}=1} \abs{\square_{\hat\jmath}^{\hat\mu}}^2  -  \Delta n_{\rm mixed}
    \end{cases}
  \end{aligned}
\end{equation} 
with $\Delta n_{\rm mixed}$ defined in~\eqref{430a}.  Notice that we  have used in eq.\,\eqref{5c}  the fact that the SCFT has a unique energy-momentum tensor which is
part of the $ {A_2[0]^{(0,0)}}$ multiplet, and that all the other $OSp(4\vert 4)$
multiplets  can be unambiguously identified at  this order.

Finally we calculate the dimension~\eqref{dimMSCgeneral} of the conformal manifold as the number of marginal scalar operators minus the number of conserved currents with which they recombine:
\begin{equation}
  \begin{aligned}
    \dim_{\mathbb{C}} {\cal M}_{\rm SC}
    & = \bigl[
    \bigl(\sym^2  \AdjG  + \chi_{\ell=4}^{} -  \Delta\chi^{(2,0)}\bigr)
    + \bigl(\sym^2  \chi_{{\rm Adj}\,\hat G} + \hat \chi_{\ell=4}^{} - \Delta\chi^{(0, 2)}\bigr)
    \\
    & \quad + \bigl(\AdjG \, \chi_{{\rm Adj}\,\hat G}  -  \Delta\chi^{(1,1)}\bigr)
    - \AdjG - \chi_{{\rm Adj}\,\hat G} - 1 \bigr]_{\mu=\hat\mu=1},
  \end{aligned}
\end{equation}
where the three parenthesized expressions count electric, magnetic, and mixed marginal scalars while the subtracted terms correspond to the flavour symmetry $G\times\hat G\times U(1)$ of the theory.

%%%%%%%%%%%

\subsection{Examples and interpretation}\label{exs}

       The marginal ${\cal N}=2$ deformations (exactly marginal or not) are the terms  enclosed in boxes in~\eqref{5c}. 
       Those  in the second line   are 
       standard quartic superpotentials involving only  the ${\cal N}=4$ hypermultiplets 
       of the electric quiver, or only
       their  twisted cousins of the magnetic quiver. The electric superpotentials (counted in 
         the Higgs-branch Hilberts series)  are of two kinds:  (i) single strings of length 4 that transform
       in the adjoint of each gauge-group factor $U(M_j)$, or in the bifundamental of next-to-nearest
       neighbour flavour groups  $U(M_j)\times U(M_{j+2})$;  and (ii)  double-string operators in the
       $\sym^2({\rm Adj}\,G)$ representation. 
       If there are abelian gauge nodes or $k=1$ some of these operators are  absent. The same statements of course
       hold for   magnetic superpotentials and the mirror quiver.

       The more interesting  deformations, the ones made  out of  both types of hypermultiplets, 
       are    in the third line of \eqref{5c}.   For  quivers with no  abelian nodes,  these  
       mixed operators  are 
        all possible    $\vert {\rm Adj}\,G\vert \times \vert {\rm Adj}\,\hat G\vert$  
        gauge-invariant products  of two fundamental hypermultiplets  and two fundamental twisted 
        hypermultiplets\footnote{More
 precisely, all but those involving the overall combination  $\sum_j \sum_{p,\bar p}Q_j^p \tilde Q_j^{\bar p}\delta_{p \bar p}$  or its mirror. These  are the  scalar
 partners of the two missing  $U(1)$ flavour symmetries that are gauged.}
      \bee\label{53}
     {\cal O}_{j ;  \hat\jmath}^{(\bar p, r ; \bar{\hat p}, \hat r)} \, =\,  (\tilde Q_j^{\bar p}  Q_{j}^r)  \bigl(\widetilde {\hat 
     Q}_{\hat\jmath}^{\bar{\hat p}}  { \hat Q}_{\hat\jmath}^{\hat r} \bigr)  \ ,  
      \eeq
 where hats denote  the scalars of the  (twisted) hypermultiplets.

  Some of the above  operators can  be  identified with  superpotential deformations  
    involving   both  hypermultiplets and  vector multiplets. Consider,  in particular, 
   the  following  gauge-invariant chiral operators of the electric theory
     \bee\label{54}
      {\cal O}_{j ,  j^\prime}^{(\bar p, r)}\ =\  (\tilde Q_j^{\bar p}  Q_{j}^r)\,\Tr(\Phi_{j^\prime})\ , 
      \eeq 
 where $\Phi_j$ is  the ${\cal N}=2$ chiral field  in   the ${\cal N}=4$ vector multiplet  at  the $j$th 
 gauge-group   node.  
 It can be easily shown  that $\Tr(\Phi_{j })$ 
  is the  scalar superpartner  of  the $j$th topological  $U(1)$ current, 
  so that the   operators  \eqref{54}  are  
 the same as the operators \eqref{53} when these latter are  restricted   to  the Cartan subalgebra of $\hat G$.
 Similarly, projecting \eqref{53} onto the Cartan subalgebra of $G$ gives mixed superpotential deformations
 of the magnetic Lagrangian.  
     The remaining   
  $(\vert {\rm Adj}\,G\vert  -   {\rm rank}G ) \times (\vert {\rm Adj}\,\hat G\vert  -   {\rm rank}\hat G )$  
  deformations involve both charged hypermultiplets and  monopole operators 
  and have a priori no Lagrangian description.

     We can also understand why some mixed  operators are absent when the quiver has abelian nodes.
  Recall that the
   ${\cal N}=4$ superpotential   reads     
   \bee
 W = \sum_{j=1}^k   \Bigl(  
 Q_{j, j-1}\Phi_j \tilde Q_{j, j-1} + \tilde Q_{j, j+1}\,\Phi_j   Q_{j, j+1} +  \sum_{p,\bar p= 1}^{M_j} \tilde Q_{j}^{\bar p} 
 \,\Phi_j    Q_{j }^p \delta_{p\bar p} 
 \Bigr) \ , 
 \eeq
   from which one derives the following  $F$-term conditions : 
    $\tilde Q_{j}^{\bar p} \,\Phi_j = \Phi_j    Q_{j }^p =0$ for all $j, p$ and  $\bar p$. 
 Note that   $\Phi_j $ is an $N_j\times N_j$ matrix, while  $\tilde Q_{j}^{\bar p}$ and $Q_{j }^p =0$ are bra and ket
 vectors. If  (and only if)  $j$ is an abelian node, these conditions imply 
  ${\cal O}_{j ,  j}^{(\bar p, r)} = 0$ so that these operators should be subtracted. 
  This explains
 the first of the three  terms in the subtraction   $\Delta\chi^{(1,1)}$, eq.\,\eqref{5cc}. 
 The second  is likewise explained by 
  the $F$-term conditions at  abelian nodes of the magnetic quiver.  Finally  $\Delta n_{\rm mixed}$
  corrects  some  overcounting in  these abelian-node  subtractions. 

    We should stress  that the  factorization of  
    mixed marginal deformations $B_1[0]^{(1,1)}$ in terms of electric and magnetic chiral 
     multiplets need  not be a general property    
    of  all 3d ${\cal N}=4$ theories.
  As a   counterexample [that does not come from a brane construction]  consider  the $SU(3)$ gauge theory with $M_1$ hypermultiplets in the fundamental representation and $M_2$ in its symmetric square.
  This is  a good theory  for $M_1+3M_2\geq 5$ (in particular 
  $\Delta({\bf m})\geq 1$ for ${\bf m}\neq 0$).  For $M_1+3M_2\geq 6$ it has no magnetic flavour symmetry, yet there are mixed marginal deformations in the $M_1\overline{M_2}+\overline{M_1}M_2$ representation of the electric flavour symmetry $U(M_1)\times U(M_2)$.  
 Even in  3d ${\cal N}=4$ theories that do arise from brane constructions,  
  complicated $(p,q)$-string webs with both F-string and D-string open ends, corresponding to   $B_1[0]^{(j^H,j^C)}$ multiplets, need not  factorize into F-string and D-string parts.  However, we expect this failure   to  appear if at all  at large $j^H,j^C$.

  \smallskip
  
  We may summarize the discussion   as follows:
 \medskip
 
      \fbox{%
    \parbox{13.8cm}{%
         Marginal  chiral operators of $T_\rho^{\hat\rho}[SU(N)]$ 
           transform  in the $\sym^2({\rm Adj}\,G + {\rm Adj}\,\hat G)$
 representation of the electric and magnetic  flavour symmetry, plus  
 strings of length 4  (in either adjoints or bifundamentals of individual factors),   
 modulo redundancies for  quivers with  abelian nodes  and in the special cases
 $k=1$ or $\hat k=1$.}%
}
 
 \bigskip
 
 Note that  the above  logic could be extended to 
    chiral operators of  arbitrary dimension  $\Delta = n$.   
     Operator overcounting    arises,  however,  in this case   at  electric or magnetic
     gauge nodes of  rank $\leq  n-1$,  making the combinatorial problem considerably harder. 
 \bigskip
 
  We  now  illustrate these results with selected  examples:   
 \medskip\smallskip
 
 \textbf{sQCD}$_{\mathbf{3}}$: The electric theory has gauge group  $U(N_c)$  with $N_c\geq 2$,  and 
   $N_f \geq 2N_c$ fundamental flavours. Its electric and magnetic quivers are  drawn
      below.  The magnetic quiver  with   $N_f =2N_c$   (upper right figure)  differs
 from the one  for $N_f >2N_c$   (lower right  figure). 
  Both have  $N_f-1$ balanced nodes, corresponding to the
  electric  $SU(N_f)$ flavour symmetry, but their magnetic symmetry is,  respectively,   $SU(2)$ and  $U(1)$:
 
 
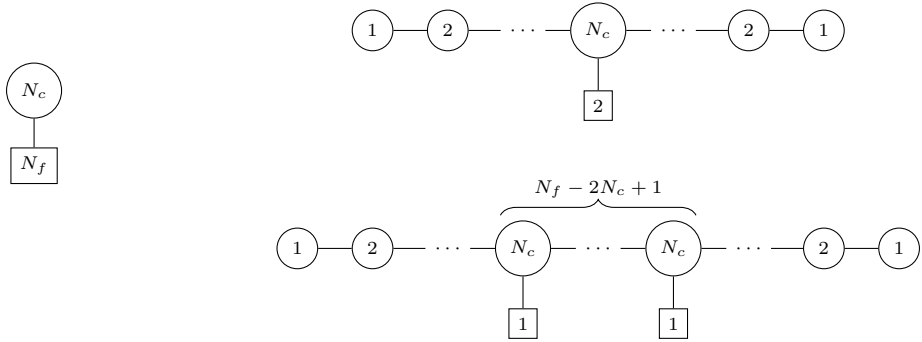
\begin{center}
   \begin{tikzpicture}
   \begin{scope}
    \node(N1)   [circle,draw] at (-.5,-.8 )  {  {\scriptsize $N_c$}  }  ;
     \node(M1) [rectangle,draw] at (-.5,-1.8 ) {\scriptsize $N_f$};
     \draw (M1) -- (N1);
  \end{scope}
    \begin{scope}
    \node(N4) [circle,draw] at (4,0) {\scriptsize $1$};
    \node(N5) [circle,draw] at (5,0) {\scriptsize $2$};
    \node(Ndots)  at (6,0) {\scriptsize ${\cdots}$};
    \node(N6) [circle,draw] at (7,0) {\scriptsize $N_c$};
    \node(Ndotss)  at (8,0) {\scriptsize ${\cdots}$};
    \node(N7) [circle,draw] at (9,0) {\scriptsize $2$};
    \node(N8) [circle,draw] at (10,0) {\scriptsize $1$};
    \node(M6) [rectangle,draw] at (7,-1) {\scriptsize $2$};
        \draw (N4) -- (N5) ;
     \draw (N5) -- (Ndots) ;
     \draw (Ndots) -- (N6);
     \draw (N6) -- (Ndotss);
     \draw (Ndotss) -- (N7);
     \draw (N7) -- (N8);
    \draw (M6) -- (N6);  
  \end{scope}
     \begin{scope}[yshift=-.6cm]
     \node(N4) [circle,draw] at (3,-2.3) {\scriptsize $1$};
    \node(N5) [circle,draw] at (4,-2.3) {\scriptsize $2$};
    \node(Ndots)  at (5,-2.3) {\scriptsize ${\cdots}$};
    \node(N6) [circle,draw] at (6. ,-2.3) {\scriptsize $  N_c $};
    \node(Ndotsmiddle) at (7,-2.3) {\scriptsize ${\cdots}$};
    \node(N6a) [circle,draw] at (8,-2.3) {\scriptsize $ N_c $};

    \node(Ndotss)  at (9,-2.3) {\scriptsize ${\cdots}$};
    \node(N7) [circle,draw] at (10,-2.3) {\scriptsize $2$};
    \node(N8) [circle,draw] at (11,-2.3) {\scriptsize $1$};
    \node(M6) [rectangle,draw] at (6. ,-3.3) {\scriptsize $1$};
     \node(M6a) [rectangle,draw] at (8,-3.3) {\scriptsize $1$};
        \draw (N4) -- (N5) ;
     \draw (N5) -- (Ndots) ;
     \draw (Ndots) -- (N6);
     \draw(N6) -- (Ndotsmiddle) ;
     \draw(Ndotsmiddle) -- (N6a) ;
     \draw (N6a) -- (Ndotss);
     \draw (Ndotss) -- (N7);
     \draw (N7) -- (N8);
    \draw (M6) -- (N6);
      \draw (M6a) -- (N6a);   
     \draw [decorate,decoration={brace,amplitude=6pt},yshift=4pt]
          (5.7,-2) -- (8.3, -2) node [midway,yshift=10pt] {\scriptsize $N_f-2N_c+1$};
      \end{scope}
   \end{tikzpicture}
  \end{center}
 
 \smallskip
\noindent  The $N_f>2N_c$ theories have    ${1\over 2} N_f^2 (N_f^2-3)$ electric,  one
  magnetic,   and  $(N_f^2-1)$ mixed marginal operators from 2-string states.
  For $N_f=2N_c+2$ there are three extra marginal operators from length~$4$ magnetic strings, while for other $N_f$ there is only one.
  There are none from length~$4$ electric strings,
  and no abelian-node redundancies.   The number of D-term conditions is  $N_f^2 +1$,  so that  the
  complex  dimension
  of the superconformal manifold is  
  $\dim{\cal M}_{SC} ={1\over 2} N_f^2 (N_f^2-3)$ if $N_f\neq 2N_c+2$, and $\dim{\cal M}_{SC} ={1\over 2} N_f^2 (N_f^2-3) + 2$ otherwise.
  When  $N_f = 2N_c$   the number of electric operators is the same,  but 
  there are  now  six 2-string magnetic operators,  $ 3(N_f^2-1) $  mixed   ones, three  length-4 strings,  and 
 $N_f^2 +3$  D-term conditions, hence $\dim{\cal M}_{SC} ={1\over 2} N_f^2 (N_f^2+1)+3$.

  \medskip\smallskip
 
 ${\mathbf{sQED_3}}$: This is a $U(1)$ theory with $N_f> 2$ charged hypermultiplets. The magnetic quiver
 has $N_f-1$ abelian balanced nodes and one charged hypermultiplet at each  end of the chain:
 \medskip
 
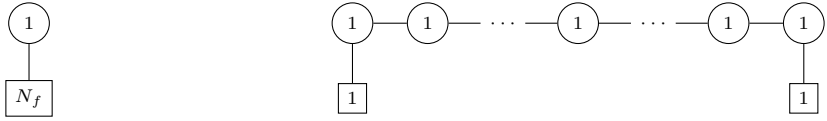
\begin{center}
   \begin{tikzpicture}
   \begin{scope}
    \node(N1)   [circle,draw] at (-.3,0)  {  {\scriptsize $1$}  }  ;
     \node(M1) [rectangle,draw] at (-.3,-1) {\scriptsize $N_f$};
     \draw (M1) -- (N1);
  \end{scope}
    \begin{scope}
        \node(M1) [rectangle,draw] at (4,-1) {\scriptsize $1$};
     \node(N4) [circle,draw] at (4,0) {\scriptsize $1$};
    \node(N5) [circle,draw] at (5,0) {\scriptsize $1$};
    \node(Ndots)  at (6,0) {\scriptsize ${\cdots}$};
    \node(N6) [circle,draw] at (7,0) {\scriptsize $1$};
    \node(Ndotss)  at (8,0) {\scriptsize ${\cdots}$};
    \node(N7) [circle,draw] at (9,0) {\scriptsize $1$};
    \node(N8) [circle,draw] at (10,0) {\scriptsize $1$};
     \node(M2) [rectangle,draw] at (10,-1) {\scriptsize $1$};
        \draw (N4) -- (N5) ;
     \draw (N5) -- (Ndots) ;
     \draw (Ndots) -- (N6);
     \draw (N6) -- (Ndotss);
     \draw (Ndotss) -- (N7);
     \draw (N7) -- (N8);
    \draw (M1) -- (N4);
    \draw (M2) -- (N8);
  \end{scope}
   \end{tikzpicture}
  \end{center}
\smallskip

 \noindent   This theory has ${1\over 4} N_f^2 (N_f+1)^2-N_f^2$ marginal electric operators
 (because   the antisymmetric combination $Q^{[p}Q^{r]}$ vanishes),   one magnetic operator, and
 no  mixed ones. To prove  this latter assertion one computes $\Delta n_{\rm mixed} = 3$
 from eq.\,\eqref{430a} [checking in passing that the expression  is mirror symmetric]. 
 In the special case  $N_f=4$ there is in addition two length-4 magnetic strings. Note that   for $N_f\gg 1$
  the dimension of the  superconformal manifold of sQED$_3$   is reduced by 
  a factor two compared to the superconformal manifold  of  sQCD$_3$.

  \medskip\smallskip
 
 ${\mathbf{T[SU(N)]}}$:  This theory is defined by 
  the self-dual fully-balanced quiver  shown below. 
 
  
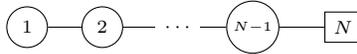
\begin{center}
   \begin{tikzpicture}
    \node(N1) [circle,draw] at (0,0) {\scriptsize $1$};
    \node(N2) [circle,draw] at (1,0) {\scriptsize $2$};
    \node(Ndots)  at (2,0) {\scriptsize ${\cdots}$};
    \node(N3) [circle,draw] at (3,0) {\tiny $\!\!N{-}1\!\!$};
    \node(M3) [rectangle,draw] at (4.2,0) {\scriptsize $N$};
        \draw (N1) -- (N2) ;
     \draw (N2) -- (Ndots) ;
     \draw (Ndots) -- (N3);
    \draw (N3) -- (M3);
  \end{tikzpicture}
  \end{center}
 
 \noindent For $N\geq 3$ there are  ${1\over 2}N^2(N^2-1)-1$ electric operators, as many  magnetic
 operators,  and $(N^2-1)^2$ mixed ones.  The dimension of the superconformal manifold is 
 $\dim{\cal M}_{SC}= N^2(2N^2-5)$.  The case $T[SU(2)]$ was discussed already separately.

 %%%%%%%%%%
 %%%%%%%%%%%

 \subsection{The holographic perspective}\label{holo}
 
          In this last part  we discuss  the relation to string theory and sketch some directions for future work.

  Recall that   the $T_\rho^{\hat\rho}[SU(N)]$ theories are  holographically   
  dual to type IIB string theory in 
  the supersymmetric backgrounds  of refs.\,\cite{Assel:2011xz,Assel:2012cj}.   
  The geometry has a AdS$_4\times$S$^2_H\times$S$^2_C$ fiber over a basis which is  the infinite 
   strip $\Sigma$. The $SO(2,3)\times SO(3)_H\times SO(3)_C$ symmetry of the SCFT is realized as
   isometries of the fiber.  The solution features   singularities on the upper  (lower) boundary of the
   strip which correspond to  D5-brane sources wrapping S$^2_H$ (NS5-brane sources wrapping S$^2_C$). 
  These two-spheres are trivial in homology, yet  the 
  branes  are  stabilized by non-zero worldvolume fluxes that counterbalance
  the negative tensile stress  \cite{Bachas:2000ik}.  
  
    There is a total of  $k+1$ NS5-branes and $\hat k+1$ D5-branes. Their position along the boundary
    of the strip is a function of their linking number,    which increases  from left to right for D5-branes  and
    decrease for NS5-branes \cite{Assel:2011xz}. Branes with the same linking number  overlap giving non-abelian
    flavour symmetries. 
  The linking number of a fivebrane can be equivalently defined as
  \begin{itemize}

\item  the D3-brane charge dissolved in the fivebrane ; 

\item the worldvolume flux on the wrapped  two-sphere;

\item the node of the corresponding quiver, for instance   the $\hat\imath$th D5-brane  provides a 
fundamental hypermultiplet  at the $\hat l_{\hat\imath} = i$ 
 node of the electric quiver (see \autoref{sec:A}).
  
 \end{itemize}
  The $R$-symmetry spins $J^H, J^C$ are the angular momenta of a state on the two spheres.
  Given the above dictionary, can we understand the results of this paper from the string-theory side?
 
 \smallskip
 
 Consider first the Higgs-branch chiral ring which consists of the highest weights of all $B_1[0]^{(j^H, 0)}$
 multiplets. When decomposed in terms of conformal primaries these multiplets read \cite{Cordova:2016emh}
\bee
 B_1[0]^{(j^H, 0)}_{j^H}\ =  
  [0]^{(j^H, 0)}_{j^H} \,\oplus\, [0]^{(j^H-1, 1)}_{j^H+1} \,\oplus\, [1]^{(j^H-1, 0)}_{j^H+1}\,\oplus\,
  {\rm fermions}_{j^H+{1\over 2}}\ . 
 \eeq 
 Note  that the  top component includes a vector boson  with scaling   dimension
 $\Delta = j^H+1$. This is a massless gauge boson in AdS$_4$ for  $j^H=1$ (`conserved current' multiplet)
 and a massive gauge boson for $j^H>1$. As explained in ref.\,\cite{Bachas:2017wva},
 both massless and massive  vector bosons are states of  fundamental open strings on the D5-branes.
 Their vertex operators   include a scalar
 wavefunction on  S$^2_H$ with angular momentum $J^H = j^H-1$.   Consider  such an open string
 stretching between two D5-branes with linking numbers $\ell$ and $\ell^\prime$.  Since these latter  are
 magnetic-monopole fields on S$^2_H$,  the open string couples to a net field $(\ell - \ell^\prime)$.
 Its  wavefunction is therefore given by the well-known monopole spherical harmonics with\footnote{This
 celebrated result goes back to the early days of quantum mechanics  \cite{Tamm:1931dda}. We  have used it
  implicitly   when expressing determinants as  $q$-Pochammer symbols. For an  amusing real-time manifestation of the effect  see
 \cite{Bachas:2009ve}\,.}
 \bee\label{57a}
 j^H - 1 = {1\over 2} \vert \ell - \ell^\prime \vert + \mathbb{N}
 \eeq
 where $\mathbb{N}$ are  the natural numbers.  Recalling that the linking numbers also designate the
 nodes of the electric quiver, we understand   
  why the Higgs-branch chiral ring includes   strings of  minimal
 length $\vert \ell - \ell^\prime \vert + 2$ 
  transforming in the bi-fundamental of 
  $U(M_\ell) \times U(M_{\ell^\prime})$ for all $k \geq \ell^\prime > \ell >  0$ \cite{Bachas:2017wva}. The 
  bifundamental strings of length 3 and 4 in eq.\,\eqref{414} are of this kind.  
 \smallskip
 
      The $\Delta = 2$ chiral ring  also includes  strings of length 4  in the adjoint   
      of $U(M_j)$ for all $k>  j  > 1$,  see \eqref{414}. The  corresponding  open-string  vector bosons  
      on the $i$th stack of 
       D5-branes  do not feel a monopole field  
      ($\ell = \ell^\prime =i$) but have  angular momentum
      $j^H-1 = 1$. Notice however that   these length-4 operators  
     are missing
      at  the two ends of the quiver, i.e.\  for $i  = 1$ and  for $i= k$. 
   How can one  understand  this from the string theory side?

  A plausible explanation comes from a  well-known  
   effect   dubbed `stringy exclusion principle' 
  in ref.\,\cite{Maldacena:1998bw}. The relevant setup features $K$ NS5-branes  and a set of probe D-branes
  ending on them.  The worldsheet theory in this background has an affine algebra  
  $\mathfrak{su}(2)_K$\footnote{The bosonic subalgebra has level $K-2$ and an  extra factor 2
  is added by fermions.} and D-branes (Cardy states) labelled by the set of dominant affine weights
  $\lambda = 0, 1, \cdots, K-1$. The ground states of open strings 
  stretched between two such D-branes  have   weights $\nu$  in the interval  $$\bigl[\,\vert \lambda - \lambda^\prime\vert,
   {\rm min}(\lambda + \lambda^\prime,  2K - \lambda - \lambda^\prime)\,\bigr]$$ and in  steps of two
  \cite{DiFrancesco:1997nk}.   Translating   $\lambda = \ell-1$ (see \cite{Bachas:2009ve}),   $\mu = 2(j^H-1)$ and  $K= k+1$
  (the total number of NS5-branes) gives in replacement of \eqref{57a}
  \bee\label{58a}
   j^H - 1\ =\  {1\over 2} \vert \ell - \ell^\prime \vert  \,, \ {1\over 2} \vert \ell - \ell^\prime \vert +1\,, \ \cdots \  ,
   {\rm min}\Bigl({\ell + \ell^\prime\over 2} -1\,,  k - {\ell + \ell^\prime\over 2} \Bigr)\ . 
  \eeq
 The intuitive understanding of the  upper cutoff is that a string 
 cannot remain in its ground state if  its  angular momentum exceeds the size of the
  sphere. 
   It follows   that for  $\ell=\ell^\prime = 1$ or $k$,  only the $j^H=1$ states survive, in agreement with our
  findings for the single-string part of the Higgs-branch chiral ring.
  
     To be sure  this is just an argument, not a proof,  because
     in the solutions dual to $T_\rho^{\hat\rho}[SU(N)]$ the 3-sphere threaded by the NS5-brane flux
     is highly deformed by  the strong back reaction of the D-branes.
     The perfect match  with the field theory side suggests, however,  that 
     the detailed  geometry does not matter when it comes to  the above stringy effect.\footnote{The match between field theory and (multi-string) symmetric products of single-string states counted using the stringy exclusion principle seems to continue holding to higher orders until the occurrence of low gauge-group rank exclusion effects discussed below.}
      
 %%%%%%%%%%%%     
      
       The  superconformal index  brings to light  other exclusion effects 
       associated to abelian gauge nodes of the electric and magnetic quivers,   as summarized in
       eqs.\,\eqref{5c} and \eqref{5cc}. For higher elements of the chiral ring, these effects are
       more generally related to the finite ranks of   the gauge groups. This  is  a ubiquitous phenomenon in 
       holography -- McGreevy et al  coined the name `giant graviton' for it 
       in  the prototypical AdS$_5\times$S$^5$ example \cite{McGreevy:2000cw}. 
        We did not manage  to  find   a simple 
           explanation for  giant-graviton exclusions in the  problem at hand.
           Part of the difficulty is that, as opposed  to  the 5-brane linking numbers, the  
            gauge group ranks  have   a less  direct  meaning  on the gravitational  side of the  AdS/CFT
            correspondence.\footnote{Note  that two theories with the 
             same flavour symmetry, i.e.\ the same disposition of five-branes,
             can have very different  gauge-group
           ranks. 
           This feature (called `fine print' in ref.\,\cite{Bachas:2017wva}) is best  illustrated by sQCD$_3$ with a
           fixed number of flavours, $N_f$,   but an arbitrary number of colors
           $N_c \in (2, [N_f/2]-1)$, see \autoref{exs}. }
  \smallskip    
   %%%%%%%%%%%    
         
            We conclude our  discussion of the AdS side with a remark about  gauged  ${\cal N}=4$ supergravity. 
     In addition to the graviton,  this  has  $n$ vector multiplets and global 
       $SL(2)\times SO(6,n)$ symmetry,   part of which may be  gauged. Insisting that the gauged theory
       have  a supersymmetric  AdS$_4$ vacuum
    restricts the form of the gauge   group to
 be  $G_H\times G_C\times G_0\subset SO(6,n)$, where the (generally) 
 non-compact  
 $G_H$ and $G_C$  contain the  $R$-symmetries
 $SO(3)_H$ and  $SO(3)_C$ \cite{Louis:2014gxa}.
 
   The vector bosons of   spontaneously-broken gauge symmetries
 belong to $B$-type multiplets with $(j^H, j^C)= (2,0)$ or $(0,2)$. These can describe 
  the length-4 
 marginal operators in  the Higgs-branch  or Coulomb-branch chiral rings.  
  As noted on the other hand  in ref.\,\cite{Bachas:2017wva},   there is no room for 
   elementary  $ (1,1) $ multiplets in ${\cal N}=4$ supergravity, because such
 multiplets  have  extra 
   spin-${3\over 2}$ fields.
       But we have  just seen   that linear-quiver theories have no single-string $(1,1)$
   operators, so the above limitation does not apply. All mixed marginal deformations  correspond to 
     double-string
           operators that
             can be described effectively by modifying the  boundary conditions
          of   their single-string constituents \cite{Witten:2001ua,Berkooz:2002ug}. 
          Note that boundary conditions  change  the quantization, not the solution.  So
          
     \bigskip
           
     \fbox{%
    \parbox{14.4cm}{%
         Gauged 
           ${\cal N}=4$ supergravity has  the  necessary   ingredients to describe the complete  moduli space
           of  the $T_\rho^{\hat\rho}[SU(N)]$ theories, provided one considers both classical and quantization moduli.
   }%
}         
     
            \bigskip
            
          \noindent          This quells, at least for linear quivers,  the concern  raised in  \cite{Bachas:2017wva} 
           that reduction of string theory
          to gauged 4d supergravity  may truncate  away  part of the moduli space. 
          As   pointed out,  however, recently  by one of us \cite{Bachas:2019rfq} such
          quantization moduli of  gauged supergravity can be singular in the full-fledged ten-dimensional string theory.

 %%%%%%%%%%  
 %%%%%%%%%
 
 \subsection{One last  comment}
 
      We end  with a remark about  the Hilbert series of $T_\rho^{\hat\rho}[SU(N)]$ theories. 
  As we explained in  \autoref{sec3},  the full chiral ring  consists of  the highest-weights of all 
  $B$-type  multiplets in the theory  with arbitrary  $(j^H, j^C)$. 
  The relevant and  marginal operators can be identified unambiguously in  the index, 
  as can the entire  Higgs-branch and Coulomb-branch subrings.  But
   general mixed elements (with $j^H, j^C\geq 1$ not both~$1$) cannot be extracted unambiguously.  
   A calculation that does not rely on the superconformal  index 
    would 
   therefore be of great interest.

   A  natural conjecture  
   for the   full Hilbert series \cite{Carta:2016fjb}
    is  that it is the coordinate ring of the union of all branches~$B_\sigma$ (for the $T_\rho^{\hat\rho}$ theory, $\sigma$ ranges over partitions between $\rho$ and $\hat\rho^T$),
\begin{equation}
 \operatorname
  {HS}\biggl(\bigcup_\sigma B_\sigma\biggm|x_+,x_-\biggr)
  = \sum_\Lambda (-1)^{|\Lambda|-1}   \operatorname{HS}\biggl(\bigcap_{\sigma\in\Lambda} B_\sigma\biggm|x_+,x_-\biggr)
\end{equation}
where $\Lambda$ runs over all non-empty subsets of the branches of the theory.
In words, the full Hilbert series would be the sum of Hilbert series of every branch, minus corrections due to pairwise intersections and so on.
It can be  checked  that this conjecture is consistent with the Higgs branch and Coulomb branch limits ($q^{1/4}t^{\mp 1/2}\to 0$ with $q^{1/4}t^{\pm 1/2}$ fixed).
One can also compare the   number of $B_1[0]^{(1,1)}$ multiplets suggested by this conjecture
to the number extracted from the index.
In the limited set of examples that we checked (with zero or one mixed branch)  we found an exact match.
Finding a better way  to confirm or falsify this  conjecture is an interesting problem.
 \vskip  6mm

 {\bf Acknowledgements:}   We thank  Benjamin Assel, Santiago Cabrera, 
 Ken Intriligator, Noppadol Mekareeya,  Shlomo Razamat  and Alberto Zaffaroni
  for useful  discussions and  for   correspondence. 
  We are particularly indebted to  Amihay Hanany for his many  patient explanations
 of   aspects of  3d   quiver gauge theories.  CB gratefully   acknowledges the hospitality of 
  the String Theory group at  Imperial College  where part of this work was done.

\appendix

\section{Index and  plethystic exponentials}\label{appB}
\setcounter{equation}{0}

The twisted partition function on $S^2\times S^1$ of  the  $T^{\hat \rho}_{{\rho}}$ theory
is  given by a multiple sum over monopole charges  and a multiple  integral over
 gauge fugacities,  see e.g.~\cite{Willett:2016adv}
 \begin{equation}\label{B1eq}
  \begin{aligned}
   \mathcal{Z}_{S^2\times S^1}
    & = \prod_{j=1}^k\Biggl[ \frac{1}{N_j!} \sum_{m_j\in\ZZ^{N_j}} \int \prod_{\alpha=1}^{N_{j}}\frac{dz_{j,\alpha}}{2\pi i z_{j,\alpha}} \Biggr] \Biggl\{
    \prod_{j=1}^k \prod_{\alpha=1}^{N_j} w_j^{m_{j,\alpha}}
     Z_{j,\alpha}^{\text{vec},\text{diag}} \\
    & \qquad \prod_{j=1}^k \prod_{\alpha\neq\beta}^{N_j} Z_{j,\alpha,\beta}^{\text{vec},\text{off-diag}}
    \prod_{j=1}^k \prod_{p=1}^{M_j} \prod_{\alpha=1}^{N_j} Z_{j,p,\alpha}^{\text{fund},\text{hyp}}
    \prod_{j=1}^{k-1} \prod_{\alpha=1}^{N_j}\prod_{\beta=1}^{N_{j+1}}    Z_{j,\alpha,j+1,\beta}^{\text{bifund},\text{hyp}}
    \Biggr\}
  \end{aligned}
\end{equation}
where 
\begin{subequations} \label{Zs}
\begin{align}
  Z_{j,\alpha}^{\text{vec},\text{diag}}
  & =  
    \frac{(q^{\frac{1}{2}}t;q)_{\infty}}{(q^{\frac{1}{2}}t^{-1};q)_{\infty}}
  \\
  Z_{j,\alpha,\beta}^{\text{vec},\text{off-diag}}
  & =
    \begin{aligned}[t]
      &   (q^{\frac{1}{2}}t^{-1})^{- \frac{1}{2}\abs{ m_{j,\alpha}-m_{j,\beta}}}
      (1-q^{\frac{1}{2}\abs{ m_{j,\alpha}-m_{j,\beta}}} 
      z_{j,\beta}z_{j,\alpha}^{-1}   )   \\
      & \qquad  \times \frac{(tq^{\frac{1}{2}+\abs{ m_{j,\alpha}-m_{j,\beta}}}
       z_{j,\beta}z_{j,\alpha}^{-1} \,;\, q)_{\infty}}{(t^{-1}q^{\frac{1}{2}+\abs{ m_{j,\alpha}-m_{j,\beta}}} z_{j,\beta}z_{j,\alpha}^{-1} \,;\, q)_{\infty}}
    \end{aligned}
  \\
  Z_{j,p,\alpha}^{\text{fund},\text{hyp}}
  & =    (q^{\frac{1}{2}}t^{-1})^{\frac{1}{2}\abs{ m_{j,\alpha}}}
  \frac{(t^{-\frac{1}{2}}q^{\frac{3}{4}+\frac{1}{2}\abs{ m_{j,\alpha}}}
   z_{j,\alpha}^{\pm 1}\mu_{j,p}^{\mp 1}
   \,;\, q)_{\infty}}{(t^{\frac{1}{2}}q^{\frac{1}{4}+\frac{1}{2}\abs{ m_{j,\alpha}}}
    z_{j,\alpha}^{\mp 1}\mu_{j,p}^{\pm 1}  \, ;\, q)_{\infty}}
  \\
  Z_{j,\alpha,j+1,\beta}^{\text{bifund},\text{hyp}}
  & =     (q^{\frac{1}{2}}t^{-1})^{\frac{1}{2}\abs{m_{j,\alpha}-m_{j+1,\beta}}}
  \frac{(t^{-\frac{1}{2}}q^{\frac{3}{4}+\frac{1}{2}\abs{ m_{j,\alpha}-m_{j+1,\beta}}}
  z_{j,\alpha}^{\pm 1}z_{j+1,\beta}^{\mp 1}  \, ;\, q)_{\infty}}{(t^{\frac{1}{2}}q^{\frac{1}{4}+\frac{1}{2}\abs{ m_{j,\alpha}-m_{j+1,\beta}}}z_{j,\alpha}^{\mp 1}
  z_{j+1,\beta}^{\pm 1}     \,;\, q)_{\infty}}\ . 
\end{align}
\end{subequations}
The  expressions  \eqref{Zs} are the one-loop determinants of the   ${\cal N}=4$  multiplets of   $T^{\hat \rho}_{{\rho}}$, namely 
 the Cartan and charged  vector multiplets, and the 
 fundamental and
bifundamental  hypermultiplets.  The variables  
$q$, $t$ are the  fugacities defined in eq.\,\eqref{sci},  $z_{j, \alpha}$
(where $\alpha$ labels the Cartan generators) are
  the   $S^1$ holonomies  of the   $U(N_j)$ gauge field 
 and $m_{j, \alpha}$ its   2-sphere  fluxes, {\it viz}.  the monopole charges of the corresponding local operator in $\mathbb{R}^3$ .  
 Furthermore   $\mu_{j,p}$ are  flavor fugacities,   $w_j$ is a fugacity for the topological $U(1)$ symmetry
whose conserved current is 
$\Tr{\star F_{(j)}}$, while
 the   $q$-Pochhammer  symbols 
  $(a; q)_\infty$   are defined by 
\bee
(a; q)_\infty  = \prod_{n=0}^\infty (1 -aq^n) 
\qquad
{\rm and}\quad 
(\dots a^{\pm 1} b^{\mp 1};q)_\infty=(\dots ab^{-1};q)_\infty(\dots a^{-1}b;q)_\infty\ . 
\eeq
Compared to   the  expressions   in ref.\,\cite{Willett:2016adv}  we have here replaced  
the background flux  coupling to any given multiplet 
 by its absolute value.
This is allowed because of some cancellation between factors in the numerator and denominator Pochhammer symbols, as explicited for instance around~\eqref{absm}.
The theory is  also free from  parity anomalies, so that the 
overall signs are unambiguous.\footnote{There exists a  subtle
sign $(-)^{e\cdot m}$ related to the  change of  spin of dyonic states with charges
$(e,m)$. The $T^{\hat \rho}_{{\rho}}$ theory has no Chern-Simons terms, so the flux ground states
have no electric charge, $e$,  and contribute with plus signs to the index. For  excited states
in the flux background  this sign 
 can be absorbed in the fugacities
$z_{j,\alpha}$; it  is  in the end   irrelevant since  the $z_{j,\alpha}$
 integrations  project    to gauge-invariant states. }

\smallskip

At leading order in the $q$ expansion, 
     the contribution of each monopole sector ${\bf m} = \{m_{j, \alpha}\}$ to the 
     superconformal index    
     is   $ (q^{\frac{1}{2}}t^{-1})^{\Delta({\bf m})}$,   where
\begin{equation}
 2  \Delta({\bf m})  = \sum_{j=1}^k \sum_{\alpha,\beta=1}^{N_j} -\abs{m_{j,\alpha}-m_{j,\beta}}
  + \sum_{j=1}^k M_j \sum_{\alpha=1}^{N_j} \abs{m_{j,\alpha}}
  + \sum_{j=1}^{k-1} \sum_{\alpha=1}^{N_j} \sum_{\beta=1}^{N_{j+1}} \abs{m_{j,\alpha}-m_{j+1,\beta}} .
\end{equation}
The sphere Casimir energy $\Delta({\bf m})$ is the scaling dimension [and the $SO(3)_C$ spin]  of the 
corresponding monopole operator \cite{Borokhov:2002ib,Borokhov:2003yu}. It is known that in  ${\cal N}=4$ theories  monopole-operator
dimensions are one-loop exact,  
and that they are strictly positive
for  good  linear quivers
 \cite{Gaiotto:2008ak}.   
The index  \eqref{B1eq} admits  therefore an  expansion   in positive powers of $q$. 
\smallskip
  
 It is useful to rewrite the superconformal index   in terms of the plethystic exponential 
 (PE) which is defined,  for any function
 $f(v_1, v_2, \cdots )$ of arbitrarily many variables that vanishes at~$0$,  by the following expression
 \bee
 {\rm PE}(f) = \exp\left( \sum_{n=1}^\infty {1\over n} f(v_1^n, v_2^n, \cdots) \right)\ . 
 \eeq
The reader can verify  the following simple  identities:
\bee
        {\rm PE} (f + g) =  {\rm PE} (f) \,{\rm PE}(g)\, ,  \quad 
 {\rm PE}(- v) =  (1-v)\, , \quad
 {\rm PE}\bigl((a,q)_\infty \bigr) =  {\rm PE}\bigl( -{a\over 1-q}\bigr)\ . 
\eeq
Using these identities one  can bring the index to  the following form
\begin{equation} 
  \begin{aligned}
    \mathcal{Z}_{S^2\times S^1}
    & = \! \prod_{j=1}^k\Biggl[ \frac{1}{N_j!} \! \sum_{m_j\in\ZZ^{N_j}} \! \int \! \prod_{\alpha=1}^{N_{j}}\frac{dz_{j,\alpha}}{2\pi i z_{j,\alpha}} \Biggr] \Biggl\{ \!   (q^{\frac{1}{2}}t^{-1})^{\Delta({\bf m})} \!
     \prod_{j=1}^k      \biggl[   \prod_{\alpha}^{N_j} w_j^{ m_{j,\alpha}} \!\!
     \prod_{\alpha\neq\beta}^{N_j} (1-q^{{1\over 2} \abs{ m_{j,\alpha}-m_{j,\beta}}}z_{j,\beta}z_{j,\alpha}^{-1})
     \biggr]
    \\
    & \qquad \times {\rm PE} \Biggl( \sum_{j=1}^k \sum_{\alpha,\beta=1}^{N_j} \frac{q^{1\over 2} (t^{-1}-t)}{1-q}\,
    q^{ \abs{ m_{j,\alpha}-m_{j,\beta}}}z_{j,\beta}z_{j,\alpha}^{-1}
    \\
    & \qquad\qquad\qquad  + \frac{(q^{1\over 2} t)^{1\over 2} (1 - q^{1\over 2} t^{-1})}{1-q}
    \sum_{j=1}^k \sum_{p=1}^{M_j} \sum_{\alpha=1}^{N_j}
   q^{{1\over 2} \abs{m_{j,\alpha}}} \sum_{\pm}z_{j,\alpha}^{\mp 1}\mu_{j,p}^{\pm 1}
    \\
    & \qquad\qquad\qquad + \frac{(q^{1\over 2} t)^{1\over 2} (1 - q^{1\over 2} t^{-1})}{1-q}
    \sum_{j=1}^{k-1} \sum_{\alpha=1}^{N_j}\sum_{\beta=1}^{N_{j+1}}
    q^{{1\over 2} \abs{m_{j,\alpha}-m_{j+1,\beta}}} \sum_{\pm}z_{j,\alpha}^{\mp 1}z_{j+1,\beta}^{\pm 1}
    \Biggr)
    \Biggr\}\ . 
  \end{aligned}
\end{equation}
 This is equation \eqref{fullindexoflinearquiverQpQm} in the main text.  Notice that 
 after extracting  some  factors, the contributions
 of vector, fundamental and bifundamental multiplets add up    in the  argument
 of the plethystic exponential, as they would in  the standard  exponential function.

   The usefulness of the above rewriting can be  illustrated with a simple example, that of a
   free   hypermultiplet 
   whose superconformal index is 
\bee
  \mathcal{Z}_{S^2\times S^1}^{\text{free}\, \text{hyp}}
  =   
  \frac{(t^{-\frac{1}{2}}q^{\frac{3}{4} }
 \mu^{\mp 1}
   \,;\, q)_{\infty}}{(t^{\frac{1}{2}}q^{\frac{1}{4} }
   \mu^{\pm 1}  \, ;\, q)_{\infty}}\
   =\    {\rm PE}\left(  \frac{(q^{1\over 4}t ^{1\over 2} - q^{3\over 4}t ^{-{1\over 2}})}{1-q}(\mu + \mu^{-1})
   \right)\ .
\eeq
One recognizes in the PE exponent the contributions of the  charge-conjugate ${\cal N}=2$ chiral multiplets, 
each contributing to the index with  one scalar  
($\Delta = J_3^H = {1\over 2}$ and $J_3 = J_3^C = 0$) and one  fermionic state  (with 
$\Delta = 1$,   $ J_3^H = 0$ and $J_3 = J_3^C = {1\over 2}$). As for the factor of $(1-q)$, this  sums up   
descendant states obtained by the action of the  derivative that raises both $\Delta$ and $J_3$ by one unit. 
Multiparticle states (created by products of fields) are taken care of by
the plethystic exponential,  the information in them is in this simple case redundant.

  Of course in 
 interacting theories  supersymmetric multiparticle states  may be  null,  due  for example  to $F$-term conditions. 
The plethystic exponent must in this case be interpreted appropriately, as we discuss  in the main text. 

%%%%%%%%%%%%

\section{Combinatorics of linear quivers}\label{sec:A}

  We collect here formulae for  the different parametrizations of the discrete data  of  the good linear quivers, 
 and we establish two  lemmas  used in \autoref{44a} of the main text.
  
The mirror-symmetric parametrization of the quiver is in terms of
 two partitions $(\rho,\hat\rho)$ with an equal total number $N$ of boxes,  if these partitions are 
 viewed as Young diagrams.
We label entries of these partitions and of their transposes as
\begin{equation}
  \begin{aligned}
    \rho & = (l_1,l_2,\dots,l_{k+1}) && \text{with} \quad l_1\geq l_2\geq\dots\geq l_{k+1}\geq 1, \\
    \rho^T & = (l^T_1, l^T_2,\dots, l^T_{l_1}) && \text{with} \quad l^T_1\geq l^T_2\geq\dots\geq l^T_{l_1} \geq 1, \\
    \hat\rho & = (\hat l_1,\hat l_2,\dots,\hat l_{\hat k+1}) && \text{with} \quad \hat l_1\geq\hat l_2\geq\dots\geq\hat l_{\hat k+1}\geq 1, \\
    \hat\rho^T & = (\hat l^T_1,\hat l^T_2,\dots,\hat l^T_{\hat l_1}) && \text{with} \quad \hat l^T_1\geq\hat l^T_2\geq\dots\geq\hat l^T_{\hat l_1}\geq 1,
  \end{aligned}
\end{equation}
where we used the fact that the number of rows of   $\rho^T$ is given by the longest row $l_1$ of~$\rho$,
we denoted the number of rows of~$\rho$ as $l^T_1=k+1\geq 2$, and likewise for hatted quantities.
To simplify formulae, the sequences $(l_j)$, $(l^T_{\hat\jmath})$, $(\hat l_{\hat\jmath})$, $(\hat l^T_j)$ are extended with zeros when $j$ or $\hat\jmath$ goes beyond the last entry.
The total number of boxes is $\sum_j l_j = \sum_{\hat\jmath}l^T_{\hat\jmath} = \sum_{\hat\jmath}\hat l_{\hat\jmath} = \sum_j \hat l^T_j  = N $.
\smallskip

    In the string-theory embedding   $\rho$ and $\hat\rho$ describe how  $N$ D3-branes end on two sets of fivebranes: 
     on $k+1$ NS5-branes to  the left and on $\hat k+1$ D5-branes to  the right.\footnote{In some of the earlier
     literature, especially ref.\,\cite{Assel:2011xz}, $\rho$ designated the  partition of D3-branes 
     among  D5-branes and 
      $\hat\rho$    the  partition among  NS5-branes. Our flipped convention here is chosen so  as to 
      remove all  hats from  the data of the  electric quiver, defined as the theory whose manifest 
       flavour symmetry  is realized on D5-branes. Note in particular that in the parametrization \eqref{B2a}
       the number of same-length rows of $\hat\rho$ runs over $j= 1, \cdots , k$.
     } 
      The number of D3-branes
     ending on the $j$th NS5-brane (or   its  linking number which is invariant under
     brane moves) is $l_j$, and likewise for the hatted quantities.
     A useful alternative parametrization
      of these
     partitions is in terms of the numbers  of their same-length rows
   \begin{equation}\label{B2a}
  \begin{aligned}
    \rho & = ( \underbrace{\hat k+ \cdots +\hat k}_{\hat M_{\hat k}} \,+\, \cdots  +\, 
     \underbrace{\ell + \cdots +\ell }_{\hat M_\ell} + \cdots +\,  \underbrace{1+ \cdots +1}_{\hat M_1}
     ) \   , 
    \\
       \hat\rho & = ( \underbrace{ k+ \cdots +  k}_{  M_k} +\, \cdots +\, 
     \underbrace{\ell + \cdots +\ell }_{  M_\ell}  + \cdots +\,  \underbrace{1+ \cdots +1}_{  M_1}
     ) \ , 
  \end{aligned}
\end{equation}
where we used  the good property $\hat\rho^T > \rho$ which implies that $l_1\leq \hat k$ and 
$\hat l_1 \leq k$. 
Note that here some of the   $M_\ell$ and $\hat M_\ell$  may vanish, when there are no 
fundamental hypermultiplets
at  the corresponding gauge-group nodes.  Note also that   
  the label $\xi$ for  groups of balanced nodes in  \autoref{44a} 
  runs over stacks of NS5-branes with $\hat M_\ell >1$, i.e.\ over nodes in the
magnetic quiver with non-abelian   flavour groups.

The electric and magnetic gauge groups are $\prod_{j=1}^k U(N_j)$ and $\prod_{\hat\jmath=1}^{\hat k} U(\hat N_{\hat\jmath})$:
\begin{equation}
  \begin{aligned}
    \begin{tikzpicture}[inner sep=1pt,anchor=base]
      \node(N1)[circle,draw] at (0,0) {$N_1$};
      \node(N2)[circle,draw] at (1,0) {$N_2$};
      \node(dots) at (2,0) {${\,\cdots\,}$};
      \node(Nk)[circle,draw] at (3,0) {$N_k$};
      \node(M1)[draw] at (0,-1) {$M_1$};
      \node(M2)[draw] at (1,-1) {$M_2$};
      \node(Mk)[draw] at (3,-1) {$M_k$};
      \draw(M1)--(N1)--(N2)--(M2);
      \draw(N2)--(dots)--(Nk)--(Mk);
      \node[anchor=base west] at (4,0) {$N_j=N_{j-1}+\hat l^T_j- l_j$ with $N_0=0$,};
      \node[anchor=base west] at (4,-1) {$M_j=\hat l^T_j-\hat l^T_{j+1}$,};
    \end{tikzpicture}
    \\
    \begin{tikzpicture}[inner sep=1pt,anchor=base]
      \node(N1)[circle,draw] at (0,0) {$\hat N_1$};
      \node(N2)[circle,draw] at (1,0) {$\hat N_2$};
      \node(dots) at (2,0) {${\,\cdots\,}$};
      \node(Nk)[circle,draw] at (3,0) {$\hat N_{\hat k}$};
      \node(M1)[draw] at (0,-1) {$\hat M_1$};
      \node(M2)[draw] at (1,-1) {$\hat M_2$};
      \node(Mk)[draw] at (3,-1) {$\hat M_{\hat k}$};
      \draw(M1)--(N1)--(N2)--(M2);
      \draw(N2)--(dots)--(Nk)--(Mk);
      \node[anchor=base west] at (4,0) {$\hat N_{\hat\jmath}=\hat N_{\hat\jmath-1}+l^T_{\hat\jmath}-\hat l_{\hat\jmath}$ with $\hat N_0=0$,};
      \node[anchor=base west] at (4,-1) {$\hat M_{\hat\jmath}=l^T_{\hat\jmath}-l^T_{\hat\jmath+1}$.};
    \end{tikzpicture}
  \end{aligned}
\end{equation}
The 3d $\mathcal{N}=4$ flavour group is $G\times\hat G$ with $G=\bigl(\prod_{j=1}^k U(M_j)\bigr)/U(1)$ and $\hat G=\bigl(\prod_{\hat\jmath=1}^k U(\hat M_{\hat\jmath})\bigr)/U(1)$.
By definition of transposition, $\hat l^T_j$ counts rows of $\hat\rho$ with at least $j$~boxes, so the following difference counts rows of~$\hat\rho$ with exactly $j$~boxes:
\begin{equation}
  \begin{aligned}
    M_j & = \hat l^T_j-\hat l^T_{j+1} = \#\{\hat\imath\mid\hat l_{\hat\imath}=j\} , \\
    \text{and likewise} \quad \hat M_{\hat\jmath} & = l^T_{\hat\jmath}- l^T_{\hat\jmath+1} = \#\{i\mid l_i=\hat\jmath\} .
  \end{aligned}
\end{equation}
We restrict our attention to \emph{good theories}: those with all $N_j\geq 1$ and $\hat N_{\hat\jmath}\geq 1$.  In particular,
$1\leq \hat N_1 = l^T_1 - \hat l_1 = k+1 - \hat l_1$, namely $\hat l_1\leq k$.  Likewise, $l_1\leq\hat k$.

An important quantity is the \emph{balance} of a node.  It takes a very simple form in terms of the partitions:
\begin{equation}
  \begin{aligned}
    & N_{j+1}+N_{j-1}+M_j-2N_j
    = (N_{j+1} - N_j) - (N_j - N_{j-1}) + M_j \\
    & \qquad = \hat l^T_{j+1} - l_{j+1} - \hat l^T_j + l_j + \hat l^T_j-\hat l^T_{j+1}
    = l_j - l_{j+1} .
  \end{aligned}
\end{equation}
The node~$j$ is \emph{balanced} if this vanishes.
An interval $\mathcal{B}\subseteq[1,k]$ of balanced nodes of the electric quiver thus corresponds to $|\mathcal{B}|+1$ consecutive~$l_j$ equal to the same value~$\hat\jmath$.
In terms of the transposed partition, this means $\hat M_{\hat\jmath} = l^T_{\hat\jmath} - l^T_{\hat\jmath+1} = |\mathcal{B}|+1$.
This is the well-known $SU(|\mathcal{B}|+1)$ flavour symmetry enhancement.

\begin{lemma}\label{lem:BalancedAb}
  If the electric quiver has a balanced abelian node $N_j=1$ then one of the following
  possibilities   holds: 
  \begin{enumerate}
  \item $1<j<k$ and $M_j=0$ and $N_{j-1}=N_{j+1}=1$;
  \item $j=k=1$ and $M_1=2$ (this is the $T[SU(2)]$ theory);
  \item $j=1$ and $M_1=1$ and $N_2=1$;
  \item $j=k$ and $M_k=1$ and $N_{k-1}=1$;
  \item $j=1$ and $M_1=0$ and $N_2=2$;
  \item $j=k$ and $M_k=0$ and $N_{k-1}=2$.
  \end{enumerate}
  The corresponding magnetic gauge group (at position $\hat\jmath\coloneqq l_j$) is abelian in the first four cases and non-abelian in the last two.
\end{lemma}
\begin{proof}
  The balance condition reads $N_{j-1}+N_{j+1}+M_j=2N_j=2$.
  This implies that $(N_{j-1},M_j,N_{j-1})$ are $(1,0,1)$, $(0,2,0)$, $(0,1,1)$, $(1,1,0)$, $(0,0,2)$ or $(2,0,0)$.
  For each case where $N_{j-1}=0$ we deduce $j=1$ because all nodes in $[1,k]$ have non-zero rank.
  Similarly, $N_{j+1}=0$ means $j=k$.
  We then work out the rank of the magnetic gauge group in each case.
\smallskip

  Case 1. From $N_j-N_{j-1}=0$ and $M_j=0$ and $N_{j+1}-N_j=0$ we see that
  $l_j=\hat l^T_j=\hat l^T_{j+1}=l_{j+1}$ (we denote this $\hat\jmath$).
  Thus the intersection of $\rho$ (drawn in blue below) and $\hat\rho^T$ (drawn in red and dashed) includes a $(j+1)\times\hat\jmath$ rectangle (drawn as thick black lines),
  and the two partitions share a boundary.
  
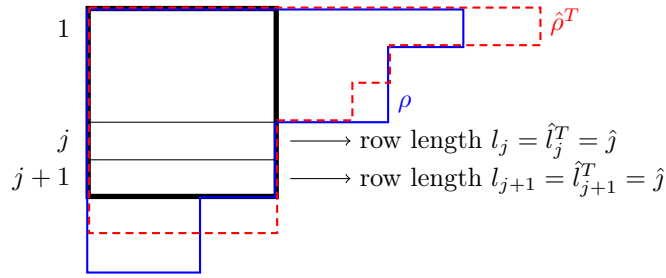
\begin{center}
    \begin{tikzpicture}[scale=.5]
      \draw[line width=2pt,shift={(.025,.025)}] (0,0) rectangle (5,-5);
      \draw[blue,thick] (0,0) -- (10,0) -- (10,-1) -- (8,-1) -- (8,-3) node [anchor=south west] {$\rho$}
      -- (5,-3) -- (5,-5) -- (3,-5) -- (3,-7) -- (0,-7) -- cycle;
      \draw[red,densely dashed,thick,shift={(.05,.05)}] (0,0) -- (12,0) -- (12,-1) node [anchor=south west] {$\hat\rho^T$}
      -- (8,-1) -- (8,-2) -- (7,-2) -- (7,-3) -- (5,-3) -- (5,-6) -- (0,-6) -- cycle;
      \draw (0,-3) -- (5,-3);
      \draw (0,-4) -- (5,-4);
      \node[anchor=east] at (-.2,-.5) {$1$};
      \node[anchor=east] at (-.2,-3.5) {$j$};
      \node[anchor=east] at (-.2,-4.5) {$j+1$};
      \draw[->] (5.4,-3.5) -- (7,-3.5) node [anchor=west] {row length $l_j=\hat l^T_j=\hat\jmath$};
      \draw[->] (5.4,-4.5) -- (7,-4.5) node [anchor=west] {row length $l_{j+1}=\hat l^T_{j+1}=\hat\jmath$};
    \end{tikzpicture}
  \end{center}  \vskip 2mm
  By definition, $\hat N_{\hat\jmath}$ counts boxes in rows $1$ through~$\hat\jmath$ of $\rho^T$, minus those in the same rows of~$\hat\rho$.
  Removing the common rectangle, this compares the numbers of boxes of the two partitions below the rectangle.
  Since the total numbers of boxes in both partitions are the same,
  it is equivalent to comparing boxes above the lower edge of the rectangle, hence $\hat N_{\hat\jmath} = N_{j+1} = 1$.
\smallskip

  Case 2. $T[SU(2)]$ is self-mirror and abelian.
\smallskip

  Cases 3.\ and 5. $N_1=1$ gives $\hat l^T_1= l_1+1$.
  Thus, $\hat N_{ l_1}$ counts boxes of $\rho^T$ (this partition has $ l_1$ rows)
  minus all boxes of $\hat\rho$ except its last ($\hat l^T_1$-th) row.
  Since $|\rho^T|=|\hat\rho|$, we conclude that the rank we care about is $\hat N_{ l_1}=\hat l_{\hat l^T_1}$.
  This in turn is equal to the number of entries of $\hat\rho^T$ equal to $\hat l^T_1$.
  Note now that $\hat l^T_1=\hat l^T_2+M_1$.
  If $M_1>0$ (case 3) then $\hat l^T_2<\hat l^T_1$ so $\hat N_{ l_1}=1$.
  If $M_1=0$ (case 5) then $\hat l^T_2=\hat l^T_1$ so $\hat N_{ l_1}\geq 2$.

\smallskip

  Cases 4.\ and 6. $N_k=1$ (and $N_{k+1}=0$) gives $ l_{k+1}=\hat l^T_{k+1}+1$, while balance gives $ l_k= l_{k+1}$.
  On general grounds, $1\leq\hat N_1= l^T_1-\hat l_1=k+1-\hat l_1$ so the number of rows~$\hat l_1$ of $\hat\rho^T$ is $\leq k$,
  hence in particular $\hat l^T_{k+1}=0$.
  From all this we deduce that $ l_k= l_{k+1}=1$ and that we want to know $\hat N_1$.
  Now use $\hat l^T_k=\hat l^T_{k+1}+M_k$.
  If $M_k=0$ then this vanishes so $\hat\rho^T$ has at most $k-1$ rows, so $\hat N_1=k+1-\hat l_1\geq 2$.
  If $M_k>0$ then $\hat\rho^T$ has $k$ rows, namely $\hat N_1=k+1-\hat l_1=1$.
\end{proof}

In the main text we introduce the number $\Delta n_{\rm mixed}$, given in~\eqref{430a},
that counts redundancies between $F$-term relations in the mixed term $x_+^2x_-^2$.

\begin{lemma}\label{lem:Deltanmixed}
  The quantity
  $\Delta n_{\rm mixed} = \delta_{N_1=1} + \delta_{N_k=1} - \sum_{j=1}^k \delta_{N_j=1} + \sum_{\hat\jmath\mid \hat N_{\hat\jmath}=1}  \hat M_{\hat\jmath}$
  is invariant under mirror symmetry.  Furthermore, $\Delta n_{\rm mixed}=3$ for abelian theories and $\Delta n_{\rm mixed}\leq 2$ otherwise.
\end{lemma}
\begin{proof}
  An important ingredient in the previous proof was an intersection point between the boundaries $\partial\rho$ of $\rho$ and $\partial\hat\rho^T$ of~$\hat\rho^T$ (we do not include the two coordinate axes in these boundaries).  Denote by $(j,\hat\jmath)$ the position of such an intersection point, where $(0,0)$ is the upper left corner, so that the partitions share a $j\times\hat\jmath$ rectangle but neither contains the box at positions $(j+1,\hat\jmath+1)$.  Then $N_j$, which counts boxes of $\hat\rho^T$ above the intersection minus those of~$\rho$, is equal to $\hat N_{\hat\jmath}$, which counts the same difference for boxes to the left of the intersection.

  Let us define the \emph{label} (an integer $\geq 1$) of each connected component of the $\partial\rho\cap\partial\hat\rho^T$ intersection of boundaries as this difference in the number of boxes above this connected component.
  Let us now understand $\Delta n_{\rm mixed}$ in terms of the components with label~$1$.

  Consider first the (non-zero) terms in $\sum_{\hat\jmath\mid \hat N_{\hat\jmath}=1} \hat M_{\hat\jmath}$, namely consider $\hat\jmath$ with $N_{\hat\jmath}=1$ and $\hat M_{\hat\jmath}\geq 1$.
  There are a few cases.
  \begin{itemize}
  \item $1<\hat\jmath<\hat k$: then $\hat N_{\hat\jmath\pm 1}\geq 1=N_{\hat\jmath}$ so $\hat l_{\hat\jmath+1}\leq l^T_{\hat\jmath+1}\leq l^T_{\hat\jmath}\leq\hat l_{\hat\jmath}$, where the middle inequality comes from $\hat M_{\hat\jmath}=l^T_{\hat\jmath}-l^T_{\hat\jmath+1}\geq 0$.  This corresponds to a vertical edge of length~$\hat M_{\hat\jmath}$ from $(l^T_{\hat\jmath+1},\hat\jmath)$ to $(l^T_{\hat\jmath},\hat\jmath)$, shared by $\rho$ and $\hat\rho^T$, and with label $N_{\hat\jmath}=1$.
  \item $1=\hat\jmath<\hat k$: now $\hat l_{\hat\jmath}=l^T_{\hat\jmath}-1$ so the shared vertical edge has length $\hat M_{\hat\jmath}-1$ from $(l^T_{\hat\jmath+1},\hat\jmath)$ to $(\hat l_{\hat\jmath},\hat\jmath)$.
  \item $1<\hat\jmath=\hat k$: now $\hat l_{\hat\jmath+1}=l^T_{\hat\jmath+1}+1$ so the shared edge has length $\hat M_{\hat\jmath}-1$ from $(\hat l_{\hat j+1},\hat\jmath)$ to $(l^T_{\hat j},\hat\jmath)$.
  \item $1=\hat\jmath=\hat k$: one checks the shared vertical edge has length $\hat M_{\hat\jmath}-2$ (non-negative because of the balance condition).
  \end{itemize}
  Conversely, every shared vertical edge from $(i_1,\hat\jmath)$ to $(i_2,\hat\jmath)$ with label~$1$ shows up in this list: indeed, the label means $N_{\hat\jmath}=1$ and the edge implies that $l^T_{\hat\jmath+1}\leq i_1<i_2\leq l^T_{\hat\jmath}$ are separated by at least $\hat M_{\hat\jmath}\geq i_2-i_1\geq 1$.
  Altogether, the total length of all vertical edges with label~$1$ shared by $\rho$ and~$\hat\rho^T$ is
  \begin{equation}
    \sum_{\hat\jmath\mid \hat N_{\hat\jmath}=1} (\hat M_{\hat\jmath}-\delta_{\hat\jmath=1}-\delta_{\hat\jmath=\hat k})
    = \biggl( \sum_{\hat\jmath\mid \hat N_{\hat\jmath}=1} \hat M_{\hat\jmath} \biggr)
    - \delta_{\hat N_1=1} - \delta_{\hat N_{\hat k}=1} .
  \end{equation}

  Next consider the other sum in $\Delta n_{\rm mixed}$, namely $\sum_{j=1}^k \delta_{N_j=1}$.  Separating four cases as above we find that this sum counts the \emph{number} (rather than length) of shared \emph{horizontal} ``edges'' with label~$1$.  To be more precise, we include among these ``edges'' one zero-length edge (intersection point) for each integer point along the shared vertical edges, as we depict in the following figure (shared horizontal edges are in black bold, and circled numbers are the labels).
  
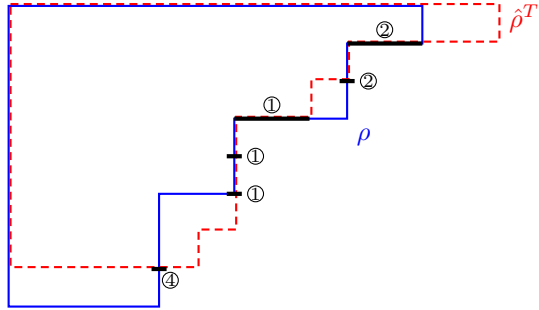
\begin{center}
    \begin{tikzpicture}[scale=.5]
      \draw[blue,thick] (0,0) -- (11,0) -- (11,-1) -- (9,-1) -- (9,-3) node [anchor=north west] {$\rho$}
      -- (6,-3) -- (6,-5) -- (4,-5) -- (4,-8) -- (0,-8) -- cycle;
      \draw[red,densely dashed,thick,shift={(.05,.05)}] (0,0) -- (13,0) -- (13,-1) node [anchor=south west] {$\hat\rho^T$}
      -- (9,-1) -- (9,-2) -- (8,-2) -- (8,-3) -- (6,-3) -- (6,-6) -- (5,-6) -- (5,-7) -- (0,-7) -- cycle;
      \draw[ultra thick] (9,-1) -- (11,-1) node [pos=.5,above=2pt,circle,thin,draw,inner sep=0pt]{\scriptsize $2$};
      \draw[ultra thick] (8.8,-2) -- (9.2,-2) node [right=2pt,circle,thin,draw,inner sep=0pt]{\scriptsize $2$};
      \draw[ultra thick] (6,-3) -- (8,-3) node [pos=.5,above=2pt,circle,thin,draw,inner sep=0pt]{\scriptsize $1$};
      \draw[ultra thick] (5.8,-4) -- (6.2,-4) node [right=2pt,circle,thin,draw,inner sep=0pt]{\scriptsize $1$};
      \draw[ultra thick] (5.8,-5) -- (6.2,-5) node [right=2pt,circle,thin,draw,inner sep=0pt]{\scriptsize $1$};
      \draw[ultra thick] (3.8,-7) -- (4.2,-7) node [pos=.5,below right=2pt,circle,thin,draw,inner sep=0pt]{\scriptsize $4$};
    \end{tikzpicture}
  \end{center}

  We are ready to put together these observations.  In each connected component of the shared boundary of $\rho$ and $\hat\rho^T$ with label~$1$, the total length of vertical edges is one less than the number of horizontal edges (including zero-length, as discussed above).  Thus,
  \begin{equation}
    \Delta n_{\rm mixed} = \delta_{N_1=1} + \delta_{N_k=1} + \delta_{\hat N_1=1} + \delta_{\hat N_{\hat k}=1}
    - \#\{\text{shared components with label~$1$}\},
  \end{equation}
  which is manifestly self-mirror.

  The end of the proof is straightforward: $\Delta n_{\rm mixed}$ is at most $4-1$, with equality if and only if $N_1=N_k=\hat N_1=\hat N_{\hat k}=1$ and the shared boundary has a single connected component with label~$1$.  In particular the horizontal edges corresponding to $N_1=1$ and to $N_k=1$ must belong to the same component so all $N_j=1$: the theory is abelian.
\end{proof}

\section{\texorpdfstring{$T[SU(2)]$}{T[SU(2)]}  index  as holomorphic blocks}\label{app:C}

As is well-known from the study of 3d ${\cal N}=2$ theories~\cite{Pasquetti:2011fj,Dimofte:2011py,Beem:2012mb}
 (see also \cite{Razamat:2014pta}  for the ${\cal N}=4$ case), superconformal indices (and various other partition functions) are bilinear combinations of basic building blocks, refered to as (anti)holomorphic blocks, which are partition functions on $\mathcal{D}^{2}\times S^{1}$.
We  work out here this factorization for $T[SU(2)]$, and  then verify that the resulting closed-form expression~\eqref{eq:tsu2closed} reproduces our expansion of the superconformal index at order $O(q)$.
The structure generalizes but we did not find it useful in concrete calculations, 
because for generic theories  this  factorized form  contains  a large number of terms.

The expression for the full superconformal index of $T[SU(2)]$ reads:
\begin{equation}\label{eq:tsu21}
  \mathcal{Z}^{T[SU(2)]}_{S^{2}\times S^{1}} = \frac{(q^{\frac{1}{2}}t;q)_{\infty}}{(q^{\frac{1}{2}}t^{-1};q)_{\infty}} \sum_{ m\in\mathbb{Z}}\, (q^{\frac{1}{2}}t^{-1})^{|m|} w^m \oint_{S^1}\frac{dz}{2\pi i z} \prod_{p=1}^2\prod_{\pm} \frac{(t^{-\frac{1}{2}}q^{\frac{3}{4}+\frac{|m|}{2}}z^{\pm}\mu_p^{\mp};q)_{\infty}}{(t^{\frac{1}{2}}q^{\frac{1}{4}+\frac{|m|}{2}}z^{\mp}\mu_p^{\pm};q)_{\infty}}\ , 
\end{equation}
where $m$ is the unique monopole charge, and $z$ runs over the  unit circle in the classical Coulomb branch~$\mathbb{C}$.
The integrand has poles at\footnote{At first sight there is also  a pole at $z=0$, but in fact the $q$-Pochhammer factors tend to zero there.}   
\begin{equation}
  \begin{aligned}
    z = z_{s,j} \coloneqq \mu_s t^{\frac{1}{2}}q^{\frac{1}{4}+\frac{|m|}{2}+j}
    \quad \text{and} \quad
    z = \mu_s \bigl( t^{\frac{1}{2}}q^{\frac{1}{4}+\frac{|m|}{2}+j} \bigr)^{-1}
    \quad \text{for $s=1,2$ and integer $j\geq 0$.}
  \end{aligned}
\end{equation}
We calculate the index as an expansion in powers of~$q$, hence $|q|<1$, with $|t|=|\mu_s|=1$.
The poles that we named $z_{s,j}$ thus lie inside the $|z|=1$ contour and the other poles outside.
\smallskip

To warm up, compute the contribution to $\mathcal{Z}^{T[SU(2)]}_{S^{2}\times S^{1}}$ from the pole at $z_{s,0}$ for $m=0$:
\begin{equation}
  C_s \coloneqq \prod_{p\neq s} \frac{(q \mu_s \mu_p^{-1};q)_{\infty}}{(\mu_s^{-1}\mu_p;q)_{\infty}}
  \frac{(q^{\frac{1}{2}}t^{-1}\mu_s^{-1}\mu_p;q)_{\infty}}{(q^{\frac{1}{2}}t\mu_s\mu_p^{-1};q)_{\infty}} .
\end{equation}
Before moving on to other residues, we note that the identity
\begin{equation}\label{absm}
  \bigl(iq^{\frac{1}{8}}a^{\frac{1}{2}}\bigr)^{|m|}\frac{(q^{\frac{3}{4}+\frac{|m|}{2}}a;q)_{\infty}}{(q^{\frac{1}{4}+\frac{|m|}{2}}a^{-1};q)_{\infty}}
  = \bigl(iq^{\frac{1}{8}}a^{\frac{1}{2}}\bigr)^m\frac{(q^{\frac{3}{4}+\frac{m}{2}}a;q)_{\infty}}{(q^{\frac{1}{4}+\frac{m}{2}}a^{-1};q)_{\infty}}
\end{equation}
allows us to replace $|m|\to m$ throughout~\eqref{eq:tsu21}.
The resulting expression involves both positive and negative powers of~$q$, which would make our lives harder  
 if we wanted to expand in powers of~$q$, but leads to nicer residues.
We compute the contribution from the $z_{s,j}$ pole for any~$m$:
\begin{equation}
  \begin{aligned}
    & \frac{(q^{\frac{1}{2}}t;q)_{\infty}}{(q^{\frac{1}{2}}t^{-1};q)_{\infty}} (q^{\frac{1}{2}}t^{-1}w)^m
    \prod_{p=1}^2 \frac{(q^{1+j+\frac{|m|+m}{2}}\mu_s\mu_p^{-1};q)_{\infty}(q^{\frac{1}{2}-j-\frac{|m|-m}{2}}t^{-1}\mu_s^{-1}\mu_p;q)_{\infty}}
    {\bigl((q^{-j-\frac{|m|-m}{2}} \mu_s^{-1}\mu_p;q)_{\infty}\bigr)'(q^{\frac{1}{2}+j+\frac{|m|+m}{2}} t \mu_s \mu_p^{-1};q)_{\infty}} \\
    & = C_s\, (q^{\frac{1}{2}}t^{-1}w)^{k_+-k_-}
    \prod_{p=1}^2 \frac{(q^{\frac{1}{2}} t \mu_s \mu_p^{-1};q)_{k_+} (q^{\frac{1}{2}-k_-}t^{-1}\mu_s^{-1}\mu_p;q)_{k_-}}
    {(q\mu_s\mu_p^{-1};q)_{k_+}(q^{-k_-} \mu_s^{-1}\mu_p;q)_{k_-}}
  \end{aligned}
\end{equation}
where the prime in the first line denotes the removal of the vanishing factor in the $q$-Pochhammer symbol for $p=s$,
and we then used finite $q$-Pochhammer $(a;q)_k=(a;q)_\infty/(aq^k;q)_\infty$ and
changed variables to $k_{\pm}\coloneqq j+\frac{|m|\pm m}{2}\geq 0$.
Altogether 
\begin{equation}
  \mathcal{Z}^{T[SU(2)]}_{S^{2}\times S^{1}}
  = \sum_{s=1}^2 C_s \prod_{\pm} \biggl(
  \sum_{k_{\pm}\geq 0} (q^{\frac{1}{2}}t^{-1}w^{\pm 1})^{k_{\pm}} \prod_{p=1}^2
  \frac{(q^{\frac{1}{2}} t \mu_s \mu_p^{-1};q)_{k_{\pm}}}{(q\mu_s\mu_p^{-1};q)_{k_{\pm}}} \biggr).
\end{equation}
We recognize here the $q$-hypergeometric series
\begin{equation}
  {}_2\phi_1\biggl[\begin{matrix}a,b\\c\end{matrix}\biggm|q,z\biggr]
  \coloneqq\sum_{k\geq 0}\frac{(a;q)_k(b;q)_k}{(q;q)_k(c;q)_k}z^k .
\end{equation}
In terms of $\mu \coloneqq\mu_1\mu_2^{-1}$ and $\hat\mu \coloneqq w$
\begin{equation}\label{eq:tsu2closed}
  \mathcal{Z}^{T[SU(2)]}_{S^{2}\times S^{1}}
  = \frac{(q \mu ;q)_{\infty}}{(\mu^{-1};q)_{\infty}}
  \frac{(q^{\frac{1}{2}}t^{-1}\mu^{-1};q)_{\infty}}{(q^{\frac{1}{2}}t\mu ;q)_{\infty}}
  \prod_{\pm} \Biggl( {}_2\phi_1\biggl[\begin{matrix}
    q^{\frac{1}{2}} t ,    q^{\frac{1}{2}} t \mu  \\
    q\mu 
  \end{matrix}\biggm|q,q^{\frac{1}{2}}t^{-1}\hat\mu^{\pm 1}\biggr] \Biggr)
  + \bigl(\mu\leftrightarrow\mu^{-1}\bigr) .
\end{equation}
This is the factorized form of the index.
It is possible to show, using complicated identities obeyed by $q$-hypergeometric series,
that this result is mirror-symmetric

\smallskip

To compare with the main text we   expand in powers of~$q$ and organize 
the series in terms of supercharacters so as to  extract the representation content:
\begin{equation}\label{C9}
  \begin{aligned}
    \mathcal{Z}^{T[SU(2)]}_{S^{2}\times S^{1}}
    & = 1 + q^{\frac{1}{2}} t \chi_3(\mu) + q^{\frac{1}{2}} t^{-1} \chi_3(\hat\mu) + q t^2 \chi_5(\mu) 
    + q t^{-2} \chi_5(\hat\mu) - q \bigl(1+\chi_3(\mu)+\chi_3(\hat\mu) \bigr)   + O(q^{\frac{3}{2}}) \\
    & = 1 + \chi_3(\mu)\,{\cal I}_{(1,0)} + \chi_3(\hat\mu)\,{\cal I}_{(0,1)} +
    \boxed{ \chi_5(\mu)\,{\cal I}_{(2,0)} } - {\cal I}_{(1,1)} + 
    \boxed{\chi_5(\hat\mu)\,{\cal I}_{(0,2)} }  + O(q^{\frac{3}{2}})
  \end{aligned}
\end{equation}
\hskip 0.3mm

\noindent  where $\chi_3(\mu)\coloneqq \mu+1+\mu^{-1}$ and $\chi_5(\mu)\coloneqq \mu^2+\mu+1+\mu^{-1}+\mu^{-2}$ are characters of $SU(2)$, and we used the short-hand notation for the superconformal indices
${\cal I}_{(J^H, J^C)}\coloneqq{\cal I}_{B_{1}[0]^{(J^H, J^C)}}(q,t)$. This agrees with eq.\,\eqref{432}
of \autoref{sec:4}. 
\smallskip

As explained in the paper  the following BPS multiplets  can be  unambiguously  identified:
\begin{itemize}
\item $1$: the identity;
\item $\chi_3(\mu){\cal I}_{(1,0)}$: the $SU(2)$ electric-flavour currents;
\item $\chi_3(\hat\mu){\cal I}_{(0,1)}$: the $S{\widehat U(}2)$ magnetic-flavour currents;
\item $\chi_5(\mu){\cal I}_{(2,0)}$: products of two electric currents;
\item $\chi_5(\hat\mu){\cal I}_{(0,2)}$: products of two magnetic currents;
\item $-{\cal I}_{(1,1)}$: the energy-momentum tensor multiplet $A_2[0]^{(0,0)}$.
\end{itemize}
The bottom component $\tilde{Q}^{\bar p}Q^r$ of an electric-current multiplet
 is the product of an antifundamental and a fundamental chiral scalar   (the $F$-term condition imposes $\tilde{Q}^1Q^1+\tilde{Q}^2Q^2=0$).
Since the gauge group is abelian,   $\tilde{Q}^{\bar p}Q^r$ has rank~$1$ hence zero determinant.
This removes one  of  the six products of two electric currents,
 thus explaining why there are only five such products in \eqref{C9}.
\hskip 4mm

Altogether we see that the $T[SU(2)]$ theory has no mixed marginal (or relevant)  chiral operators. 
All exactly-marginal deformations   are purely electric or purely magnetic superpotentials.
After imposing the D-term conditions the supeconformal manifold has dimension $10-7 = 3$.

%%%%%%%%%%%%%%%%%


\begin{thebibliography}{99}

%\cite{Gaiotto:2008ak}
\bibitem{Gaiotto:2008ak}
  D.~Gaiotto and E.~Witten,
  ``S-Duality of Boundary Conditions In N=4 Super Yang-Mills Theory,''
  Adv.\ Theor.\ Math.\ Phys.\  {\bf 13} (2009) no.3,  721
  %doi:10.4310/ATMP.2009.v13.n3.a5
  \href{https://arxiv.org/abs/0807.3720}{[arXiv:0807.3720 [hep-th]]}.
  %%CITATION = doi:10.4310/ATMP.2009.v13.n3.a5;%%
  %342 citations counted in INSPIRE as of 18 Nov 2018   
    
      %\cite{Hanany:1996ie}
\bibitem{Hanany:1996ie}
  A.~Hanany and E.~Witten,
  ``Type IIB superstrings, BPS monopoles, and three-dimensional gauge dynamics,''
  Nucl.\ Phys.\ B {\bf 492} (1997) 152
  %doi:10.1016/S0550-3213(97)00157-0, 10.1016/S0550-3213(97)80030-2
  \href{https://arxiv.org/abs/hep-th/9611230}{[hep-th/9611230]}.
  %%CITATION = doi:10.1016/S0550-3213(97)00157-0, 10.1016/S0550-3213(97)80030-2;%%
  %1063 citations counted in INSPIRE as of 04 Apr 2019

          
 %\cite{Assel:2011xz}
\bibitem{Assel:2011xz}
  B.~Assel, C.~Bachas, J.~Estes and J.~Gomis,
  ``Holographic Duals of D=3 N=4 Superconformal Field Theories,''
  JHEP {\bf 1108} (2011) 087
  %doi:10.1007/JHEP08(2011)087
  \href{https://arxiv.org/abs/1106.4253}{[arXiv:1106.4253 [hep-th]]}.
  %%CITATION = doi:10.1007/JHEP08(2011)087;%%
  %52 citations counted in INSPIRE as of 18 Nov 2018
 
  %\cite{Assel:2012cj}
\bibitem{Assel:2012cj}
  B.~Assel, C.~Bachas, J.~Estes and J.~Gomis,
  ``IIB Duals of D=3 N=4 Circular Quivers,''
  JHEP {\bf 1212} (2012) 044
  %doi:10.1007/JHEP12(2012)044
  \href{https://arxiv.org/abs/1210.2590}{[arXiv:1210.2590 [hep-th]]}.
  %%CITATION = doi:10.1007/JHEP12(2012)044;%%
  %30 citations counted in INSPIRE as of 26 Dec 2018 
  
     %\cite{Cottrell:2016nsu}
\bibitem{Cottrell:2016nsu}
  W.~Cottrell and A.~Hashimoto,
  ``Resolved gravity duals of ${\cal N}=4$ quiver field theories in 2+1 dimensions,''
  JHEP {\bf 1610} (2016) 057
  %doi:10.1007/JHEP10(2016)057
  \href{https://arxiv.org/abs/1602.04765}{[arXiv:1602.04765 [hep-th]]}.
  %%CITATION = doi:10.1007/JHEP10(2016)057;%%
  %3 citations counted in INSPIRE as of 24 Apr 2019
           
 %\cite{Lozano:2016wrs}
\bibitem{Lozano:2016wrs}
  Y.~Lozano, N.~T.~Macpherson, J.~Montero and C.~Nunez,
  ``Three-dimensional $ \mathcal{N}=4 $ linear quivers and non-Abelian T-duals,''
  JHEP {\bf 1611} (2016) 133
 % doi:10.1007/JHEP11(2016)133
  \href{https://arxiv.org/abs/1609.09061}{[arXiv:1609.09061 [hep-th]]}.
  %%CITATION = doi:10.1007/JHEP11(2016)133;%%
  %25 citations counted in INSPIRE as of 24 Apr 2019   
  

                  
   %\cite{Bachas:2017wva}
\bibitem{Bachas:2017wva}
  C.~Bachas, M.~Bianchi and A.~Hanany,
  ``$ \mathcal{N}=2 $ moduli of AdS$_{4}$ vacua: a fine-print study,''
  JHEP {\bf 1808} (2018) 100
   Erratum: [JHEP {\bf 1810} (2018) 032]
  %doi:10.1007/JHEP08(2018)100, 10.1007/JHEP10(2018)032
  \href{https://arxiv.org/abs/1711.06722}{[arXiv:1711.06722 [hep-th]]}.
  %%CITATION = doi:10.1007/JHEP08(2018)100, 10.1007/JHEP10(2018)032;%%
  %3 citations counted in INSPIRE as of 17 Nov 2018

  %\cite{Cordova:2016xhm}
\bibitem{Cordova:2016xhm}
  C.~Cordova, T.~T.~Dumitrescu and K.~Intriligator,
  ``Deformations of Superconformal Theories,''
  JHEP {\bf 1611} (2016) 135
 % doi:10.1007/JHEP11(2016)135
  \href{https://arxiv.org/abs/1602.01217}{[arXiv:1602.01217 [hep-th]]}.
  %%CITATION = doi:10.1007/JHEP11(2016)135;%%
  %61 citations counted in INSPIRE as of 18 Nov 2018
 
 %\cite{Leigh:1995ep}
\bibitem{Leigh:1995ep}
  R.~G.~Leigh and M.~J.~Strassler,
``Exactly marginal operators and duality in four-dimensional N=1
supersymmetric gauge theory,''
  Nucl.\ Phys.\ B {\bf 447} (1995) 95
 % doi:10.1016/0550-3213(95)00261-P
  \href{https://arxiv.org/abs/hep-th/9503121}{[hep-th/9503121]}.
  %%CITATION = doi:10.1016/0550-3213(95)00261-P;%%
 
 
           %\cite{Kol:2002zt}
\bibitem{Kol:2002zt}
  B.~Kol,
  ``On conformal deformations,''
  JHEP {\bf 0209} (2002) 046
  %doi:10.1088/1126-6708/2002/09/046
  \href{https://arxiv.org/abs/hep-th/0205141}{[hep-th/0205141] ;} 
   ``On Conformal Deformations II,''
  \href{https://arxiv.org/abs/1005.4408}{[arXiv:1005.4408 [hep-th]]}.
  %%CITATION = ARXIV:1005.4408;%%                    

  %\cite{Aharony:2002hx}
\bibitem{Aharony:2002hx}
  O.~Aharony, B.~Kol and S.~Yankielowicz,
  ``On exactly marginal deformations of N=4 SYM and Type IIB supergravity on AdS(5) x S**5,''
  JHEP {\bf 0206} (2002) 039
  %doi:10.1088/1126-6708/2002/06/039
  \href{https://arxiv.org/abs/hep-th/0205090}{[hep-th/0205090]}.
  %%CITATION = doi:10.1088/1126-6708/2002/06/039;%%
  %74 citations counted in INSPIRE as of 07 Nov 2017           
   
   %\cite{Benvenuti:2005wi}
\bibitem{Benvenuti:2005wi} 
  S.~Benvenuti and A.~Hanany,
``Conformal manifolds for the conifold and other toric field
theories,''
  JHEP {\bf 0508}, 024 (2005)
 % doi:10.1088/1126-6708/2005/08/024
  \href{https://arxiv.org/abs/hep-th/0502043}{[hep-th/0502043]}.
  %%CITATION = doi:10.1088/1126-6708/2005/08/024;%%
  %62 citations counted in INSPIRE as of 23 Oct 2017  
 
%\cite{Green:2010da}
\bibitem{Green:2010da}
D.~Green, Z.~Komargodski, N.~Seiberg, Y.~Tachikawa and B.~Wecht,  ``Exactly Marginal Deformations and Global Symmetries,''
  JHEP {\bf 1006} (2010) 106
  %doi:10.1007/JHEP06(2010)106
  \href{https://arxiv.org/abs/1005.3546}{[arXiv:1005.3546 [hep-th]]}.

%\cite{deAlwis:2013jaa}
\bibitem{deAlwis:2013jaa}
  S.~de Alwis, J.~Louis, L.~McAllister, H.~Triendl and A.~Westphal,
  ``Moduli spaces in $AdS_4$ supergravity,''
  JHEP {\bf 1405} (2014) 102
  %doi:10.1007/JHEP05(2014)102
  \href{https://arxiv.org/abs/1312.5659}{[arXiv:1312.5659 [hep-th]]}.
  %%CITATION = doi:10.1007/JHEP05(2014)102;%%
  %23 citations counted in INSPIRE as of 18 Nov 2018       

%\cite{Bachas:2017rch}
\bibitem{Bachas:2017rch}
  C.~Bachas and I.~Lavdas,
  ``Quantum Gates to other Universes,''
  Fortsch.\ Phys.\  {\bf 66} (2018) no.2,  1700096
 % doi:10.1002/prop.201700096
  \href{https://arxiv.org/abs/1711.11372}{[arXiv:1711.11372 [hep-th]];} 
   ``Massive Anti-de Sitter Gravity from String Theory,''
  JHEP {\bf 1811} (2018) 003
  %doi:10.1007/JHEP11(2018)003
  \href{https://arxiv.org/abs/1807.00591}{[arXiv:1807.00591 [hep-th]]}.
  %%CITATION = doi:10.1007/JHEP11(2018)003;%%
  %6 citations counted in INSPIRE as of 13 Apr 2019
  
  %\cite{Bachas:2019rfq}
\bibitem{Bachas:2019rfq}
  C.~Bachas,
  ``Massive AdS Supergravitons and Holography,''
  \href{https://arxiv.org/abs/1905.05039}{[arXiv:1905.05039 [hep-th]]}.
  %%CITATION = ARXIV:1905.05039;%%
  
  
  
  %\cite{Cordova:2016emh}
\bibitem{Cordova:2016emh}
  C.~Cordova, T.~T.~Dumitrescu and K.~Intriligator,
  ``Multiplets of Superconformal Symmetry in Diverse Dimensions,''
  JHEP {\bf 1903} (2019) 163
  %doi:10.1007/JHEP03(2019)163
  \href{https://arxiv.org/abs/1612.00809}{[arXiv:1612.00809 [hep-th]]}.
  %%CITATION = doi:10.1007/JHEP03(2019)163;%%
  %59 citations counted in INSPIRE as of 10 Apr 2019



%\cite{Hanany:2011db}
\bibitem{Hanany:2011db}
  A.~Hanany and N.~Mekareeya,
  ``Complete Intersection Moduli Spaces in N=4 Gauge Theories in Three Dimensions,''
  JHEP {\bf 1201} (2012) 079
  %doi:10.1007/JHEP01(2012)079
  \href{https://arxiv.org/abs/1110.6203}{[arXiv:1110.6203 [hep-th]]}.
  %%CITATION = doi:10.1007/JHEP01(2012)079;%%
  %25 citations counted in INSPIRE as of 20 Apr 2019
                                                                                  
                                                                                                                                                          
                                                                                     
 %\cite{Cremonesi:2013lqa}
\bibitem{Cremonesi:2013lqa}
  S.~Cremonesi, A.~Hanany and A.~Zaffaroni,
  ``Monopole operators and Hilbert series of Coulomb branches of $3d$  $\mathcal{N} = 4$ gauge theories,''
  JHEP {\bf 1401} (2014) 005
  %doi:10.1007/JHEP01(2014)005
  \href{https://arxiv.org/abs/1309.2657}{[arXiv:1309.2657 [hep-th]]}.
  %%CITATION = doi:10.1007/JHEP01(2014)005;%%
  %92 citations counted in INSPIRE as of 13 Apr 2019

 %\cite{Cremonesi:2014uva}
\bibitem{Cremonesi:2014uva}
  S.~Cremonesi, A.~Hanany, N.~Mekareeya and A.~Zaffaroni,
  ``T$_{\rho}^{\sigma}$ (G) theories and their Hilbert series,''
  JHEP {\bf 1501} (2015) 150
  %doi:10.1007/JHEP01(2015)150
  \href{https://arxiv.org/abs/1410.1548}{[arXiv:1410.1548 [hep-th]]}.
  %%CITATION = doi:10.1007/JHEP01(2015)150;%%
  %40 citations counted in INSPIRE as of 20 Apr 2019

%\cite{Cremonesi:2014kwa}
\bibitem{Cremonesi:2014kwa}
  S.~Cremonesi, A.~Hanany, N.~Mekareeya and A.~Zaffaroni,
  ``Coulomb branch Hilbert series and Hall-Littlewood polynomials,''
  JHEP {\bf 1409} (2014) 178
  %doi:10.1007/JHEP09(2014)178
  \href{https://arxiv.org/abs/1403.0585}{[arXiv:1403.0585 [hep-th]]}.
  %%CITATION = doi:10.1007/JHEP09(2014)178;%%
  %43 citations counted in INSPIRE as of 20 Apr 2019

 %\cite{Cremonesi:2017jrk}
\bibitem{Cremonesi:2017jrk}
  S.~Cremonesi,
  ``3d supersymmetric gauge theories and Hilbert series,''
  Proc.\ Symp.\ Pure Math.\  {\bf 98} (2018) 21
  \href{https://arxiv.org/abs/1701.00641}{[arXiv:1701.00641 [hep-th]]}.
  %%CITATION = ARXIV:1701.00641;%%
  %10 citations counted in INSPIRE as of 13 Apr 2019

%\cite{Razamat:2014pta}
\bibitem{Razamat:2014pta}
  S.~S.~Razamat and B.~Willett,
  ``Down the rabbit hole with theories of class $ \mathcal{S} $,''
  JHEP {\bf 1410} (2014) 99
  %doi:10.1007/JHEP10(2014)099
  \href{https://arxiv.org/abs/1403.6107}{[arXiv:1403.6107 [hep-th]]}.
  %%CITATION = doi:10.1007/JHEP10(2014)099;%%
  %32 citations counted in INSPIRE as of 05 Dec 2018

%\cite{Cremonesi:2016nbo}
\bibitem{Cremonesi:2016nbo}
  S.~Cremonesi, N.~Mekareeya and A.~Zaffaroni,
  ``The moduli spaces of 3d $ \mathcal{N}\ge 2 $ Chern-Simons gauge theories and their Hilbert series,''
  JHEP {\bf 1610} (2016) 046
  %doi:10.1007/JHEP10(2016)046
  \href{https://arxiv.org/abs/1607.05728}{[arXiv:1607.05728 [hep-th]]}.
  %%CITATION = doi:10.1007/JHEP10(2016)046;%%
  %18 citations counted in INSPIRE as of 13 Apr 2019

  
                                                                                                                                                
   %\cite{Carta:2016fjb}
\bibitem{Carta:2016fjb}
  F.~Carta and H.~Hayashi,
  ``Hilbert series and mixed branches of $T[SU(N)]$ theory,''
  JHEP {\bf 1702} (2017) 037
  %doi:10.1007/JHEP02(2017)037
  \href{https://arxiv.org/abs/1609.08034}{[arXiv:1609.08034 [hep-th]]}.
  %%CITATION = doi:10.1007/JHEP02(2017)037;%%
  %4 citations counted in INSPIRE as of 13 Apr 2019 
 
           %\cite{Witten:2001ua}
\bibitem{Witten:2001ua}
  E.~Witten,
  ``Multitrace operators, boundary conditions, and AdS / CFT correspondence,''
  \href{https://arxiv.org/abs/hep-th/0112258}{[hep-th/0112258]}.
  %%CITATION = HEP-TH/0112258;%%
  %418 citations counted in INSPIRE as of 28 Apr 2019                       
                                                 
      %\cite{Berkooz:2002ug}
\bibitem{Berkooz:2002ug}
  M.~Berkooz, A.~Sever and A.~Shomer,
  ``'Double trace' deformations, boundary conditions and space-time singularities,''
  JHEP {\bf 0205} (2002) 034
  %doi:10.1088/1126-6708/2002/05/034
  \href{https://arxiv.org/abs/hep-th/0112264}{[hep-th/0112264]}.
  %%CITATION = doi:10.1088/1126-6708/2002/05/034;%%
  %199 citations counted in INSPIRE as of 28 Apr 2019    
   
    
      
  %\cite{Pasquetti:2011fj}
\bibitem{Pasquetti:2011fj}
  S.~Pasquetti,
  ``Factorisation of N = 2 Theories on the Squashed 3-Sphere,''
  JHEP {\bf 1204} (2012) 120
 % doi:10.1007/JHEP04(2012)120
  \href{https://arxiv.org/abs/1111.6905}{[arXiv:1111.6905 [hep-th]]}.
  %%CITATION = doi:10.1007/JHEP04(2012)120;%%
  %107 citations counted in INSPIRE as of 27 Apr 2019

%\cite{Dimofte:2011py}
\bibitem{Dimofte:2011py}
  T.~Dimofte, D.~Gaiotto and S.~Gukov,
  ``3-Manifolds and 3d Indices,''
  Adv.\ Theor.\ Math.\ Phys.\  {\bf 17} (2013) no.5,  975
  %doi:10.4310/ATMP.2013.v17.n5.a3
  \href{https://arxiv.org/abs/1112.5179}{[arXiv:1112.5179 [hep-th]]}.
  %%CITATION = doi:10.4310/ATMP.2013.v17.n5.a3;%%
  %161 citations counted in INSPIRE as of 08 Apr 2019


%\cite{Beem:2012mb}
\bibitem{Beem:2012mb}
  C.~Beem, T.~Dimofte and S.~Pasquetti,
  ``Holomorphic Blocks in Three Dimensions,''
  JHEP {\bf 1412} (2014) 177
  %doi:10.1007/JHEP12(2014)177
  \href{https://arxiv.org/abs/1211.1986}{[arXiv:1211.1986 [hep-th]]}.
  %%CITATION = doi:10.1007/JHEP12(2014)177;%%
  %134 citations counted in INSPIRE as of 07 Apr 2019
  
  
%\cite{Okazaki:2019ony}
\bibitem{Okazaki:2019ony}
  T.~Okazaki,
  ``Mirror symmetry of 3d $\mathcal{N}=4$ gauge theories and supersymmetric indices,''
  \href{https://arxiv.org/abs/1905.04608}{[arXiv:1905.04608 [hep-th]]}.
  %%CITATION = ARXIV:1905.04608;%%
  %2 citations counted in INSPIRE as of 24 Jun 2019

 
 %\cite{Borokhov:2002ib}
\bibitem{Borokhov:2002ib}
  V.~Borokhov, A.~Kapustin and X.~k.~Wu,
  ``Topological disorder operators in three-dimensional conformal field theory,''
  JHEP {\bf 0211} (2002) 049
  %doi:10.1088/1126-6708/2002/11/049
  \href{https://arxiv.org/abs/hep-th/0206054}{[hep-th/0206054] ; \,} 
  %%CITATION = doi:10.1088/1126-6708/2002/11/049;%%
  %149 citations counted in INSPIRE as of 10 Apr 2019
 ``Monopole operators and mirror symmetry in three-dimensions,''
  JHEP {\bf 0212} (2002) 044
  %doi:10.1088/1126-6708/2002/12/044
  \href{https://arxiv.org/abs/hep-th/0207074}{[hep-th/0207074]}.
 
 %\cite{Borokhov:2003yu}
\bibitem{Borokhov:2003yu}
  V.~Borokhov,
  ``Monopole operators in three-dimensional N=4 SYM and mirror symmetry,''
  JHEP {\bf 0403} (2004) 008
  %doi:10.1088/1126-6708/2004/03/008
  \href{https://arxiv.org/abs/hep-th/0310254}{[hep-th/0310254]}.
  %%CITATION = doi:10.1088/1126-6708/2004/03/008;%%
  %48 citations counted in INSPIRE as of 10 Apr 2019 
 
 

 
 %\cite{Bhattacharya:2008zy}
\bibitem{Bhattacharya:2008zy}
  J.~Bhattacharya, S.~Bhattacharyya, S.~Minwalla and S.~Raju,
  ``Indices for Superconformal Field Theories in 3,5 and 6 Dimensions,''
  JHEP {\bf 0802} (2008) 064
 % doi:10.1088/1126-6708/2008/02/064
  \href{https://arxiv.org/abs/0801.1435}{[arXiv:0801.1435 [hep-th]]}.
  %%CITATION = doi:10.1088/1126-6708/2008/02/064;%%
  %152 citations counted in INSPIRE as of 07 Apr 2019
 
 %\cite{Kim:2009wb}
\bibitem{Kim:2009wb}
  S.~Kim,
  ``The Complete superconformal index for N=6 Chern-Simons theory,''
  Nucl.\ Phys.\ B {\bf 821} (2009) 241
   Erratum: [Nucl.\ Phys.\ B {\bf 864} (2012) 884]
  %doi:10.1016/j.nuclphysb.2012.07.015, 10.1016/j.nuclphysb.2009.06.025
  \href{https://arxiv.org/abs/0903.4172}{[arXiv:0903.4172 [hep-th]]}.
  %%CITATION = doi:10.1016/j.nuclphysb.2012.07.015, 10.1016/j.nuclphysb.2009.06.025;%%
  %218 citations counted in INSPIRE as of 07 Apr 2019
 
  %\cite{Imamura:2011su}
\bibitem{Imamura:2011su}
  Y.~Imamura and S.~Yokoyama,
  ``Index for three dimensional superconformal field theories with general R-charge assignments,''
  JHEP {\bf 1104} (2011) 007
  %doi:10.1007/JHEP04(2011)007
  \href{https://arxiv.org/abs/1101.0557}{[arXiv:1101.0557 [hep-th]]}.
  %%CITATION = doi:10.1007/JHEP04(2011)007;%%
  %169 citations counted in INSPIRE as of 05 Dec 2018
 
 
 
         %\cite{Krattenthaler:2011da}
\bibitem{Krattenthaler:2011da} 
  C.~Krattenthaler, V.~P.~Spiridonov and G.~S.~Vartanov,
  ``Superconformal indices of three-dimensional theories related by mirror symmetry,''
  JHEP {\bf 1106}, 008 (2011)
  % doi:10.1007/JHEP06(2011)008
  \href{https://arxiv.org/abs/1103.4075}{[arXiv:1103.4075 [hep-th]]}.
  %%CITATION = doi:10.1007/JHEP06(2011)008;%%
  %71 citations counted in INSPIRE as of 19 Feb 2019



%\cite{Kapustin:2011jm}
\bibitem{Kapustin:2011jm}
  A.~Kapustin and B.~Willett,
  ``Generalized Superconformal Index for Three Dimensional Field Theories,''
  \href{https://arxiv.org/abs/1106.2484}{[arXiv:1106.2484 [hep-th]]}.
  %%CITATION = ARXIV:1106.2484;%%
  %98 citations counted in INSPIRE as of 07 Apr 2019



%\cite{Benini:2013yva}
\bibitem{Benini:2013yva}
  F.~Benini and W.~Peelaers,
  ``Higgs branch localization in three dimensions,''
  JHEP {\bf 1405} (2014) 030
  %doi:10.1007/JHEP05(2014)030
  \href{https://arxiv.org/abs/1312.6078}{[arXiv:1312.6078 [hep-th]]}.
  %%CITATION = doi:10.1007/JHEP05(2014)030;%%
  %64 citations counted in INSPIRE as of 07 Apr 2019
 

%\cite{Willett:2016adv}
\bibitem{Willett:2016adv}
  B.~Willett,
  ``Localization on three-dimensional manifolds,''
  J.\ Phys.\ A {\bf 50} (2017) no.44,  443006
  %doi:10.1088/1751-8121/aa612f
  \href{https://arxiv.org/abs/1608.02958}{[arXiv:1608.02958 [hep-th]]}.
  %%CITATION = doi:10.1088/1751-8121/aa612f;%%
  %24 citations counted in INSPIRE as of 05 Dec 2018
 

 
 %\cite{Dolan:2008vc}
\bibitem{Dolan:2008vc}
  F.~A.~Dolan,
  ``On Superconformal Characters and Partition Functions in Three Dimensions,''
  J.\ Math.\ Phys.\  {\bf 51} (2010) 022301
  %doi:10.1063/1.3211091
  \href{https://arxiv.org/abs/0811.2740}{[arXiv:0811.2740 [hep-th]]}.
  %%CITATION = doi:10.1063/1.3211091;%%
  %36 citations counted in INSPIRE as of 10 Apr 2019
  
  %\cite{Distler:2017xba}
\bibitem{Distler:2017xba}
  J.~Distler, B.~Ergun and F.~Yan,
  ``Product SCFTs in Class-S,''
  \href{https://arxiv.org/abs/1711.04727}{[arXiv:1711.04727 [hep-th]]}.
  %%CITATION = ARXIV:1711.04727;%%
  %7 citations counted in INSPIRE as of 15 May 2019 
 

%\cite{Razamat:2016gzx}
\bibitem{Razamat:2016gzx} 
  S.~S.~Razamat and G.~Zafrir,
  ``Exceptionally simple exceptional models,''
  JHEP {\bf 1611} 061 (2016)
  %doi:10.1007/JHEP11(2016)061
  \href{https://arxiv.org/abs/1609.02089}{[arXiv:1609.02089 [hep-th]]}.
  %%CITATION = doi:10.1007/JHEP11(2016)061;%%
  %5 citations counted in INSPIRE as of 24 Jun 2019

%\cite{Fazzi:2018rkr}
\bibitem{Fazzi:2018rkr} 
  M.~Fazzi, A.~Lanir, S.~S.~Razamat and O.~Sela,
  ``Chiral 3d SU(3) SQCD and $ \mathcal{N}=2 $ mirror duality,''
  JHEP {\bf 1811} 025 (2018)
  %doi:10.1007/JHEP11(2018)025
  \href{https://arxiv.org/abs/1808.04173}{[arXiv:1808.04173 [hep-th]]}.
  %%CITATION = doi:10.1007/JHEP11(2018)025;%%
  %6 citations counted in INSPIRE as of 24 Jun 2019
 
 
  
  
 
 %\cite{Bachas:2000ik}
\bibitem{Bachas:2000ik}
  C.~Bachas, M.~R.~Douglas and C.~Schweigert,
  ``Flux stabilization of D-branes,''
  JHEP {\bf 0005} (2000) 048
  %doi:10.1088/1126-6708/2000/05/048
  \href{https://arxiv.org/abs/hep-th/0003037}{[hep-th/0003037]}.
  %%CITATION = doi:10.1088/1126-6708/2000/05/048;%%
  %221 citations counted in INSPIRE as of 27 Apr 2019 
 
  %\cite{Tamm:1931dda}
\bibitem{Tamm:1931dda}
  I.~Tamm,
  ``Die verallgemeinerten Kugelfunktionen und die Wellenfunktionen eines Elektrons im Felde eines Magnetpoles,''
  Z.\ Phys.\  {\bf 71} (1931) no.3-4,  141; 
  %doi:10.1007/BF01341701
  %%CITATION = doi:10.1007/BF01341701;%%
  %97 citations counted in INSPIRE as of 30 Apr 2019
  T.~T.~Wu and C.~N.~Yang,
  ``Dirac Monopole Without Strings: Monopole Harmonics,''
  Nucl.\ Phys.\ B {\bf 107} (1976) 365.
 % doi:10.1016/0550-3213(76)90143-7
  %%CITATION = doi:10.1016/0550-3213(76)90143-7;%%
  %505 citations counted in INSPIRE as of 30 Apr 2019
 
 %\cite{Bachas:2009ve}
\bibitem{Bachas:2009ve}
  C.~Bachas, C.~Bunster and M.~Henneaux,
  ``Dynamics of Charged Events,''
  Phys.\ Rev.\ Lett.\  {\bf 103} (2009) 091602
 % doi:10.1103/PhysRevLett.103.091602
  \href{https://arxiv.org/abs/0906.4048}{[arXiv:0906.4048 [hep-th]]}.
  %%CITATION = doi:10.1103/PhysRevLett.103.091602;%%
  %5 citations counted in INSPIRE as of 30 Apr 2019
  
 %\cite{Maldacena:1998bw}
\bibitem{Maldacena:1998bw}
  J.~M.~Maldacena and A.~Strominger,
  ``AdS(3) black holes and a stringy exclusion principle,''
  JHEP {\bf 9812} (1998) 005
  %doi:10.1088/1126-6708/1998/12/005
  \href{https://arxiv.org/abs/hep-th/9804085}{[hep-th/9804085]}.
  %%CITATION = doi:10.1088/1126-6708/1998/12/005;%%
  %650 citations counted in INSPIRE as of 27 Apr 2019
 
%
\bibitem{DiFrancesco:1997nk}
  P.~Di Francesco, P.~Mathieu and D.~Senechal,
  ``Conformal Field Theory,'' Springer 1997.
  %doi:10.1007/978-1-4612-2256-9
  %%CITATION = doi:10.1007/978-1-4612-2256-9;%%
  %164 citations counted in INSPIRE as of 01 May 2019
   
     %\cite{McGreevy:2000cw}
\bibitem{McGreevy:2000cw}
  J.~McGreevy, L.~Susskind and N.~Toumbas,
  ``Invasion of the giant gravitons from Anti-de Sitter space,''
  JHEP {\bf 0006} (2000) 008
  %doi:10.1088/1126-6708/2000/06/008
  \href{https://arxiv.org/abs/hep-th/0003075}{[hep-th/0003075]}.
  %%CITATION = doi:10.1088/1126-6708/2000/06/008;%%
  %526 citations counted in INSPIRE as of 01 May 2019 
  
%\cite{Louis:2014gxa}
\bibitem{Louis:2014gxa} 
  J.~Louis and H.~Triendl,
  ``Maximally supersymmetric AdS$_{4}$ vacua in N = 4 supergravity,''
  JHEP {\bf 1410}, 007 (2014)
 % doi:10.1007/JHEP10(2014)007
  \href{https://arxiv.org/abs/1406.3363}{[arXiv:1406.3363 [hep-th]]}.
  %%CITATION = doi:10.1007/JHEP10(2014)007;%%
  %20 citations counted in INSPIRE as of 01 May 2019       
                 
 
\end{thebibliography}
\end{document}